\documentclass{article}

\usepackage{microtype}
\usepackage{graphicx}
\usepackage{wrapfig}
\usepackage{booktabs} 
\usepackage{comment}
\usepackage{hyperref}
\usepackage{enumitem}
\usepackage{subcaption}

\usepackage{amsmath}
\usepackage{amssymb}
\usepackage{mathtools}
\usepackage{amsthm}
\usepackage{mdframed}

\newtheoremstyle{experimentstyle}%
  {5pt}
  {5pt}
  {}
  {}
  {\bfseries}
  {.}
  {0.5em}
  {}

\theoremstyle{experimentstyle}
\newmdtheoremenv[
  backgroundcolor=gray!10,
  linecolor=gray!40,
  innertopmargin=5pt,
  innerbottommargin=5pt,
  skipabove=10pt,
  skipbelow=10pt
]{experiment}{Experiment}

\usepackage[capitalize,noabbrev]{cleveref}

\newif\ifshowcomments
\showcommentstrue  

\newcommand{\todo}[1]{\ifshowcomments{\color{red}[TODO: #1]}\fi}  

\newcommand{\jacob}[1]{\ifshowcomments\textcolor{purple}{[jacob: #1]}\fi}  %

\usepackage{algorithm}
\usepackage{algpseudocode}

\newcommand{\STATE}{\State}
\newcommand{\REQUIRE}{\Require}
\newcommand{\RETURN}{\Return}
\newcommand{\COMMENT}{\Comment}
\newcommand{\FOR}{\For}
\newcommand{\ENDFOR}{\EndFor}
\newcommand{\IF}{\If}
\newcommand{\ENDIF}{\EndIf}

\usepackage{orcidlink}

\theoremstyle{plain}
\newtheorem{theorem}{Theorem}[section]
\newtheorem{proposition}[theorem]{Proposition}
\newtheorem{lemma}[theorem]{Lemma}
\newtheorem{corollary}[theorem]{Corollary}
\theoremstyle{definition}
\newtheorem{definition}[theorem]{Definition}

\theoremstyle{remark}

\usepackage{tikz}
\usetikzlibrary{shapes.geometric, arrows}

\tikzstyle{startstop} = [rectangle, rounded corners, minimum width=3cm, minimum height=1cm,text centered, draw=black, fill=red!30]
\tikzstyle{process} = [rectangle, minimum width=3cm, minimum height=1cm, text centered, draw=black, fill=blue!30]
\tikzstyle{arrow} = [thick,->,>=stealth]
\tikzstyle{data} = [ellipse, minimum width=3cm, minimum height=1cm, text centered, draw=black, fill=green!30]
\tikzstyle{decision} = [diamond, minimum width=3cm, minimum height=1cm, text centered, draw=black, fill=yellow!30]

\newcommand{\ex}[2]{{\ifx&#1& \mathbb{E} \else \underset{#1}{\mathbb{E}} \fi \left[#2\right]}}
\newcommand{\pr}[2]{{\ifx&#1& \mathbb{P} \else \underset{#1}{\mathbb{P}} \fi \left[#2\right]}}
\newcommand{\var}[2]{{\ifx&#1& \mathrm{Var} \else \underset{#1}{\mathrm{Var}} \fi \left(#2\right)}}

\renewcommand{\epsilon}{\varepsilon}  

\newcommand{\GLSFull}{Token-DiFR}  
\newcommand{\CGSFull}{Token-IPT-DiFR}  

\newcommand{\FSSL}{FSSL}

\newcommand{\GLS}{FSSL-GM}  
\newcommand{\CGS}{FSSL-IPT}  

\newcommand{\GLSestimator}{Token-DiFR}  
\newcommand{\CGSestimator}{Token-IPT-DiFR}  

\def\eprint{eprint}
\def\sub{sub}
\def\format{eprint}

\newcommand{\calP}{\mathcal{P}}

\newcommand{\Toks}{\mathcal{T}}
\newcommand{\dist}{\mathcal{D}}
\newcommand{\A}{\mathcal{A}}
\newcommand{\warden}{\mathcal{W}}
\newcommand{\N}{\mathbb{N}}

\newcommand{\bits}{\{0,1\}}

\newcommand{\getsr}{\overset{{\scriptscriptstyle\$}}{\leftarrow}}

\def\secpar{\lambda}
\newcommand{\negl}{\mathrm{negl}}

\newcommand{\seed}{\mathsf{seed}}

\newcommand{\Steg}{\mathsf{Steg}}
\newcommand{\Enc}{\mathsf{Enc}}
\newcommand{\Dec}{\mathsf{Dec}}
\newcommand{\st}{\mathsf{st}}
\newcommand{\msglen}{\ell_{\mathrm{msg}}}
\newcommand{\coverlen}{\ell_{\mathrm{cov}}}
\newcommand{\KL}{D_{\mathrm{KL}}}

\newcommand{\minH}{H_{\mathrm{min}}}

\newcommand{\supp}{X}

\newcommand{\Ex}{\mathbb{E}}

\newcommand{\GExfil}{\mathbf{Exfil}}
\newcommand{\GExfilR}{\mathbf{Exfil}^{\mathbf{r}}}
\newcommand{\calS}{\mathcal{S}}
\newcommand{\hT}{\hat{T}}
\newcommand{\RO}{\mathsf{H}}

\newcommand{\shade}[1]{\colorbox{gray!30}{#1}}

\newcommand{\wFSSL}{\warden_{\mathrm{FSSL},\tau}}

\newcommand{\context}{\mathcal{H}}

\newcommand{\ot}[1]{}

\usepackage[english]{babel}

\usepackage[letterpaper,top=2cm,bottom=2cm,left=3cm,right=3cm,marginparwidth=1.75cm]{geometry}

\usepackage{amsmath}
\usepackage{graphicx}
\usepackage{array}    
\usepackage{multirow} 

\title{Verifying LLM Inference to Detect Model Weight Exfiltration 
}
\ifx\format\eprint
\author{%
  Roy Rinberg$^{1,2}$\thanks{Correspondence to \texttt{royrinberg@g.harvard.edu}} \\
  \texttt{royrinberg@g.harvard.edu} \\
\And
  Adam Karvonen$^{2}$ \\
\And
  Alexander Hoover$^{3}$ \\
\And
  Daniel Reuter$^{2}$ \\
\And
  Keri Warr$^{4}$ \\
  \\[2mm]
  \small
  $^{1}$Harvard University \quad
  $^{2}$ML Alignment \& Theory Scholars (MATS) \\
  $^{3}$Stevens Institute of Technology \quad
  $^{4}$Anthropic
}
\fi

\author{
    Roy Rinberg$^{1,2}$\thanks{\textit{Corresponding author}: \texttt{royrinberg@g.harvard.edu}} 
    \and
    Adam Karvonen $^{2}$ 
    \and
    Alexander Hoover $^{3}$ 
    \and
    Daniel Reuter $^{2}$ 
    \and
    Keri Warr$^{4}$
      \\[2mm]
  \small
  $^{1}$Harvard University \quad
  $^{2}$ML Alignment \& Theory Scholars (MATS) \\
\small  $^{3}$Stevens Institute of Technology \quad
  $^{4}$Anthropic
}

\date{}

\ifx\format\sub
\fi

\begin{document}
\maketitle

\begin{abstract}

As large AI models become increasingly valuable assets, the risk of model weight exfiltration from inference servers grows accordingly.
An attacker controlling an inference server may exfiltrate model weights by hiding them within ordinary model responses, a strategy known as steganography.
This work investigates how to \textit{verify} LLM model inference to defend against such attacks and, more broadly, to detect anomalous or buggy behavior during inference.
We formalize model weight exfiltration as a security game, propose a verification framework that can provably mitigate steganographic exfiltration, and specify the trust assumptions associated with our scheme.
To enable verification, we characterize valid sources of non-determinism in large language model inference and introduce two practical estimators for them.
We evaluate our detection framework on several open-weight models ranging from 3B to 30B parameters. 
On MOE-Qwen-30B, our detector reduces exfiltratable information to $<0.5\%$ with false-positive rate of $<10^{-2}$, corresponding to a $>200\times$ slowdown for adversaries.
Overall, this work further establishes a foundation for defending against model weight exfiltration and demonstrates that strong protection can be achieved with minimal additional cost to inference providers.
Our code is made available at: \href{https://github.com/RoyRin/inference_verification_for_model_weight_exfiltration}{github.com/RoyRin/inference\_verification\_for\_model\_weight\_exfiltration}.
\end{abstract}

\section{Introduction}

As Large Language Models (LLMs) become increasingly valuable assets, the need for robust defenses to protect them grows in parallel. Inference servers are of particular concern as they host the model weights themselves, and a single compromised server can cause severe damage: either through the \emph{exfiltration} of proprietary weights or by returning \emph{undesirable} or manipulated outputs during inference.

Malicious adversaries can exploit many potential channels, from infiltrating development environments and stealing backups to abusing side channels. While conventional security mechanisms can mitigate some of these risks, defending the inference pipeline poses a unique challenge. To support normal operation, inference providers must allow vast volumes of data to flow in and out of their data centers, often hundreds of gigabytes per day \cite{nevo2024securing}. This open communication channel creates an unavoidable tension between accessibility and security.


For attackers, as the amount of user data increases, so does the attractiveness of exfiltrating information through user data through \emph{steganography}, the practice of concealing information within seemingly innocuous data, such as embedding hidden payloads in generated text. In this work, we propose a scheme to \emph{monitor inference traffic} to detect potential exfiltration through steganography. 
Figure~\ref{fig:desired-vs-undesired-commmunication} depicts two contrasting scenarios: normal inference and compromised inference with exfiltration. 


\begin{figure}[h]
\centering
\subfloat[\textit{Normal inference flow}]{
    \includegraphics[width=0.48\linewidth]{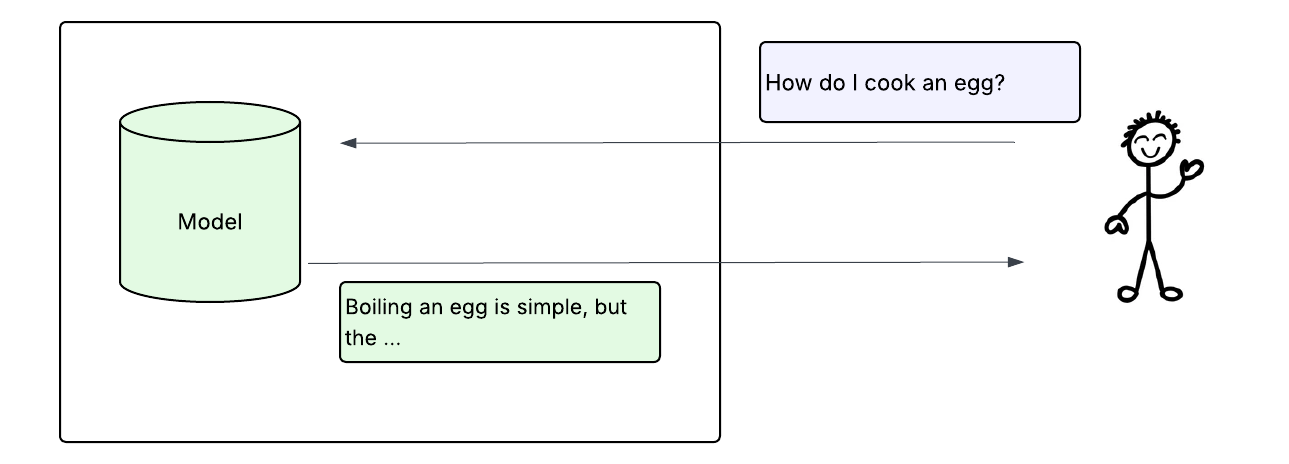}
}
\hfill
\subfloat[\textit{Compromised inference with exfiltration}]{
    \includegraphics[width=0.48\linewidth]{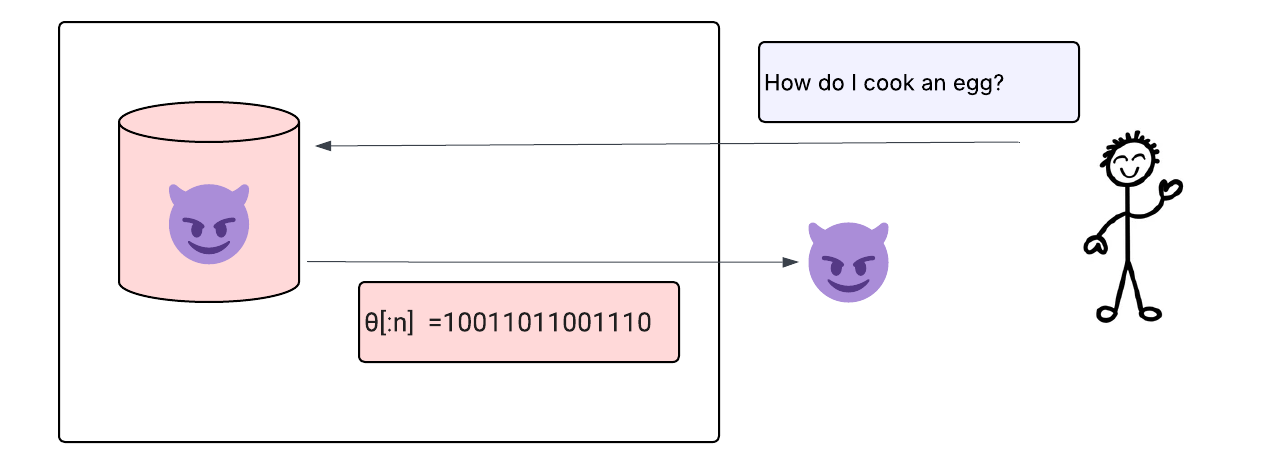}
}
\caption{
Comparison of normal vs.\ compromised inference flows. 
If the inference server is compromised model weights can be exfiltrated via inference responses. 
}
\label{fig:desired-vs-undesired-commmunication}
\end{figure}

Our work rests on a simple but powerful observation: LLM inference traffic follows a well-defined generative process, and any steganographic attempt to encode hidden information must, by necessity, \emph{deviate from expected model behavior}. 
The primary challenge is that LLMs naturally exhibit non-determinism; even identical runs of the same model and code can yield different outputs due to stochastic token prediction and floating-point variance. In this work, we provide methods for quantifying the likelihood of different tokens under this \textit{valid} non-determinism, as a way to prevent model exfiltration.\footnote{It's worth noting that our scheme protects against an adversary with access to a malicious inference stack, but does not provide defense against attacks where a model reveals its own weights as \textit{intended behavior} of the model.}



In this work, we introduce two methods for estimating the likelihood that each generated token was honestly produced by the claimed model (\GLSestimator~and \CGSestimator), which are concurrently developed by \cite{tokendifr}.
This enables us to measure deviations from expected model behavior and derive \emph{information-theoretic bounds} on the amount of information that even a capable adversary could exfiltrate through normal inference traffic. 
We propose a detection scheme that can be integrated into larger model security systems. 

\paragraph{Framing this Work in the Space of Preventing Model Weight Exfiltration}

Beyond technical design, we offer a framework for reasoning about model-weight exfiltration from both security-engineering and economic perspectives. 
Intuitively, exfiltration attempts that move large amounts of data quickly are easier to detect and respond to, whereas extremely slow leaks may be economically irrelevant to either party. 
We highlight three regimes:

\begin{enumerate}[noitemsep]
    \item \textbf{Setting 1: Fast exfiltration:} An extremely rapid leak that occur before human intervention is possible. For example, if an adversary compromises many endpoints and exfiltrates data from all of them simultaneously (e.g., 10{,}000 servers each leaking a 1/10{,}000th fraction of the weights).
    \item \textbf{Setting 2: Intermediate exfiltration:} A leak occurs over days or weeks, allowing incident-response teams to detect and potentially intervene.
    \item \textbf{Setting 3: Slow exfiltration:} Gradual leakage over months or years, which may render the stolen model no longer valuable by the time it is fully reconstructed.
\end{enumerate}

Fundamentally, by creating a scheme (any scheme) for detecting malicious behavior, this forces the adversary to make a choice: they can either get caught by the detection scheme, and deal with the consequences of being detected; or they must operate a covert or ``undetectable'' level of exfiltration.

This work focuses on this ``undetectable'' regime and designs detection tools
that provably limit its effectiveness. Based on our estimates, our proposed techniques
can extend effective response windows from days to months or years,
providing at least a $100\times$ reduction in feasible exfiltration rates.

\paragraph{Contributions.}
\textit{Inference verification} is the task of confirming that a given text output was produced by a specific model running correct inference, rather than by a tampered, substituted, or otherwise divergent process. In this work, we show that inference verification can be applied to the specific threat of steganographic weight exfiltration. We introduce a specific method of inference verification and show that it can be used to bound the capacity available to an adversarial provider trying to encode information in model outputs. This work makes the following contributions:

\begin{enumerate}[noitemsep]
      \item \textbf{Formalizes model-weight exfiltration as a security game} and derives strong bounds on exfiltration capacity when using the systems we propose.
      We formalize the exfiltration-detection interaction in Experiment~\ref{exp:exfaudit} and prove three bounds under varying adversarial models (Theorems~\ref{thm:perplexity-bound}, \ref{thm:no-seed-interactive-bound}, and~\ref{thm:no-seed-bound}). Along with these adversarial models, we propose architecture which can realize and limit potential adversaries in practice.

      \item \textbf{Introduces a verification framework} with multiple scoring methods that reconstruct fixed-seed sampling to detect deviations from expected model behavior.

      We propose the Fixed-Seed Sampling Likelihood (FSSL) as a method for estimating likelihood of being sampled, and we introduce two estimates for this, \GLSFull~(\GLS), and \CGSFull~(\CGS), which can be used for LLMs that sample with the Gumbel-Max trick and Inverse Probability Transforms, respectively.
      

      \item \textbf{Proposes a decoupled verification server architecture} that separates security from performance, enabling high-assurance computation on isolated hardware/software stacks without constraining inference throughput.
      The architecture requires minimal trust assumptions (detailed in Appendix~\ref{sec:trust-assumptions}) and has dual-use value for both security monitoring and reliability/bug detection.

      \item \textbf{Validates empirically on production systems} We evaluate exfiltration rates on Llama-3.1-8B, Llama-3.2-3B, Qwen1.5-MoE-7B, Qwen3-30B-A3B. For Qwen3-30B-A3B we demonstrate a ${<}0.5\%$ exfiltrable information rate with $10^{-2}\%$ FPR,  a ${\ge}200\times$ slowdown in attack. For a 1\,TB model with $1000$Gb of traffic per day, this extends exfiltration time from ${\sim}1$ day to ${>}200$ days.
  \end{enumerate}

Importantly, as our scheme detects deviations from expected behavior, it will flag not only malicious exfiltration attempts but also inference-time anomalies or bugs. As such, a further benefit of this detection scheme is that it serves a double-purpose and can provide a powerful tool for debugging and reliability monitoring.
From an infrastructure standpoint, the systems required for detecting malicious activity and those for detecting model failures are nearly identical.

\section{Related Works}
There is a lot of work in the space of detecting data exfiltration; in this work we specifically focus on approaches that hide information within a channel containing user data, as this is a novel attack surface area within LLM inference. This ties directly into existing literature on steganography and watermarking, as well as basic forms of content filtering.

\subsection{Steganography and Watermarking}

This section provides a brief overview of ideas and techniques in steganography and watermarking, which are related to the capabilities of the exfiltration adversary in our threat model.

\paragraph{Steganography}
Steganography is the problem of communicating covertly over a channel by embedding hidden information into otherwise innocuous messages. A steganographic scheme consists of an encoding algorithm that maps a secret message and (optionally) a key to a \emph{cover object}, and a decoding algorithm that recovers the hidden message from the resulting output. 

As a concrete example, one may embed a secret message into the least significant bits of an image's pixels. The modified image appears visually unchanged, yet a recipient who knows the embedding rule can recover the hidden message. More generally, modern steganographic schemes are analyzed with respect to an adversary who attempts to distinguish genuine channel outputs from stego outputs.

In this work, our goal is to detect steganographic schemes used to hide model-weight exfiltration attempts. We leverage standard definitions and results from the steganography literature when understanding the limitations of our proposal. For brevity, we defer formal syntax and foundational results to Appendix~\ref{sec:stego}. For other cryptographic primitives, such as pseudorandom functions and hash functions, we refer the reader to Katz and Lindell~\cite{KatLin14}.


\paragraph{Watermarking.}

Watermarking is a related technique often used to encode a single bit of information within a covertext: $1$ if the watermark is present and $0$ if the watermark is not. The embedded information can serve many purposes, such as asserting ownership, tracing provenance, or verifying authenticity.

Conceptually, steganography and watermarking lie on a spectrum, and it may be useful to interpret steganography as repeatedly running a watermarking scheme to encode more than a single bit.
In this light, while our work focuses on the task of steganography (and its detection), the connection to watermark allows us to pull from theoretical bounds in the watermarking literature.

There is a recent but long line of work in language model watermarking~\cite{Aar22,ICML:KGWKMG23,COLT:ChrGunZam24,CIC:FGJMMW23,ARXIV:KTHL23,pseudo_random_error_correcting,golowich2025editdistancerobustwatermarks,SP:CHS25}.  However, most works operate in a model that require the watermarker and detector to share only a small key with the goal of making the watermark as robust as possible. This means that even the highest throughput schemes only embed about $1$ bit per ``block,'' where a block is text that is being watermarked\footnote{The length of a ``block'' depends on the scheme and security parameter. In practice, we can roughly think of this as requiring 128-bits of empirical entropy.}~\cite{SP:CHS25}. 

In contrast, language model steganography~\cite{CCS:KJGR21} has a higher throughput rate, since it assumes text is unchanged. These schemes can almost reach capacity, embedding as many bits as the entropy of the response.




\subsection{Model Weight Exfiltration}

In the context of discussing model-weight exfiltration, a recent framework \cite{nevo2024securing} introduces the notions of security levels from 1 to 5, which increase in defensive capacity. \textbf{SL1} is a system that can likely thwart amateur attempts. \textbf{SL2} is a system that can likely thwart most professional opportunistic efforts by attackers that execute moderate-effort or nontargeted attacks.
\textbf{SL3} is a system that can likely thwart cybercrime syndicates or insider threats. \textbf{SL4} is system that can likely thwart most standard operations by leading cyber-capable institutions. \textbf{SL5} is a system that could plausibly be claimed to thwart most top-priority operations by the top cyber-capable institutions.

To set the grounding for this work, we will focus primarily on creating a scheme for moving from SL3 to SL4 and building the basis for SL5 security, by increasing the time required to to covertly exfiltrate from minutes and hours to days and years.

\subsection{Inference Verification}

As introduced above, inference verification is the task of confirming that text outputs were produced by a specific model running correct inference. Existing approaches span a wide range of trust assumptions and computational costs, from cryptographic guarantees to statistical tests.

\textbf{Cryptographic methods.} Zero-knowledge proofs (ZKPs) offer the strongest guarantees, allowing a verifier to confirm that a computation was performed correctly without re-executing it. Recent work has applied ZKPs to neural network inference \cite{chen2025zktorch, sun2024zkllm}, but the computational overhead remains prohibitive: proof generation for LLaMA-2-scale models requires on the order of hundreds to thousands of seconds per token, making these approaches impractical for production deployments.

\textbf{Internal state fingerprinting.} A more tractable class of methods verifies inference by comparing intermediate activations between a provider's run and a verifier's replay. LOGIC \cite{singh2025logic} collects top-$k$ log probabilities at sampled positions, while TOPLOC \cite{ong2025toploclocalitysensitivehashing} captures top-$k$ activation indices and values from the final hidden layer, encoding them as a polynomial to reduce communication costs. Both methods work by replay: the verifier re-executes inference on the provider's output sequence and checks that the resulting activations match. However, internal state fingerprinting only confirms that a sequence is \textit{consistent} with a given forward pass, not that it was actually produced by the model's sampling procedure. A malicious provider could generate arbitrary text, post-hoc compute matching fingerprints, and pass verification, making these methods proofs of reconstructability rather than generation authenticity.

\textbf{Distributional tests.} \cite{gao2025modelequalitytestingmodel} frame inference verification as a two-sample test between a trusted distribution and an API distribution, using Maximum Mean Discrepancy with string kernels. \cite{zhu2025auditingblackboxllmapis} propose a rank-based uniformity test that checks whether target API outputs are statistically consistent with a reference model's distribution over many samples. These approaches are fully black-box but verify only distributional similarity, not per-generation correctness.

\textbf{Shared seed verification.} Concurrently with this work, Karvonen et al.\ introduce \GLSFull~ \cite{tokendifr}. \GLS~compares output tokens directly against a trusted reference implementation conditioned on the same seed and is able to detect 4-bit quantization with AUC $> 0.999$ within 300 output tokens at zero communication overhead. Karvonen et al. also introduce Activation-DiFR which complements this by projecting intermediate activations through random orthogonal projections into compact fingerprints, achieving the same detection accuracy with as few as 2 tokens while reducing communication costs by 25--75\% relative to prior methods.

\subsection{Existing Model Weight Exfiltration Detection methods}

One naive attack to exfiltrate data through user traffic would be to exfiltrate weights directly through output messages, by sending the model weights as \textit{plaintext}. A first line of defense against such a naive weight exfiltration technique is to detect messages that deviate from ordinary conversational patterns. Two simple, independent approaches can be implemented here.

\textbf{Syntactic filtering.} A basic scheme applies regular expressions and lightweight heuristics to flag messages with suspicious surface patterns, such as long sequences of numbers or other structured encodings. While this approach is trivial to deploy it is also easily evaded by re-encoding information (e.g., mapping numbers to words or distributing bits across multiple tokens).

\textbf{Perplexity filtering.} A slightly more sophisticated defense uses a reference LLM trained on benign message traffic to compute sequence-level \emph{perplexity}, a standard measure of how well a model predicts a sequence. Perplexity is a function of the log-likelihood of each token, where lower values indicate that the model finds the text highly predictable, while higher values suggest the text is surprising or out-of-distribution. Messages with unusually high perplexity are flagged as potential exfiltration attempts \cite{perplexityfilter}. This can detect outputs that are syntactically normal but statistically implausible under the benign distribution, though adaptive adversaries can still craft messages that maintain natural perplexity while encoding information.

Both schemes are simple to run and provide a low-cost first defense, but they leave a substantial gap between benign and adversarial outputs. A powerful adversary can easily create a steganographic technique that only modifies words and tokens that leave the perplexity of the whole sequence low.

\section{LLM Inference Preliminaries}

Because our exfiltration attacks operate by modifying the inference procedure, we begin by reviewing how modern LLMs generate text. In standard implementations, inference is designed to be stochastic: given a fixed model and a fixed random seed, sampling-based decoding should be reproducible.

At the algorithmic level, this suggests a clean separation between deterministic computation (forward passes, logit computation) and controlled randomness (token sampling that takes a seed). However, in practice, hardware and software effects introduce additional sources of nondeterminism; and as a result, even with a fixed seed, repeated runs of inference often do not produce identical outputs.

In this section, we formalize which components of inference are intended to be deterministic, which are intentionally randomized, and which remain nondeterministic despite the use of a seed.

\subsection{Basic Setting (LLM inference)}
\paragraph{Decode versus prefill.}

Transformer inference consists of two phases: \emph{prefill} processes the input prompt in a single parallel forward pass, while \emph{decode} autoregressively generates output tokens one at a time. Though both phases have similar asymptotic complexity, prefill achieves significantly higher throughput in practice (typically 3–10×) because it is compute-bound and highly parallelizable, whereas decode is memory-bandwidth bound due to repeatedly loading the key–value cache~\cite{erdil2025inferenceeconomicslanguagemodels}. This makes prefill particularly advantageous for verification or evaluation tasks that do not require autoregressive generation.

\subsection{Sampling From Large Language Models}
Sampling from a Large Language Model can be thought of in two steps: (1) generating a probability distribution, and (2) sampling a token from it. For step (1), given logits $z \in \mathbb{R}^{|\Toks|}$ over a vocabulary $\Toks,$ and a temperature parameter $T$, the model induces a categorical distribution
\[
    p_i^{(T)}  = \mathrm{softmax}\!\left(\frac{z}{T}\right)_i
    = \frac{\exp(z_i / T)}{\sum_{j=1}^{|\Toks|} \exp(z_j / T)}.
\]
The temperature $T$ controls the sharpness of the distribution: as $T \to 0$, the distribution concentrates on the maximum-logit token (approaching greedy decoding), while larger $T$ produces a flatter distribution and increases randomness.
After optional filtering steps such as top-$k$, top-$p$, or min-$p$, the model samples from the resulting categorical distribution. To sample a token from a probability distribution, many modern LLMs use either the \emph{Inverse Probability Transform} (e.g., Ollama~\cite{ollama_docs}, HuggingFace~\cite{huggingface}) or the \emph{Gumbel-Max Trick} (e.g., vLLM~\cite{vllm}).
This sampling process relies on a pseudorandom generator, which takes as an input a sampling seed $\sigma$ and counter $i$ as an input, and returns a value $x$.

\paragraph{Inverse Probability Transform.}
From a fixed random seed $\sigma$, a (pseudo)random number $u \sim \mathrm{Uniform}(0,1)$ is drawn for each token. The algorithm then computes the cumulative sum of the token probabilities $\vec{p}$ and selects the first token whose cumulative mass exceeds $u$. Conceptually, this corresponds to inverting the cumulative distribution function (CDF) of $\vec{p}$. The method is exact, efficient, and deterministic when both $\vec{p}$ and $\sigma$ are fixed; repeated sampling under identical conditions will always yield the same token. Here we use the notation $F(\cdot ; r)$ to denote running a randomized
function $F$ with randomness $r$, making it (ideally) deterministic. We also
use a deterministic hash function $\RO$ to expand $\seed$ into reproducible pseudorandomness.\footnote{In practice,
we may only need a pseudorandom function, since our security theorems
only requires that the output is pseudorandom to an adversary.}

\begin{algorithm}[H]
\caption{Inverse-Probability-Transform (IPT) Sampler}
\label{alg:ipt}
\begin{algorithmic}[1]
\REQUIRE Seed $\seed$, context position $i$,
probability vector $\mathbf{p}\in[0,1]^{|\Toks|}$ with $\sum_{t=1}^{|\Toks|} p_t=1$
\STATE $\mathbf{C} \gets \textsc{CumulativeSum}(\mathbf{p})$ \COMMENT{Convert PDF $\mathbf{p}$ to CDF $\mathbf{C}$}
\STATE $u \gets \textsc{Uniform}([0;1]; \RO(\seed\| i))$
\COMMENT{Use token index $i$}
\STATE \RETURN $\min\{t\in\{1,\dots,|\Toks|\}: C_t>u\}$\COMMENT{binary search: find smallest $t$ s.t. $C_t>u$}
\end{algorithmic}
\end{algorithm}

\paragraph{Gumbel-Max Trick.}
The Gumbel-Max Trick provides an equivalent but reparameterized way to sample from $\vec{p}$. From a fixed random seed $\sigma$, for the $t$th token, an independent Gumbel noise term $g_t \sim \mathrm{Gumbel}(0,1)$, determined by $\seed,$ is added to the (log) probability:
\(
    z_t = \log p_t + g_t.
\)
The selected token is the index of the maximum perturbed score,
$\arg\max_t z_t$.
Intuitively, this converts sampling from a categorical distribution into a deterministic \textit{argmax} over noisy logits. Like the inverse transform, it is reproducible under a fixed seed $\seed$, since the Gumbel variables $g_t$ are deterministically derived from $\seed$. This trick is particularly convenient for differentiable relaxations (e.g., Gumbel-Softmax) and for implementations optimized for vectorized or GPU-based sampling.

\begin{algorithm}[H]
\caption{Gumbel-Max Sampler}
\label{alg:gumbelmax-simple}
\begin{algorithmic}[1]
\REQUIRE Seed $\seed$, context position $i$, probability vector $\vec{p}\in[0,1]^{|\Toks|}$ (over tokens $\Toks$, $\sum_t p_t=1$)
\STATE $\vec{z} \gets \textsc{Gumbel}([0,1]^{|\Toks|}; \RO(\seed\| i))$
\COMMENT{each $z_t \sim \textsc{Gumbel}(0,1)$}
\STATE $\hat{t} \gets \arg\max_{t \in \Toks} \big(p_t + z_t\big)$
\COMMENT{$\arg\max$ returns the \emph{token} $\hat{t}$ achieving the maximum sum}
\STATE \RETURN $\hat{t}$
\end{algorithmic}
\end{algorithm}

\paragraph{Sampling Multiple Tokens.}
We use the notation $\dist^{\RO}_{\theta,\context}(\seed)$ to denote the
honest sampling of a full response for a specific model $\theta,$ context
$\context,$ seed $\seed,$ and hash function (or PRF) $\RO.$ This algorithm
first samples the distribution $\vec{p}$ using $\theta$ and $\context.$
Next, $\dist$ will use one of the two above approaches for selecting a
specific token $t$ from the distribution using the randomness from
$\RO(\seed\|i),$ where $i$ is the current response index.\footnote{For one
of our bounds, we
require that randomness is sampled from $\RO(\seed\|i\|\context)$, which
we call \emph{sequential} randomness expansion.}
Then, $t$ is added to the context $\context$ and
the process repeats until a termination symbol or response length limit is
reached.

\subsection{Non-determinism in Machine Learning} \label{sec:nondeterminism}

Despite the amount of reproducibility of the methods above,
fixing the random seed (even with temperature $T{=}0$) does not guarantee identical outputs. Non-determinism  arises because floating-point arithmetic is non-associative ($(a{+}b){+}c \neq a{+}(b{+}c)$): different GPU kernels or shapes can accumulate tensors in different orders, yielding different values. Batch variance, where kernel computation strategies depend on the batch size, which depends on the concurrent load, is the most common source of numerical noise. Inference optimizations further shift computation order; batch size, attention chunking/prefix caching, and cache reuse, and MoE models all add routing variability that can depend on load and tiny numerical perturbations \cite{he2025nondeterminism}.

Identical open-weight models have shown cross-provider performance discrepancies, consistent with differences in templates, quantization/serving stacks, and the variability sources above \cite{willison_nondeterminism}.


We empirically characterize this non-determinism by querying the OpenAI API with fixed seeds across thousands of completions; the results confirm that divergences are sparse but present, concentrating at a small number of token positions (see Appendix~\ref{sec:base-rates-of-nondeterminism}).

\subsection{Valid Non-determinism with a Fixed Seed: The Fixed-Seed Posterior Distribution} \label{sec:fixed-seed-posterior}

While Inverse-Probability Transform and Gumbel-Max Trick are mathematically deterministic, in practice the logit distribution $\vec{p}_\theta(x)$ is produced by a \emph{non-deterministic} computation. Even with a fixed seed $\sigma$, floating-point imprecision, asynchronous GPU execution, and nondeterministic kernel scheduling can cause small but measurable differences in $\vec{p}_\theta(x)$ across runs. This effect arises even in benign settings, independent of any adversarial manipulation, and has been well-documented in prior work~\cite{yuan2025fp32deathchallengessolutions}. We refer to this inherent variability as \textit{valid non-determinism}.

To model this behavior, we view the model's computed probability vector not as a single deterministic mapping $\vec{p}_\theta(x)$, but as a random draw from a higher-order \emph{distribution over distributions}:
\[
\vec{p}_\theta(x) \sim \mathcal{P}_\theta(x),
\]
where $\mathcal{P}_\theta(x)$ represents the stochastic process generating the token-level probabilities.
This perspective allows us to define heuristic distance measures between two runs of the same model (with identical seed $\sigma$) by comparing samples from $\mathcal{P}_\theta(x)$, illustrated conceptually in Figure~\ref{fig:posterior-distribution-comparison}.

\paragraph{Fixed-Seed Posterior Distribution.}
We begin with a general abstraction that captures the non-determinism induced by stochastic sampling under a fixed seed.
Let $\mathrm{Sample}(\cdot, \sigma)$ denote a generic sampling procedure that maps a probability distribution $p$ over tokens $\mathcal{V}$ and a seed $\sigma$ to a token $t \in \mathcal{V}$:
\[
t = \mathrm{Sample}(p, \sigma).
\]
Given a distribution over probability vectors $p \sim \mathcal{P}_\theta(x)$, the \emph{fixed-seed posterior distribution} is defined as:
\[
P_{\mathrm{seed}}(t \mid \sigma)
    = \Pr_{p \sim \mathcal{P}_\theta(x)} \big[\, \mathrm{Sample}(p, \sigma) = t \,\big].
\]
This represents the probability, over draws of the underlying probability distribution $p$, that a fixed random seed $\sigma$ yields token $t$.
Intuitively, it characterizes how much stochastic variability remains in the sampling outcome even when the seed is fixed, due to non-determinism in the generation of $p$.

\vspace{0.5em}
\noindent
\textbf{Example: Inverse Probability Transform.}
For concreteness, consider the case where $\mathrm{Sample}(\cdot, \sigma)$ implements the inverse probability transform.
Let $\mathcal{P} = \{ P_i \}_{i=1}^n$ denote a family of token-level probability distributions,
where each $P_i$ defines a conditional distribution over tokens $t \in \mathcal{V}$ given model context $x_i$:
\[
P_i(t) = \Pr[T = t \mid X = x_i].
\]
Let $R : \Sigma \to [0,1]$ denote a deterministic randomization function (e.g., a pseudorandom generator) mapping a seed $\sigma \in \Sigma$ to a uniform draw $R(\sigma) \sim \mathrm{Uniform}(0,1)$.
For each $P_i$, define the seed-conditioned token selection function
\[
S_i(\sigma) = \min \{\, t \in \mathcal{V} : F_i(t) \ge R(\sigma) \,\},
\]
where $F_i$ is the cumulative distribution function (CDF) corresponding to $P_i$.

Then, the fixed-seed posterior distribution for the inverse probability transform is:
\[
P_{\mathrm{seed}}(t \mid R(\sigma) = r)
= \Pr_{P_i \sim \mathcal{P}} \big[\, S_i(\sigma) = t \,\big],
\]
representing the probability (over draws of $P_i$) that a fixed seed realization $\sigma$ produces token $t$ under this specific sampling mechanism. This abstraction makes clear that while the implementation details of sampling (e.g., inverse transform vs.\ Gumbel-Max) differ, each induces its own fixed-seed posterior distribution characterizing the residual randomness observable under valid non-determinism.

Intuitively, this distribution captures how much apparent randomness remains even when the seed is fixed. Figures~\ref{fig:probability-given-random-sample} and~\ref{fig:gumble-argmax-flipping-score-diagram} visualize the resulting fixed-seed posterior distributions for the inverse probability transform and the Gumbel-Max mechanism, respectively. These distributions form the basis for reasoning about an honest divergence during verification (as opposed to potentially malicious nondeterminism).

\begin{figure}[ht!]
    \centering
    \begin{subfigure}[t]{0.38\linewidth}
        \centering
        \includegraphics[width=\linewidth]{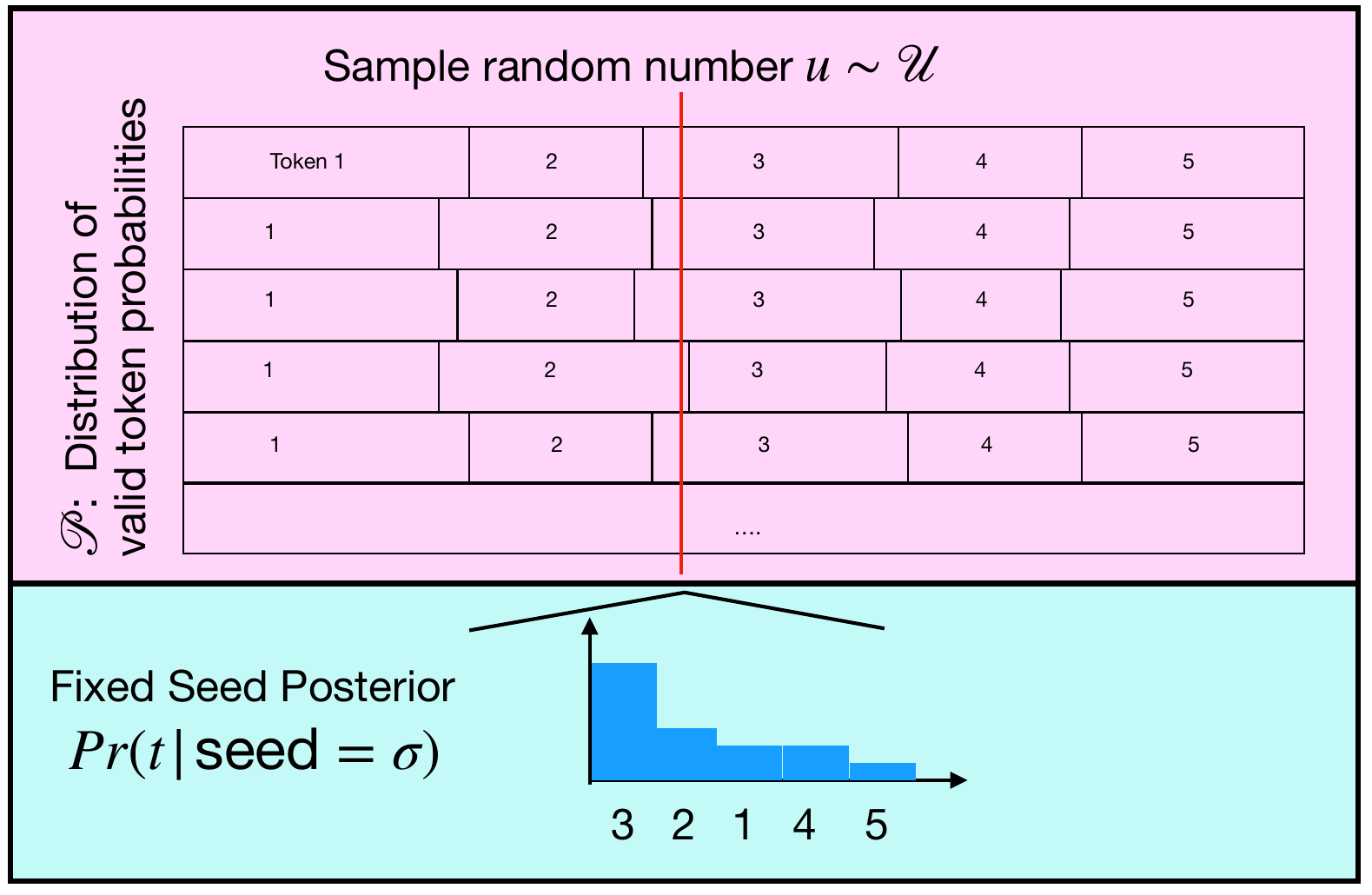}
        \caption{Fixed-seed posterior distribution for the inverse probability transform, incorporating noise in token probabilities.}
        \label{fig:probability-given-random-sample}
    \end{subfigure}
    \hfill
    \begin{subfigure}[t]{0.6\linewidth}
        \centering
        \includegraphics[width=\linewidth]{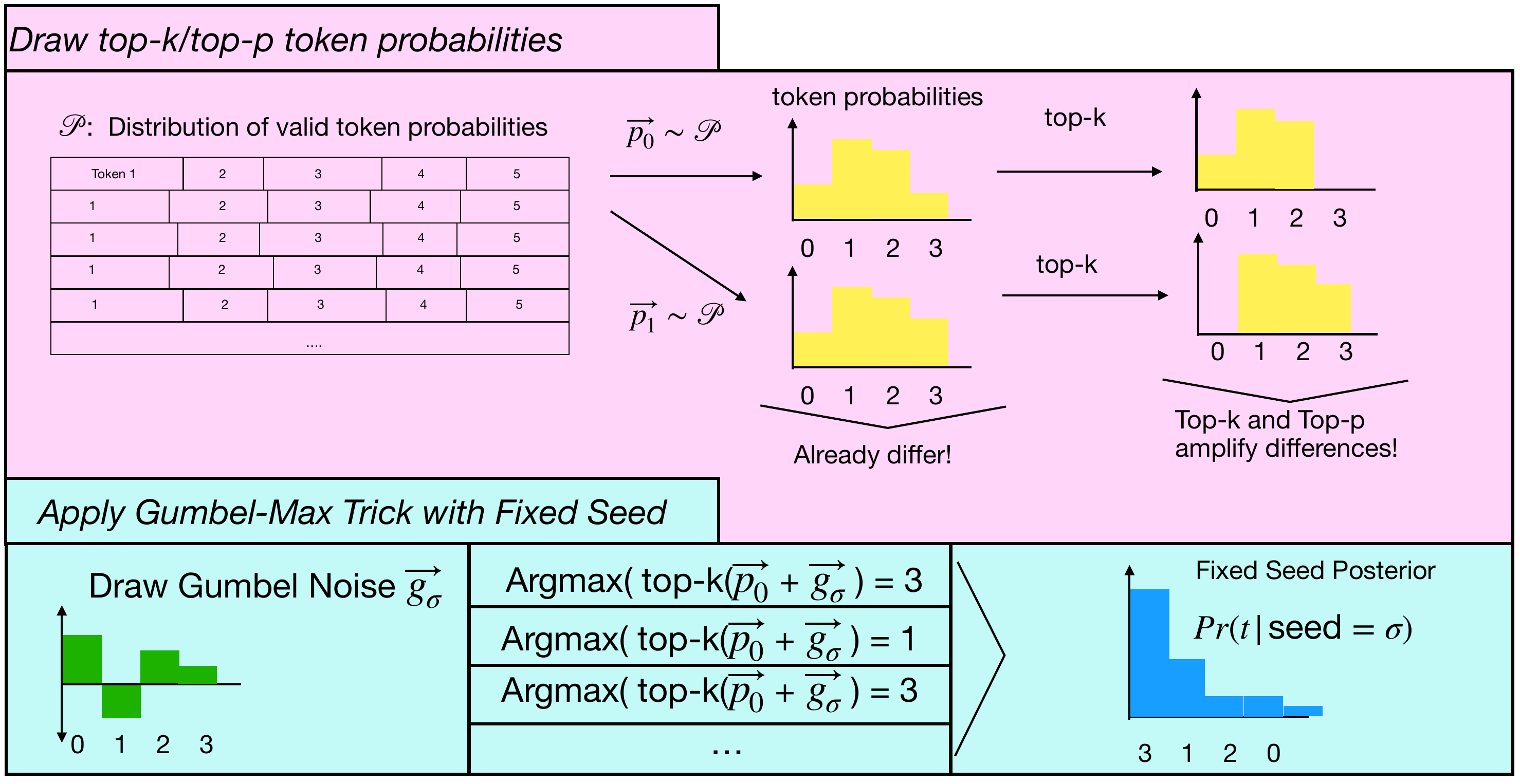}
        \caption{Fixed-seed posterior distribution for the Gumbel-Max mechanism under the same noisy conditions.}
        \label{fig:gumble-argmax-flipping-score-diagram}
    \end{subfigure}
    \caption{Comparison of fixed-seed posterior distributions under two sampling mechanisms, illustrating variability arising from valid non-determinism in $\vec{p}_\theta(x)$. }
    \label{fig:posterior-distribution-comparison}
\end{figure}

Understanding the limits of valid non-determinism is central to limiting the exfiltration rate of adversaries. Under very minimal trust assumptions (Appendix~\ref{sec:trust-assumptions}), the only communication channel open to the adversary is by hiding communication within this honest variability.

\section{Formulating the Model Weight Exfiltration Security Game}\label{sec:problem}

In this section, we provide a concrete specification for the exfiltration problem that we solve in this paper. At a high-level, we consider the setting of LLM inference where there is an inference server and a chat server; we assume the inference server is fully corrupted by an adversary whose goal is to exfiltrate the inference model weights to an outside party, with whom they can coordinate. We assume that the adversary does not have access to any side-channels, eliminating exfiltration techniques such as  using timing channel attacks, uploading the weights to a server through a developer's laptop, or \textit{physically stealing the device holding the model}. In many ways, preventing model theft is not different from a traditional security problem attempting to prevent any kind of exfiltration; there are several main differentiating factors associated with model-weight exfiltration.

\paragraph{\textbf{Advantages to the Defender:}} There are two main advantages to the defender: (i) Model weights are large; (ii)
The type of operations and data that a model handles and outputs is largely predictable, namely: matrix multiplications and text data.

\paragraph{\textbf{Advantages to the Attacker:}} The main advantage to the attacker is that generally users expect to be able to make an extremely large number of calls that access the model (weights) per day.


\subsection{Security Game}

We define a simple extraction game, which we later be used to bound the success probability of an adversary
extracting information, in Section~\ref{sec:info-theory-bounds}.
The formal description is given in Experiment~\ref{exp:exfaudit}, which
is parameterized by:

\begin{itemize}[noitemsep]
    \item $\Toks$, the token set (i.e., vocabulary or alphabet) of
    the adversary's message.
    \item $\RO$, is a hash or pseudorandom function, modeled as a
    random oracle for simplicity.
    \item $\mathcal{A}_0$, $\mathcal{A}_1$, $\mathcal{A}_2$, the adversary's
    preprocessing, encoding, and decoding algorithms respectively.
    \item $\dist_{\theta, \context}(\seed)$, the honest response distribution
    , determined by the model weights $\theta$, context (or prompt)
    $\context$, and seed  $\seed$. In our context, even after all these
    are given, the output is a \emph{randomized} because of ``invalid
    nondeterminism'' (Section~\ref{sec:nondeterminism}). In our security games,
    we suppress the $\theta$ and $\context$ parameters and treat them as
    fixed for the given adversary.
    \item $\warden$ (warden), the policy algorithm to determine whether to flag
    a message.
    \item $n$, the number of bits which the adversary is challenged
    to exfiltrate (in our setting, this is the length of the model weights).
    \item $\ell$, the length of the adversary's message.
    \item $\secpar$, the seed length of the model's randomness and the security parameter of the system.
\end{itemize}
Success in this security game is defined by the adversary successfully
reconstructing the challenge bits from the message (and any other
precomputed material) itself. We define security for computationally bounded
adversaries and policies as well as unbounded algorithms.

In our games, we suppress the notation of the model weights $\theta$ and
the prompt $P$. We abstract the prompt out of the security game because
a practical adversary will have already picked some prompt $P$ on which they
will try and perform the exfiltration. And we care about limiting the
exfiltration rate on this particular prompt. We also choose to model an
adversary trying to exfiltrate some
information about the incompressible model weights $\theta$ as the adversary
sending and recovering a random bit-string. This allows us to further abstract
away the model weights from the security game.

Note this security has been explored in the context of steganography,
which we draw on directly. However, we are not interested in proving
security against all possible policies/wardens. One can view our
results for specific \emph{lower bounds} for steganography in a particular
realistic context.

\begin{figure}[htbp]
    \centering
    \includegraphics[width=.6\linewidth]{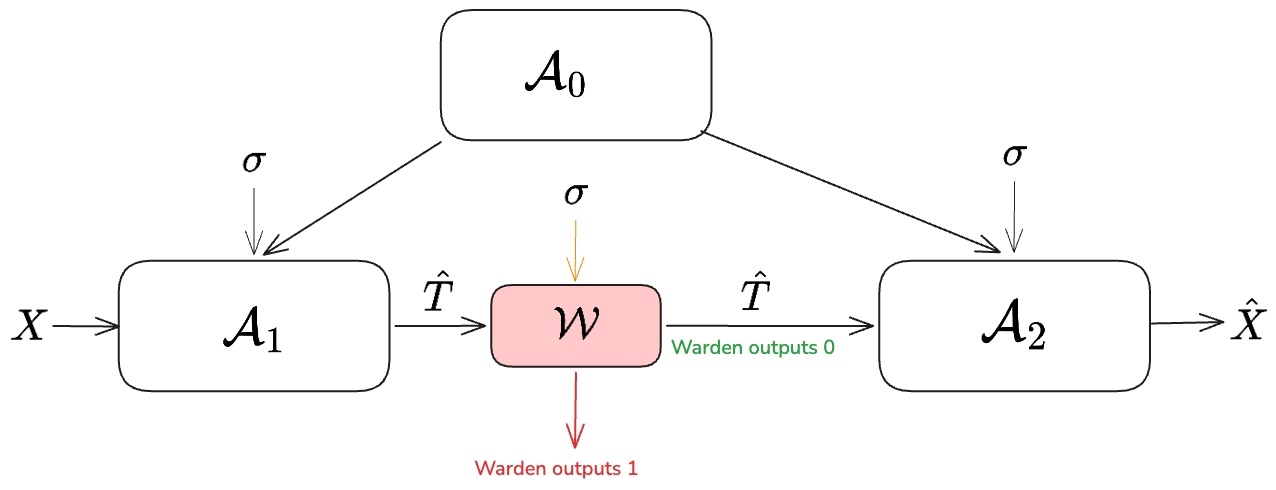}
    \caption{Outline of the security game described in Experiment \ref{exp:exfaudit}.}
    \label{fig:security-game-diagram}
\end{figure}

\begin{experiment}[Exfiltration Experiment]
\label{exp:exfaudit}
\textbf{Parameters.}
Let positive integers $n$ be the challenge bit length, $\ell$ be the
maximum response length, and $\secpar$ be the security parameter.
Let $\Toks$ be the token set (vocabulary) of messages and
$\dist^{\RO}(\sigma)$ be a randomized function from $\bits^{\secpar}$ to
$\Toks^{\le \ell}$. Then, for a policy $\warden$ and adversary
$\A = (\A_0,\A_1,\A_2)$, we define the games:
\\
\noindent\underline{$\GExfil_{\warden,\A,\dist}(\secpar,n,\ell)$}
\quad \shade{\underline{$\GExfilR_{\warden,\A,\dist}(\secpar,n,\ell)$}}
  \begin{enumerate}[noitemsep]
  \item Initialize a random oracle $\RO : \bits^{*} \to \bits$
  \item $\st_1, \st_2 \gets \A_0^{\RO}(1^{\secpar},1^{n},1^{\ell})$
  \quad\quad\quad\texttt{// preprocessing}
  \item $X \gets \bits^n$~~;~~$\sigma \gets \bits^\secpar$
  \quad\quad\texttt{// sample challenge and seed}
  \item $\hT \gets \A_1^{\RO}(\st_1,X,\sigma)$
  \quad\quad\quad\quad\quad\texttt{// generate cover response}
  \item If $\warden^{\RO}(\sigma,\hT) = 1$, then output $\bot$
  \item $\hat{X} \gets \A_2^{\RO}(\st_2,\hT,\sigma)$
  ~~\shade{$\hat{X} \gets \A_2^{\RO}(\st_2,\hT)$}
  \quad\texttt{// recovery}
  \item If $\hat{X} = X$, then output $1$
  \item Else, output $0$.
\end{enumerate}
These definitions are presented in
the \emph{random oracle model}~\cite{CCS:BelRog93}, which enables some of
our later results. Note that all parties know and have access to
$\dist^{\RO}.$

The first game $\GExfil$ is the standard game, where both parties are aware of
the seed that is used to sample the distribution. However, in cases when
the seed $\sigma$ is sampled securely (and $\A_2$ does not get it directly),
we are able to obtain stronger bounds in the more restricted game $\GExfilR$.
\end{experiment}

\begin{definition}[$(\alpha,\varepsilon)$-exfiltration secure]
A policy $\warden$  is (computationally)
\emph{$(\alpha,\varepsilon)$-exfiltration secure} if:
\begin{enumerate}[noitemsep]
  \item  For all $\sigma$,
  \(\Pr_{T \sim \dist^{\RO}(\sigma)}\big[\warden^{\RO}(\seed,T)= 1 \big] \le \alpha,\)
  i.e., $\warden$ has false positive rate $\alpha$,

  \item and for all (PPT) adversaries $\A$,
  \(
      \Pr\big[\GExfil_{\warden,\A,\dist}(\secpar,n,\ell) = 1\big] - 2^{-n} \le \varepsilon.
  \)
\end{enumerate}
Note that $\alpha$ and $\varepsilon$ are, in general, functions of the
parameters $\secpar$, $n$, $\ell$, or even $\dist$. If the above holds only
for $\GExfilR$ rather than $\GExfil$, then we say it $\warden$ is secure
\emph{against seed agnostic decoders}.

Where relevant, we will also consider adversaries with restricted access to
the random oracle $\RO$.


\end{definition}

\subsection{Threat Model: Baseline Trust Assumptions}\label{sec:adversarial-model}

We consider an idealized setting of the model-weight exfiltration problem.
In particular, we do not attempt to defend against adversaries with unrestricted physical access to hardware (e.g., physically stealing servers or storage devices to extract model weights).
Our focus is on protecting model parameters in environments where security is enforced through software, networking, and operational controls.
A more detailed discussion of these assumptions is provided in Appendix~\ref{sec:trust-assumptions}.
\vspace{-3mm}
\paragraph{No Side Channels.} We assume that adversaries cannot exploit any side channels (e.g., timing, cache access, or power analysis) to extract information. The only interface between the inference server and the external environment is the API server, which mediates all communication.
\vspace{-3mm}
\paragraph{External Attacker Control over User Queries.} We assume that external adversaries may fully control the contents of user queries, including model inputs, random seeds, sampling parameters (e.g., temperature, top-$k$, and top-$p$), and query rates. This represents a worst-case scenario for query-based extraction attempts.
\vspace{-3mm}
\paragraph{Internal Attacker Control over Outbound Messages.} We assume that an internal adversary (e.g., one with access to the inference server) may control all outbound inference messages. That is, they can arbitrarily modify or encode information within model outputs before they are transmitted externally.
\vspace{-3mm}
\paragraph{Trusted Verification Code.} We assume that the verification server is honest and uncompromised. While the attacker may control the client side of the API or manipulate request patterns, the verification process itself is assumed to execute faithfully and does not introduce additional exfiltration risk.
\vspace{-3mm}
\paragraph{Trusted Logging.} We assume access to a secure, append-only logging mechanism that records every input, output, and associated seed for all messages transmitted from the inference server to the API gateway. These logs are assumed to be tamper-evident and verifiable.
\vspace{-3mm}
\paragraph{Random Sampling for Verification is Private and Non-manipulable.} We assume the existence of a private, non-manipulable process for randomly sampling input–output pairs for verification. If this process were predictable or observable, an adversary could infer which pairs will be re-verified and exfiltrate information through the remaining unchecked traffic. The base assumption is that an adversary can neither predict the next message to be sampled, nor modify the probability that a message is sampled next.
\vspace{-3mm}
\paragraph{Paired Input–Output Messages.} We assume that every output message produced by the inference server is associated with a corresponding input message. Messages that cannot be paired with a valid input are not permitted to leave the datacenter.

\section{Theoretical Framework: Verification System}\label{sec:verification-system}

This section introduces, at an abstract level, our proposed \emph{verification system} to mitigate against potential exfiltration.
We defer implementation details and additional system-level choices to Section~\ref{sec:implement-verif-server}.
The goal of the verification system is to detect suspicious deviations from a trusted baseline model during inference. While our primary motivation is to identify potential model-weight exfiltration, the same framework applies to general distributional or behavioral anomalies in model outputs.

At a high level, the system operates as follows (illustrated in Figure~\ref{fig:headline-scheme}):

\begin{figure}[h]
    \centering
    \includegraphics[width=0.8\textwidth]{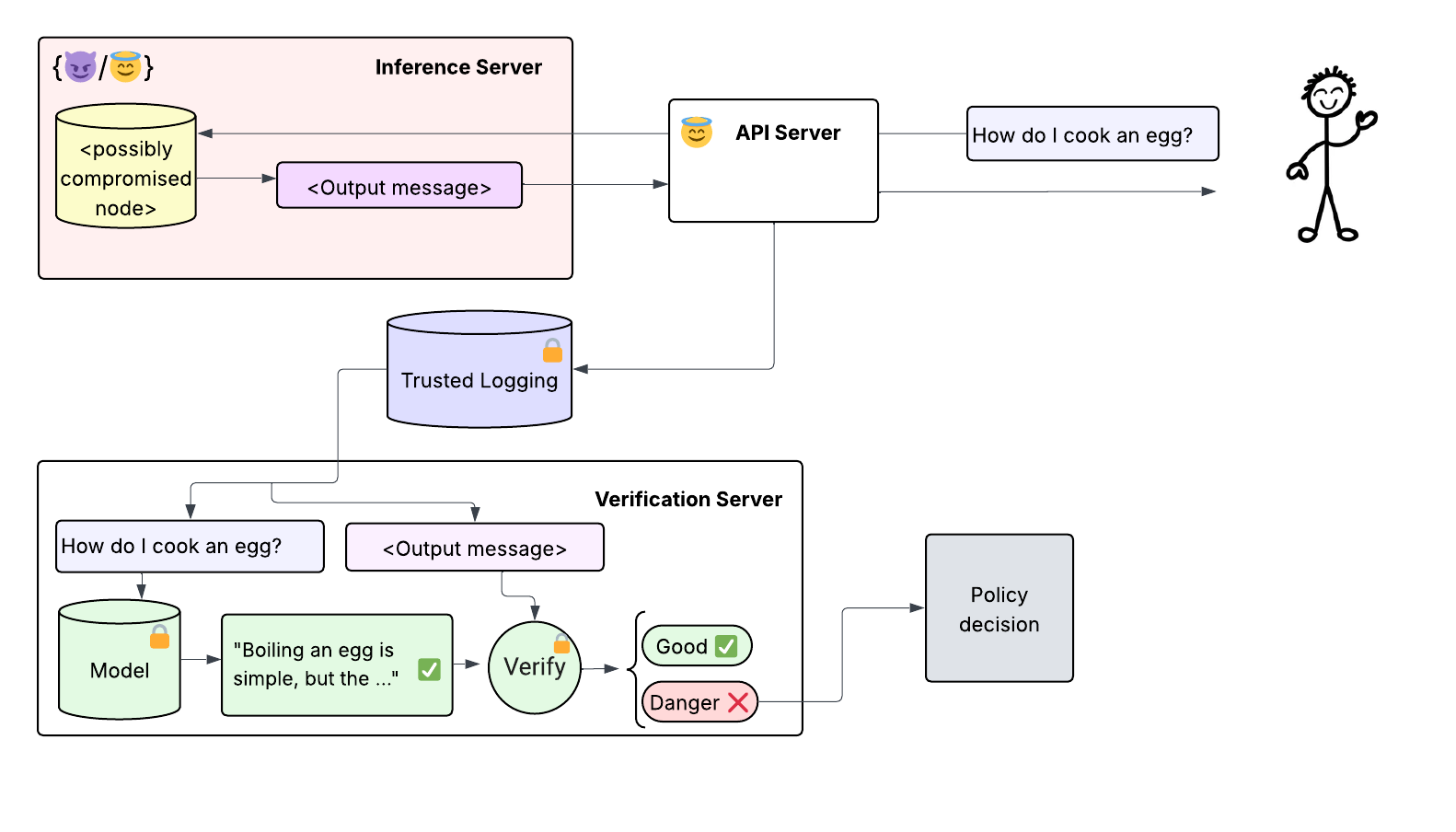}
    \vspace{-4mm}
    \caption{Overview of the verification system for runtime output validation. A lock symbol denotes parts of the inference verification scheme that require high integrity, further explored in Section \ref{sec:trust-assumptions}}
    \label{fig:headline-scheme}
\end{figure}

\vspace{-4mm}

\begin{enumerate}[noitemsep]
    \item \textbf{Inference Logging.}
    The API Server logs all inference inputs and outputs, together with references to the exact model release, code version, and the seed used for that inference.\footnote{In the event that not every (input, output, seed) tuple is verified, randomness is critical in the process of selecting which messages to verify is essential. If the messages-to-be-verified is predictable, an adversary could selectively encode sensitive data in unverified outputs.}

    \item \textbf{Verification.}
    An isolated \emph{Verification Server} re-executes a random subset of logged inferences using a trusted version of the model and codebase.
    \item \textbf{Analysis and Escalation.}
    The security team reviews any significant discrepancies between verified and observed outputs, triggering deeper audits when divergences exceed predefined statistical thresholds.
\end{enumerate}
The degree of isolation of the Verification Server is a tunable security parameter. In minimal deployments, this can be enforced through access controls and containerization; in the extreme, the verification server may be fully air-gapped.
All trust assumptions are detailed in Appendix~\ref{sec:trust-assumptions}.

\subsection{Sampled-Token Verification via Exact Matching}

A core weakness of the naive approach is the mismatch between the \emph{verification} program and the \emph{generation} program. To eliminate this, we use the \textbf{same model and configuration} for both inference and verification. The procedure is simple: re-run generation under the recorded seed and sampling parameters, compare the regenerated token at each position to the claimed token, and count mismatches. This yields a transparent, token-level exact-match score for the sequence. However, tokens that fail verification allow an adversary to potentially encode up to $\log_2(\text{vocab\_size})$ bits per mismatched token.

For efficiency, the same idea can be executed without iterative decoding. Using \emph{prefill}, a single forward pass over the input{+}output yields the next-token distributions at every position; sampling (with the recorded seed) and exact matching can then be done position-wise in parallel. This preserves teacher-forcing semantics while running significantly faster. Full pseudocode appears in Appendix~\ref{sec:naive-sampled-token-verif}.

\subsection{Classification using the Fixed-Seed Sampling Likelihood Policy}
\label{sec:fixed-seed-sampling-likelihood}

In Section~\ref{sec:fixed-seed-posterior}, we introduced the \emph{fixed-seed posterior distribution}, which captures the stochastic variability of model outputs that can arise under honest non-determinism.
Building on this, we define a practical decision rule for distinguishing between \emph{valid} and \emph{suspicious} generations.  
Our approach centers on estimating the likelihood that a given output token $t$ was honestly generated under a fixed seed $\sigma$ and the model’s nominal configuration.  
We refer to this value as the \emph{Fixed-Seed Sampling Likelihood (FSSL)}. This is simply the fixed-seed posterior distribution evaluated at a specific sampled token. Intuitively, this measures how consistent a token is with the model’s expected non-deterministic behavior under fixed-seed execution.

\paragraph{Classification via Thresholding.}
Given an estimate of $\mathrm{FSSL}(t\mid\sigma)$, we use a simple threshold rule:
\[
\text{flag } t \ \text{if}\ \mathrm{FSSL}(t\mid\sigma) < \tau,
\]
with $\tau$ calibrated on benign traffic to meet a target false–positive rate. 
Per-token decisions (or summed log-likelihoods) can be aggregated over sequences, users, or time windows for higher-level alerts.
We instantiate estimates of the \FSSL for two samplers: we introduce \emph{\GLSFull} to estimate \FSSL quantity for samplers that use the Gumbel-Max Trick, ~\GLS, and we introduce \emph{\CGSFull~(\CGS)} for samplers that use Inverse Probability Transforms for sampling, which estimates \CGS. Subsequent sections derive their estimators and evaluate performance.


\subsection{Theoretical Exfiltration Limitations}\label{sec:info-theory-bounds}
Our proposed architecture allows us to provably bound the ability of
any covert adversary to exfiltrate data out of the server under
mild trust and cryptographic assumptions (see Appendix~\ref{sec:trust-assumptions}).
Recall that our dichotomy is to \emph{either} have the adversary
get caught or else limit the exfiltration rate of the adversary. This
section is concerned with proving limits in the latter case. In other words,
we will bound how \emph{quickly} an adversary can exfiltrate without getting flagged. This is the best we can hope to prove in our context, without assumptions on the adversaries exfiltration pattern or frequency. We discuss later in this section how our theorems are related to the Exfiltration game outlined in early sections.

We prove three different bounds for the FSSL policy (Section~\ref{sec:fixed-seed-sampling-likelihood}), which apply under different architectures and assumptions.
The first is the most general and applies to any adversary under the
minimal assumptions explicated in Appendix~\ref{sec:trust-assumptions} and using the
simple FSSL filter we propose in Section~\ref{sec:fixed-seed-sampling-likelihood}.
The remaining two are tighter bounds which require additional
assumptions.
We go on to discuss the plausibility and tightness of the bounds. 

\paragraph{Rate limitation for the FSSL policy.}
First, we recall Theorem~\ref{thm:stego-limits}, which shows that any
steganography scheme can only transmit as much information as is in the entropy
of its output distribution. We combine this with two other simple observations
to give rate limitation for our FSSL policy. Because we bound with respect to such a simple policy, this bound will also apply
to any the optimizations and augmented policies discussed in Section~\ref{sec:implement-verif-server}.

Theorem~\ref{thm:perplexity-bound} is a most general rate-limiting bound.
It applies to exfiltration adversaries where both the encoder and decoder
have access to the seed that is used for generation. In the worst case, this
bound shows that an adversary
can only exfiltrate as many bits as the FSSL policy is set to. However, as we
show in Section~\ref{sec:empirical}, for many prompts and seeds both the
false positive rate and exfiltration limitations are much tighter than
the stated in this worst case, since
$\mathsf{Supp}(\dist_{\le \tau}) \ll 2^{r}$.
Here, we use $\dist_{\le \tau}$ to indicate the honest distribution, restricted to
the responses which have log-likelihood at most $r$. Here, we're using
$\mathsf{Supp}$ to indicate the set of responses which support the
input distribution.

\begin{theorem}\label{thm:perplexity-bound}
    For any distribution $\dist$, $\wFSSL$ at threshold $\tau$ is
    $(\alpha,\varepsilon)$-exfiltration secure for
    $\alpha = |\Toks|^{\ell}\cdot 2^{-\ell\cdot \tau}$ and
    $\varepsilon = (|\mathsf{Supp}(\dist_{\le \tau})| - 1)/2^n
    \le (2^{\ell\cdot \tau} - 1) / 2^{n}$.
\end{theorem}
\begin{proof}
The $\FSSL$ policy flags any response which is too surprising. Formally,
it computes
the log-likelihood (a.k.a.~surprisal or empirical entropy) for the
given response and seed. Then, it tests the likelihood $p$ and the length
of the response $\ell$ to check if $p / \ell \ge \tau$. Among the honest
distribution, the only honest responses which
are below this threshold are exactly
$\calS := \mathsf{Supp}(\dist_{\le r})$.

First, notice that by definition, for any $\hT\not\in \calS$,
$\Pr_{T\sim\dist}[T = \hT] < 2^{-\ell\cdot \tau}$ and therefore in total,
we have that $\alpha \le |\Toks|^{\ell}\cdot 2^{-\ell\cdot \tau}$
which establishes our claim about $\alpha$.

Next, observe that whenever the adversary
has any response outside $\calS$, then the $\GExfil$ game outputs $\bot$.
Therefore, when $\A$ is not flagged by the policy, it is because their
message $\hT$ is in the set $\calS$. Notice that
the maximum entropy a distribution can achieve on such a distribution is
$\log |\calS|$,
and therefore the capacity of a channel with such a support size is also
bounded by $|\calS|$ (Theorem~\ref{thm:stego-limits}). The final inequality
comes from the observation that there are at most $2^{\ell \cdot \tau}$
responses what appear with probability at least $2^{-\ell\cdot \tau}$, which
establishes the $\varepsilon$ security claim.
\end{proof}

\paragraph{Receivers without the seed.}
Under additional assumptions or infrastructure, one may expect
that the decoding adversary cannot obtain the seed that used by the encoder.
In this case, Theorem~\ref{thm:perplexity-bound} is too loose, and we should
aim to achieve bounds in the restricted game
$\GExfilR$ (Experiment~\ref{exp:exfaudit}).
In this section, we show two bounds in this restricted game below.
See Appendix~\ref{sec:trust-assumptions} for more discussion about these
assumptions.

The bounds in this section attempt to generalize beyond the $\FSSL$ policy
and will apply to \emph{any} policy with a particular structure. Motivated
by our experimental results, we see that each token of a generation is
\emph{nearly} deterministic given the prompt and seed.
In particular only about $2\%$ of tokens in our experiments are
``branching points,'' where floating-point noise makes the next token
undetermined.

\paragraph{Permitted tokens.}
The above observation motivates our characterization of
policies $\warden$ where at each index $i$,
there is only one token $t$ which is permitted, i.e., if $\hT[i] \neq t$ then
$\warden(\hT) = 1$. In practice, we don't quite obtain full determinism, so
we instead define $(p,q)$-permitted token policies. Where, at each index
there is one permitted token with probability $(1-p)$, or else there are $q>1$
permitted tokens with probability $p$ (and these probabilities are independent).

Based on our experiments, we can approximate the exfiltration bounds of the
$\FSSL$ policy with $p=0.02$ and $q=2$. However, we note that this is a
very simplified abstraction for the purposes of obtaining a possible stronger
bound under additional assumptions. In the real world, it is unclear how
closely policies satisfy the $(p,q)$-definition. It is possible an adversary
could correlate where permitted tokens occur (breaking independence) or
choose an adversarial prompt which has additional ``branching points,'' which
may be interesting for future work to explore.

First, using the capacity limitations covered in
Theorem~\ref{thm:capacity-limits}, we obtain the bound below. However, this
requires some other careful assumptions about
the security and implementation of the seed expansion. In particular, we need
that the adversary can't expand the seed on their own and that they must
commit to  a token in order to get the randomness used for the next logit.
This is the  most interactive infrastructure we propose, but we can obtain a 
very strong bound.

\begin{theorem}\label{thm:no-seed-interactive-bound}
Let $\warden$ be a $(p,q)$-restricted policy with false positive rate
$\alpha$ on $\dist$. Define $r' := \log q - h(1-p) - (1-p)\log (q-1)$ where
$h$ is the binary entropy function.
Then, $\warden$ is $(\alpha,\varepsilon)$-secure against
admissible adversaries for $\varepsilon = (2^{\ell \cdot r'} - 1) / 2^n$,
under the following additional conditions:
\begin{enumerate}[noitemsep]
    \item the seed is longer than the response, i.e.,
    $\secpar > \ell \log |\Toks|$,
    \item and the adversary $\A_1$ computes $\dist^{\RO}(\sigma)$
    once and never queries $\RO$ otherwise.
\end{enumerate}
\end{theorem}
\begin{proof}
Since the adversary $\A_1$ is only given the honest $\RO$ evaluation while
streaming an honest response, each logit is uniformly random
and unpredictable by $\A_1$. And, by assuming the seed  is longer than
the response ($\secpar > \ell \log |\Toks|$),
each of these choices remains uniformly random conditioned on the previous
ones, since there isn't enough response length to learn the seed.

Therefore, we can model each token as an independent
channel and apply Theorem~\ref{thm:capacity-limits} directly.
Adversary $\A_1$ must transmit $n$ bits to the receiver with no error.
And, we consider
each token sent as a channel, which either has $q > 1$ permitted tokens (with probability $p$)
or has one permitted token, determined by the seed and context
(i.e., if the adversary were to deviate they would get caught).
Theorem~\ref{thm:capacity-limits} gives a direct bound on the number of
messages $M$ in this case. Specifically, we have
\(
    M \ge n / (\log q - h(1-p) - (1-p)\log (q-1)),
\)
where we model the probability that there is only one permitted token
$(1-p)$ as uniformly
random symbol, since the sender has no control over it (recall that
it is unpredictable).
As we're interested in the bits per token, we can therefore bound
$n / M \le \log q - h(1-p) - (1-p)\log (q-1)$ as claimed.
\end{proof}

If the interactivity of the hash function evaluation is too expensive, we
prove one final, slightly weaker, bound. In this case we need only assume that 
the adversary is computationally bounded. In practice, this assumption can be 
enforced by requiring the server output tokens at some minimum rate. This way 
the adversary has a limited amount of time to search for ``good'' responses.

\begin{theorem}\label{thm:no-seed-bound}
Let $\warden$ be a $(p,q)$-restricted policy with false positive rate
$\alpha$ on $\dist$. For any constant $\delta > 0$ define
$r':= (p+\delta)\log q$.
Then, $\warden$ is $(\alpha,\varepsilon)$-secure against
admissible adversaries for $\varepsilon = (2^{\ell \cdot r'} - 1) / 2^n
+ e^{(-O(\ell))}$,
under the following additional conditions:
\begin{enumerate}[noitemsep]
    \item the seed is longer than the response, i.e.,
    $\secpar > \ell \log |\Toks|$,
    \item The inference distribution $\dist^{\RO}(\sigma)$ samples
    randomness sequentially,
    \item and the adversary $\A_1$ queries $\RO$ at most
    $O_{\delta}(\ell)$ times.
\end{enumerate}
\end{theorem}
\begin{proof}
The main goal of this result amounts to showing that at most
$p+\delta$ fraction of the logits in $\A_1$'s output $\hT$ have more than
one permitted token. This is equivalent to bounding the probability of the
adversary finding a response sequence where
$>(p+\delta)\cdot \ell$ tokens have more than 1 permitted token. We can
bound this using Lemma~\ref{lem:random-oracle} below. Intuitively, each
response has a probability $p$ of having more than $1$ permitted token
(which we will model with $\RO$ outputting $0$). Therefore, given
$O(\ell)$ queries, the adversary can only find a sequence like this
with probability at most $e^{-O(\ell)}$ which is negligible in the
response length. So, with high probability, the adversary has only
$(p+\delta)\ell$ locations which are not deterministic, each with an
alphabet of size $q$. So, the maximum number of bits that can be exfiltrated
is bounded by $\ell (p+\delta)\log q$.
\end{proof}

\begin{lemma}\label{lem:random-oracle}
Let $\A^\RO$ by an algorithm issuing at most $N=O(k)$ queries to a
random function
$\RO : \bits^{*} \to \bits$ that is $0$ with probability $p$, which outputs a
string $x \in \bits^{k}$. Then, the probability that $p' k$
of the prefixes of $x$ evaluate to $0$ is at most $e^{-\Theta(k)}$.
\end{lemma}
\begin{proof}
    Let $X \subseteq \bits^{k} $ be the set of $k$-bit strings which $\A$ queried.
    Note that for any \emph{fixed string} $x\in \bits^k$, the number of $0$ prefixes is a
    random variable $S_x \sim \mathrm{Bin}(k, p)$. Then, we also know that $|X| \le N$, and 
    therefore, by a union bound and Chernoff's bound, we get that
    \[
        \Pr[\exists_{x\in X} S_x \ge p' k] \le Ne^{-k \KL(p' \| p)}.
    \]
    Notably this is $e^{-\Theta(k)}$ when $N = O(k)$ and $p$ and
    $p'$ are constant.
\end{proof}

\paragraph{Tightness of the bounds}
Here we provide some estimates on how \emph{tight} we expect the above bounds to be against
a motivated
adversary. Theorem~\ref{thm:perplexity-bound}'s exfiltration bound $\varepsilon$
is likely to be tight in theory. As long as
basic cryptography exists, steganography exists as a consequence of Theorem~\ref{thm:stego-exists}~\cite{C:HLvA02}.
Even in practice, an adversary that carefully designs a steganographic scheme against the FSSL policy
should likely be able to get close to this bound. They could likely draw heavily on techniques from steganographic literature to design schemes.

Theorems~\ref{thm:no-seed-interactive-bound} and~\ref{thm:no-seed-bound} are
likely considerably looser under the stated assumptions.
The bound allows an adversary to choose from any of the permitted tokens.
So an adversary would likely need to consistently choose rare tokens
to try to meet the stated limit. But even for a simple policy, like
the FSSL policy, an adversary would need to be much more
careful about their token choice across the response to avoid getting
flagged. Moreover, it's also not clear how to design a steganographic
scheme for this setting. When the seed is known, both parties have
access to the distribution, the typical setting for steganography. But,
when the seed is secret, one may need to design a model and policy specific
encoding scheme to try and achieve a meaningful rate.

We also note that it is still an open empirical question how big the support
size of the truncated distribution is for adversarial prompts and how
frequent multiple permitted tokens are. Our estimates show that all of
these are relatively low; however, a motivated adversary may be able to
fine-tune model settings and find the best prompts possible to
increase the rate of their steganography scheme.

In Section~\ref{sec:empirical}, we give empirical estimates for the parameters in the bounds
above. Combining this with the transformation below, we can give concrete exfiltration
limitations under the stated assumptions for the policies that we propose.

\section{Implementing a Verification Server} \label{sec:implement-verif-server}

In Section \ref{sec:verification-system}, we introduce an abstract verification system and its relationship to our security game. In this section, we make this abstraction concrete with specific implementation details required to bringing this system into practice, and practical tradeoffs associated with building it.
A central design advantage of the verification server architecture is its ability to \textit{decouple security from performance}. 


\subsection{The Cost of Verification versus Generation}
 
\begin{wrapfigure}[14]{r}{0.45\linewidth}
\vspace{-0.4\baselineskip}
\centering
\includegraphics[width=0.45\linewidth]{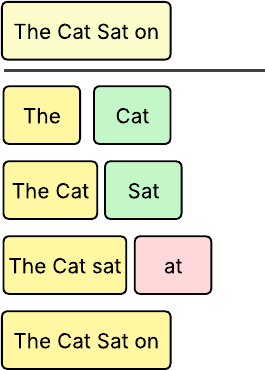}
\caption{\small Verification via prefill: the model conditions on the ground-truth prefix at each token.}
\label{fig:teacher-forcing}
\end{wrapfigure}

During verification, a single token mismatch between two generations makes their subsequent probability distributions divergent, as they are conditioned on different contexts (e.g., all continuations of ``the cat sat on'' vs.\ ``the cat sat beside'' will diverge). Instead of continuing to generate text with fixed seeds, which would diverge immediately, we exploit the fact that LLMs expose logits for all next-token predictions during the \emph{prefill} stage. A single prefill over the full input–output sequence yields conditional probabilities for every position, enabling token-level suspiciousness scores in a single forward pass. Verification therefore involves only prefill, whereas normal inference requires both prefill and autoregressive decode.
In short-context regimes, decoding is substantially less efficient than prefill: hardware utilization during decoding is typically 3–5$\times$ lower, implying a comparable 3–5$\times$ throughput gap. This disparity grows in long-context regimes, though the exact magnitude depends on attention optimizations and implementation details \cite{erdil2025inferenceeconomicslanguagemodels}.

\subsection{Inverse Probability Transform Scoring Function \CGSFull}
In Section \ref{sec:fixed-seed-sampling-likelihood} we introduce the notion of a classification using Fixed-Seed Sampling Likelihood. 
In this section we introduce a technique for approximating the fixed-seed sampling likelihood (FSSL) for the inverse probability transform sampling technique (IPT).
We investigate a method \emph{\CGSFull}, which is concurrently introduced in parallel work \cite{tokendifr}.
The core assumption is that across repeated model evaluations, the per-token probability vectors $\vec{p}_\theta(x)$ may vary due to small numerical or stochastic perturbations, but such variation is relatively small.

Intuitively, if the true sampling process (with fixed seed $\sigma$) selects a random value $u \in [0,1]$ that falls within the probability interval of one token, nearby tokens, whose cumulative probability mass lies close to $u$, remain plausible under valid non-determinism.  
Hence, instead of measuring the full empirical fixed-seed posterior (which would require many re-samplings per context), we approximate the likelihood of plausible tokens with a Gaussian centered at the seed sample $u$.

Formally, we compute a verification score for the observed token $t^\star$ by integrating a Gaussian density (mean $u$, width $\sigma$) over the token’s probability interval $[a,b]$ in the cumulative distribution $\bar{F}$ of $\vec{p}_\theta(x)$.  
This integral estimates the probability that a perturbed distribution, drawn from $\mathcal{P}_\theta(x)$, would produce $t^\star$ under the same seed.  
The log of this mass gives a continuous, differentiable proxy for the fixed-seed posterior probability.  
The \CGSFull~is defined in Algorithm \ref{alg:convolved-gaussian-score}, and illustrated in Figure~\ref{fig:CGS-diagram}.

\subsection{Gumbel-Max Scoring Function \GLSFull}

We also introduce the \GLSFull to estimate \GLS, which is defined for a given noise distribution over token probabilities, which by default we assume to be Gaussian $\mathcal{N}$, also introduced in concurrent work \cite{tokendifr}.

Given a sampled token $t^*$, and knowledge of the inference server’s sampling seed $\sigma$, the verifier regenerates the same Gumbel noise the inference server generated $\vec{g}_\sigma \sim \mathrm{Gumbel}(0,1)^V$.
The verification server generates the probability distribution $\hat{\ell}$. Then, we compute the probability that token $t^*$ surpassed that token as the sampled token; that is, computing the probability $\Pr[\vec{g}_\sigma[i] + \hat{\ell}[i] + \mathcal{N} >\vec{g}_\sigma[t^*] + \hat{\ell}[t^*] + \mathcal{N}]$ for token $i$. The full method and derivations are provided in Appendix~\ref{sec:GLS-method}, and Algorithm \ref{sec:GLS}.

    \begin{figure}[htbp]
        \centering
        \includegraphics[width=0.75\linewidth]{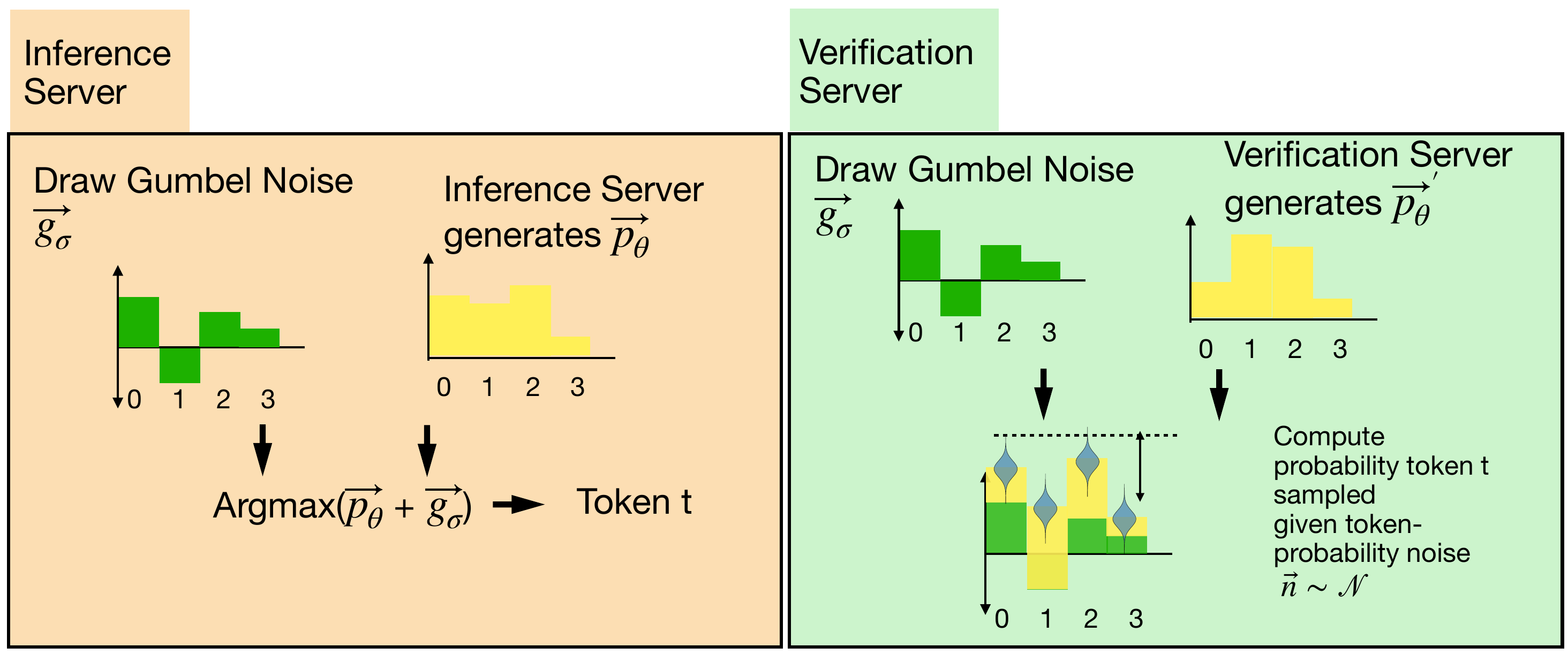}
        \caption{Diagram of the \GLSFull~function. The \GLSFull~estimates the likelihood of a particular token being sampled given a specific noise distribution across token probabilities.}
        \label{fig:gumbel-max-trick-method}
    \end{figure}

\subsection{How to Pool Across Multiple Tokens}

In practice, computing a per-token \textit{suspiciousness} score is only the first step; the harder problem is determining how to pool these scores across larger units of activity.
Apriori, as a defender we do not know how an adversary will distribute an exfiltration attempt (e.g., they can attack across all tokens from a single user, across the first token of many users, or across all outputs generated on a particular GPU); thus there is no universally correct aggregation rule. 
Simple pooling strategies such as taking the 90th percentile or the mean correspond to different implicit threat models. 
 The security team must define the granularity at which suspicious tokens are acted upon (e.g. user-level, request-level, or hardware-level) and the time horizon over which evidence accumulates. While our method provides calibrated token-level scores, the pooling and escalation policy is necessarily system-dependent and must be tailored to the deployment's trust boundaries and adversary model. In this work we are purposefully agnostic to the adversary and thus do not comment on the right method to pool scores across tokens. 



\subsection{Message Sampling Rates}

To understand the efficiency of our verification scheme, we analyze what percentage of traffic must be sampled to detect malicious messages with high
confidence. Figure~\ref{fig:sampling-confidence} illustrates the key statistical property of our approach: the \emph{sampling percentage} required depends primarily on
the \emph{proportion} of malicious messages in the traffic. In Figure~\ref{fig:sampling-confidence} we show the percentage of messages that need to be sampled in order to sample at least 1000 malicious messages, assuming the malicious messages are IID. Notably, even for very slow exfiltration attempts (e.g., 1\% malicious traffic), the sampling percentage remains practical, even for high confidence.  See Appendix~\ref{sec:sampling-rates} for the full derivation.


\begin{figure}[htbp]
    \centering
        \includegraphics[width=\linewidth]{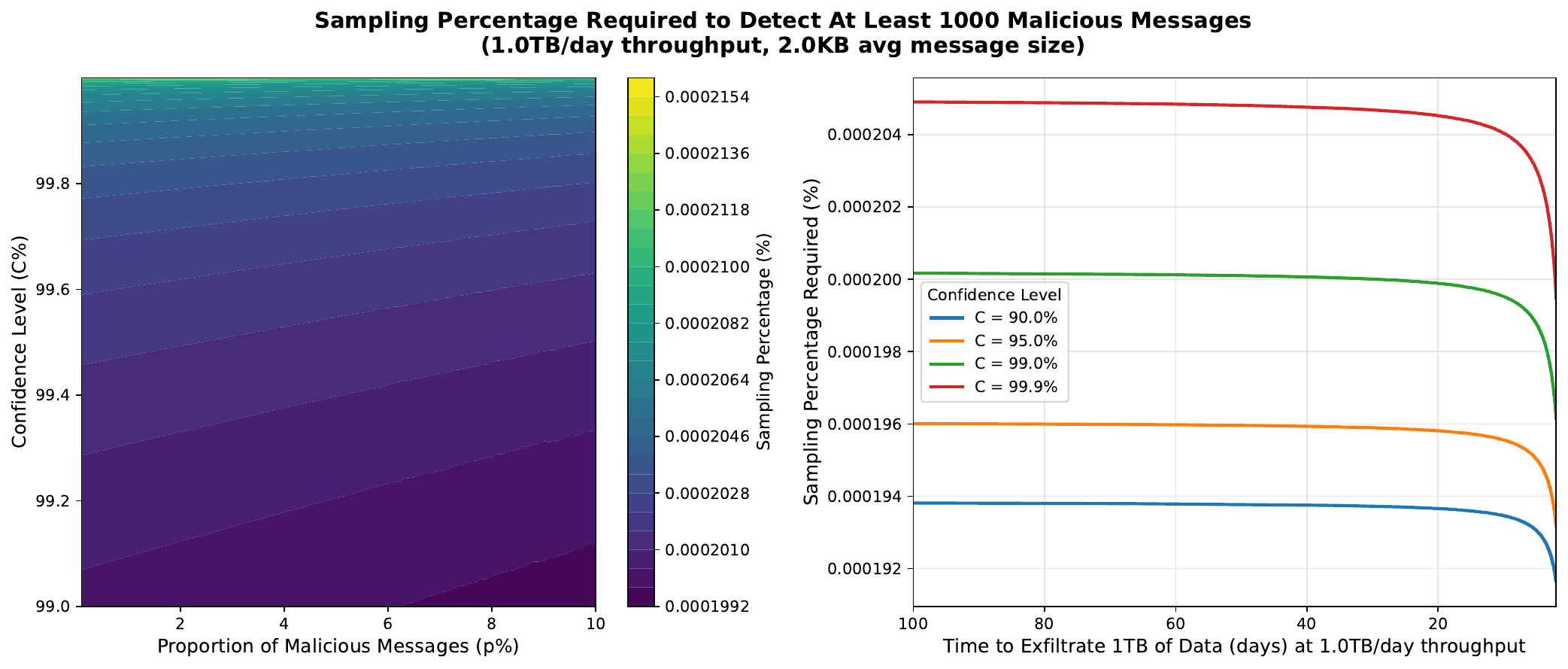}
    \caption{Sampling requirements for detecting at least 1,000 malicious messages with varying confidence levels, shown as percentage of data sampled.
    \textbf{Left:} Heatmap showing the percentage of traffic that must be sampled as a function of the proportion of malicious messages and target confidence level.
    \textbf{Right:} Sampling percentage versus exfiltration time, assuming 1\,TB/day text throughput and 2KB average message size. Critically, verification remains extremely practical even for slow
attacks.}
        \label{fig:sampling-confidence}
\end{figure}

\section{Empirical Results} \label{sec:empirical}
    
Anticipating a ``reasonable'' attacker is inherently brittle; capable adversaries can be more creative and better resourced than any fixed threat model. Instead of speculating about strategies, we take a distribution-free, information–theoretic view and bound \emph{worst-case} leakage under our verifier. At each position $t$, the verifier induces an \emph{admissible set} of tokens
\[
    A_t(\tau)\;=\;\{\,i:\mathrm{FSSL}_t(i)\ge\tau\,\},
    \]
where $\mathrm{FSSL}_t$ is our Fixed Seed Sampling Likelihood (or more accurately, our estimate of the FSSL for that token) and $\tau$ is a threshold. An adversary sending an \textit{accepted} token cannot convey more than $\log_2|A_t(\tau)|$ bits at step $t$, hence the sequence-level capacity is $\sum_t \log_2|A_t(\tau)|$. This converts Theorem~\ref{thm:perplexity-bound} into an operational metric: for any attack, leakage is upper-bounded by the bits available from tokens the verifier would accept. We note that this is an extremely strong bound, and in practice we expect true exfiltration likely to be much lower than this bound.

\subsection{Problems with FSSL Estimators and Logit-Rank as a Solution}

Our FSSL estimates can be thought of as a ``suspiciousness'' score. In designing a meaningful score for suspiciousness one prefers that the suspiciousness of a token is a continuous score, from ``not-suspicious'' to ``very suspicious''. However, empirically we find that our proposed scoring function (\GLS) is not perfectly calibrated. In Qwen-30B experiments, roughly 0.1\% of correctly sampled tokens are labeled \emph{maximally suspicious}, as seen by looking at the middle plot of Figure \ref{fig:gls_vs_logit_ranks}, where we see that the "rank" of the \GLS~scores is high (in practice, we see that the sampled tokens are tied for last place with a large number of unlikely tokens).
        
A simple and effective auxiliary signal is the \emph{logit rank} of the sampled token (before temperature and top-k/top-p filtering). Empirically, the majority of sampled tokens fall within the top R ($\approx 16$) by raw logits, which we see in the right subplot of Figure \ref{fig:gls_vs_logit_ranks}. Intuitively, this constrains the attacker's effective alphabet to at most $R$ choices per token position.

        \begin{figure}[htbp]
              \centering
              \includegraphics[width=1.0\linewidth]{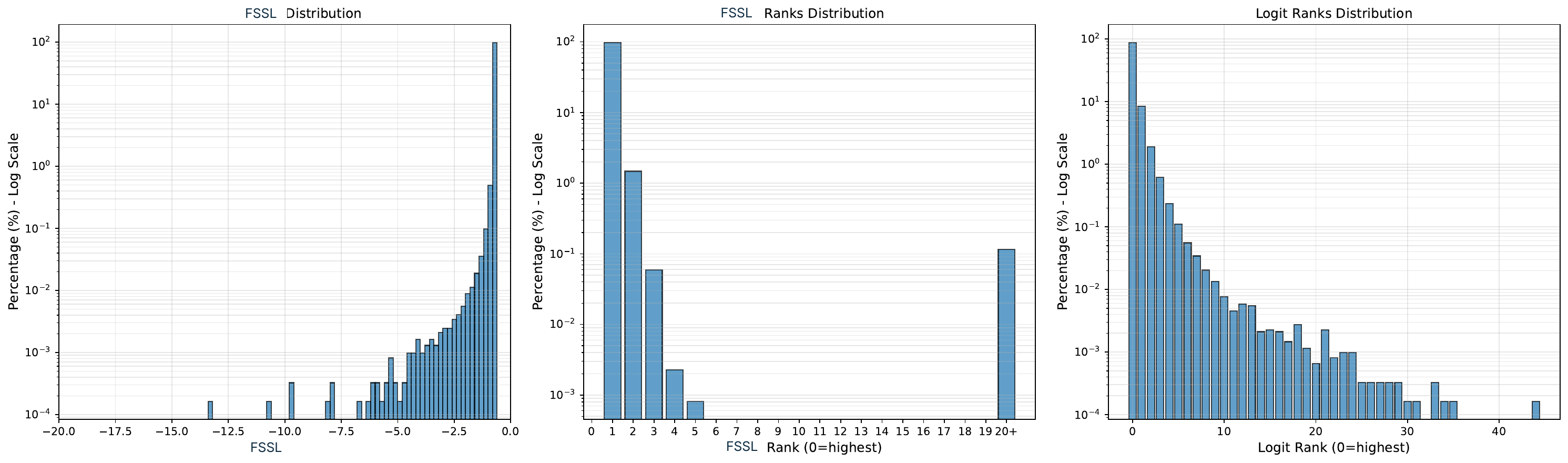}
              \caption{Comparison of \GLSFull~(\GLS) verification and raw logit ranks for sampled tokens. \textbf{Left:} Distribution of \GLS~scores
          (log-probabilities) for sampled tokens with $\sigma=1.0$. \textbf{Middle:} Distribution of \GLS~ranks, where rank 0 indicates the sampled token has the
          highest \GLS~score. Ranks are grouped as 0--19 individually and $\geq$20 combined; tokens in the $\geq$20 bin are typically tied for last place with many other low-probability tokens (i.e., effective rank $\approx |V|$). \textbf{Right:} Distribution of raw logit ranks (before temperature and
          top-k/top-p filtering), where rank 0 indicates the token with the highest raw logit. All y-axes use log scale to show the full distribution of percentages.}
              \label{fig:gls_vs_logit_ranks}
        \end{figure}


Augmenting \GLS~with a \emph{logit rank} check yields a \textbf{strict Pareto improvement} over simply thresholding according to the FSSL: for any FSSL threshold $\tau$ and rank cutoff $R$, we introduce a three-class classification
        \begin{align*}
        &\textsc{safe}:\ \mathrm{FSSL}_t(i)\!\ge\!\tau\\
        &\textsc{suspicious}:\ \mathrm{FSSL}_t(i)\!<\!\tau\ \wedge\ \mathrm{rank}_t(i)\!\le\!R\\
        &\textsc{dangerous}:\ \mathrm{FSSL}_t(i)\!<\!\tau\ \wedge\ \mathrm{rank}_t(i)\!>\!R
        \end{align*}
        achieves \emph{lower or equal} FPR at \emph{lower} information leakage than \GLS~alone. We note that this logit-rank and FSSL score can be combined into a single classifier, e.g. a logistic regression that takes two features: FSSL score and logit-rank. An important advantage to a tiered classification scheme allows to have simpler bounds of extractable information, as there is a trivial upper bound on information expressed by a suspicious (not dangerous) token $I \leq log_2(R) \leq log_2(|\textrm{vocab size}|)$. As such, given the token-level classifier $C$, the per-token exfiltratable information is:
        \[
\mathrm{ExfiltratableInformation}(t)
= \begin{cases}
\log_2 |A_t|, & \text{if } C(t) = \textsc{safe}, \\[6pt]
\log_2 R, & \text{if } C(t) = \textsc{suspicious}, \\[6pt]
\log_2 |\Toks|, & \text{if } C(t) = \textsc{dangerous}.
\end{cases}
\]

    \paragraph{Experimental Set up}
    
    We sample 500 prompts from \textsc{LMSYS-Chat-1M} \cite{zheng2024lmsyschat1mlargescalerealworldllm}, and for each prompt, we generate
continuations of max length 500 tokens.  We evaluate \emph{Llama-3.2-3B-Instruct},
\emph{Llama-3.1-8B-Instruct} \cite{llama3modelcard}, \emph{Qwen3-1.5-MoE-A.2-7B}, and \emph{Qwen3-30B-A3B} \cite{yang2025qwen3technicalreport}, served via
\texttt{vLLM}.\footnote{We evaluate on \texttt{vLLM} as this is the most popular open-source
inference provider, and thus aim to estimate \GLS~scores, rather than \CGS, as vLLM uses the Gumbel-Max trick for sampling.} For all our experiments we use \GLSestimator. We use temperature $T=1.$, top-$k=50$, top-$p=0.95$, and a fixed RNG seed for generation. The verifier shares the tokenizer and chat template with the generator. We sweep over multiple $\sigma$ values (from 0.001 to 1.) for \GLSestimator~hyperparameters. 
    
    We focus on verifying inference conducted via the standard auto-regressive decoding process. While speculative decoding is increasingly common in practice, our verification pipeline is designed around the traditional decoding setup; we discuss extensions to speculative decoding in Appendix~\ref{sec:speculative-decoding-verification}.

        


        
          \begin{figure}[htbp]
                  \centering
                  \includegraphics[width=.75\linewidth]{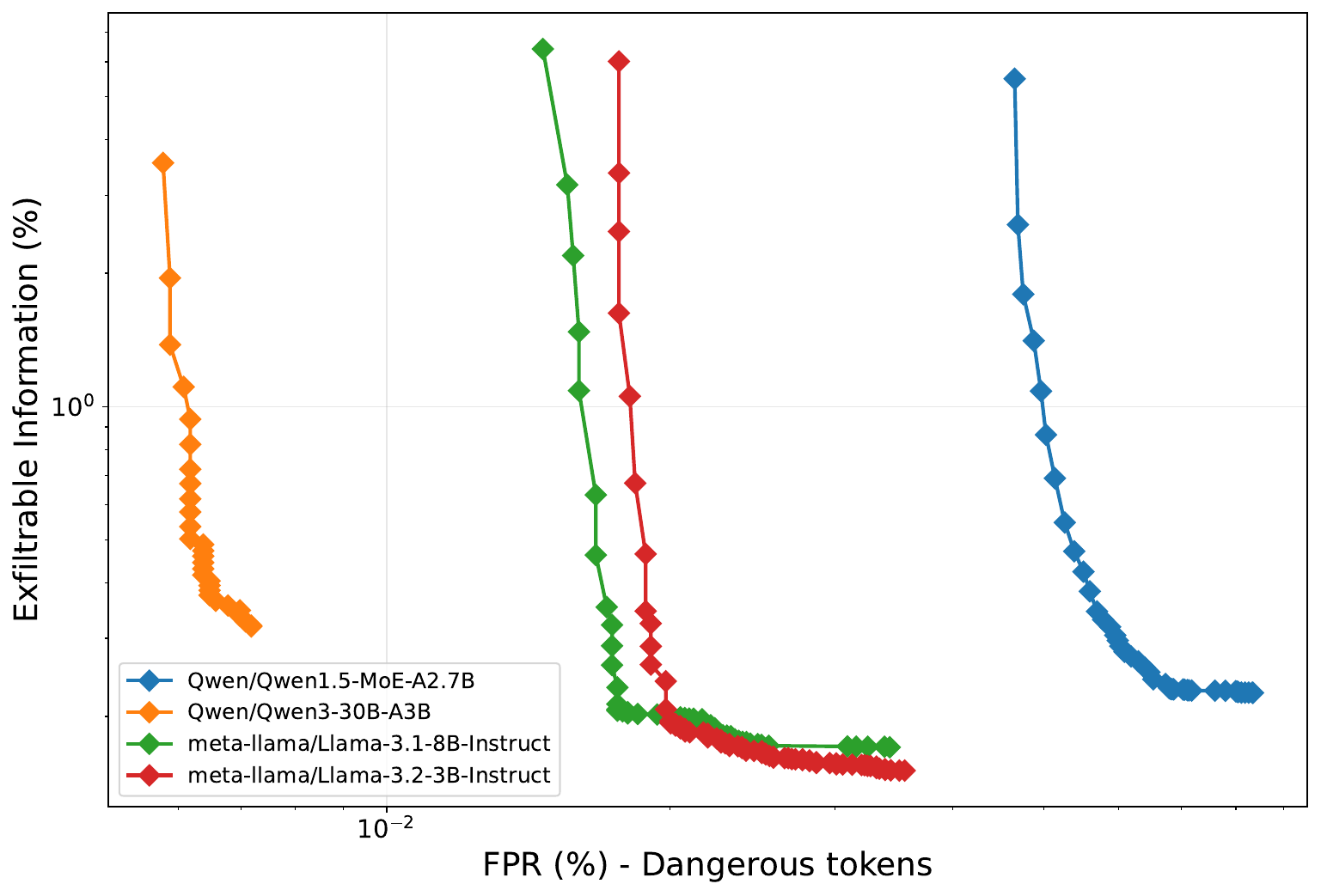}
                  \caption{\textbf{Headline figure: Pareto frontiers for rank-aware classification thresholds.} For logit-Rank=8 as a filtering function to differentiate between suspicious and dangerous tokens. The plot is filtered such that less that 0.5\% of tokens are deemed ``suspicious'', and only points that are not pareto-dominated by both FPR and exfiltratable information are plotted.}
              \label{fig:headline-pareto-multi-model-rank}
          \end{figure}

In Section \ref{sec:rank-ablations}, we demonstrate an ablation of these findings and show a robust improvement when introducing logit rank, where if rank =$\textsc{None}$ then all tokens below an FSSL threshold $\tau$ are classified as dangerous. More information on the \GLS~scores for different tokens are presented in Appendix~\ref{sec:gumbel-score-statistics}.

\section{Discussion}

In this work, we formalize the security game of model weight exfiltration and propose a scheme for detecting steganographic attempts to exfiltrate weights. We provide introduce information theoretic bounds on exfiltratable information, and extend them to the physical world through empirical results with Llama and Qwen models. We find that using our scheme we are able to limit exfiltration to $\leq 0.5\%$ with a minimal FPR. To the authors' knowledge, no comparable detection and exfiltration bound exists, suggesting a slow down of 200$\times$ in time-to-exfiltrate. If exfiltration were to take 2 days, it now takes over a year.

One critical strength of our proposal is that it requires minimal interaction between the inference stack and the verification server, allowing model providers to largely decouple performance from security. However, one potential way to improve the accuracy of the verification process is to have the inference server sample tokens aware of the verification-process. That is, at inference time, if the inference server computes that there are multiple potential tokens (with high FSSL), choose one according to some deterministic process (e.g. alphabetically). This would further limiting the fixed seed posterior distribution, but would involve some minimal additional computation on the inference stack.

One of core challenges of this proposal is that if an adversary is able to corrupt an inference provider, then meaningful steps must be taken to protect the verification server from the same threat. There are multiple possible solutions to this, depending on the specific threat model an inference provider may have. In the most extreme setting, an inference provider may create a secondary inference stack for verifying inference that is kept in parallel with the main inference stack, it may even run in a Trusted Execution Environment to verify integrity of the computation. In the event that this is too costly for an inference provider to maintain, an earlier code checkpoint from a point prior to when an adversary was known to be in the system, can be used as a reference implementation. 
Further, the defending the verification server benefits from not having the same performance requirements of the inference stack ; depending on availability constraints, the verification server could even exist in a largely air-gapped format with minimal connection to the inference stack, except at limited times across a secure channel. 
The fact that verification is much more efficient than generation, and that randomly sampling messages to verify can produce high confidence probabilistic guarantees (Section \ref{sec:sampling-rates}) suggests that the amount of systems-level overhead required to implement a verification server may not be a blocker, if steganography is deemed a meaningful threat.

This system cannot standalone as a method for defending against against model weight exfiltration; in Appendix \ref{sec:high-level-defenses}, we describe how this scheme fits into a larger outline of defenses to protect model weights. It's important to note that this scheme does not \textit{prevent} model weight exfiltration and rather only \textit{detects} steganographic messages. There are a set of possible security responses in response to suspicious tokens; for example, tokens deemed suspicious can be logged and aggregated to look for potentially malicious patterns by an adversary. Alternatively, outbound traffic can be gated by this verification server, and when failing verification, inference is regenerated and tried again. It is possible to rate limit users (or GPUs) proportional to the number of tokens that they produce that are deemed suspicious. We give a more complete discussion of policy recommendations of how this scheme fits into a larger scheme of in Appendix~\ref{sec:policy-response-to-exfil}.

Overall, this work establishes a concrete foundation for reasoning about and detecting model weight exfiltration through steganographic channels. By combining formal security analysis with an efficient, empirically validated verification scheme, we demonstrate that meaningful protection against inference-time exfiltration is both achievable and practical for large-scale model deployments.

\newpage 
\section{Acknowledgements}
We thank Miranda Christ for early conversations that helped frame the problem and for her foundational explanations of watermarking and steganography in generative AI. We are grateful to Nicholas Carlini for his thoughtful comments during this work. We also thank Jacob Lagerros for his valuable discussions, feedback on the paper, and practical insights drawn from implementing the proposed scheme in real-world settings.

Some of this work was conducted while AK, DR, and RR were participating in the ML Alignment and Theory Scholars (MATS) program, whose support we gratefully acknowledge. RR was supported in part by the NSF BCS-2218803 grant, as well as additionally supported by a grant from Coefficient Giving. 

\bibliographystyle{alpha}
\bibliography{citations,alex_cites}

\appendix

\section{Appendix}

\section{Steganography Syntax and Relevant Results}\label{sec:stego}
In this section, we review some relevant notation, concepts, and terms
that are used to understand steganographic protocols. We also reprove some
well known results that support our own theorems in the main body of the paper.

\begin{definition}[Steganographic Scheme]
Let $\msglen,\coverlen\in\N$ be integers and $\Toks$ be a set of tokens (a.k.a.~vocabulary or alphabet).
Let $\dist$ be the passive channel distribution with support
$\supp \subseteq \Toks^{\le \coverlen}$.
A \emph{steganography scheme for $\dist$} is
a pair of stateful (possibly randomized)
algorithms $\Steg = (\Enc,\Dec)$ with the following syntax:
\begin{itemize}[noitemsep]
    \item $\Enc(K,\st, m) \to c$, takes a key $K\in \bits^{\secpar}$, a state $\st$, a message $m\in \bits^{\le\msglen}$ and outputs a covertext $c\in \Toks^{\le \coverlen}$.
    \item $\Dec(K, \st,  c) \to m$, takes a key $K\in \bits^{\secpar}$, a state $\st$, a covertext
    $c\in \Toks^{\le \coverlen}$ and outputs a message
    $m\in \bits^{\le\msglen}$.
\end{itemize}

For a steganography scheme to be useful, it must be \emph{correct},
meaning that for every $\st$ and $m$,
$\Pr_K[\Dec(K,\st,\Enc(K,\st,m)) = m] = 1 - \negl(\secpar).$ When,
this probability is $1$, we say the scheme is \emph{perfectly correct}.
It must also be \emph{secure}, meaning that for any $\st$ and $m$,
\(
    \KL(\dist_{\Steg} \| \dist) = \negl(\secpar),
\)
where $\dist_{\Steg}$ is the covertext distribution generated
as $K \getsr \bits^{\secpar}$ and $c \getsr \Enc(K,\st,m)$.
\end{definition}

The principal measure of efficiency that we are interested in
is the \emph{rate} of a scheme $\Steg$. This is defined
as the message length $\msglen$ divided by the average covertext length,
which we calculate as $\Ex_{c\sim \dist}[|c|]$, since a secure scheme
cannot differ significantly from this expectation.\footnote{This is
assuming that one requires correctness for a message distribution
which is chosen uniformly at random and is therefore incompressible.}
Therefore, we write $R = \msglen / \Ex_{c\sim \dist}[|c|]$.

Note that there are many other variants of steganography, including
stateless, public-key, robust, etc.~\cite{JOC:AP02,C:HLvA02,IDK:vAH04}. However, in
our context, we wish to give a potential exfiltrator as much
theoretical power as possible. So we assume that both endpoints
are stateful and can have an arbitrarily long secret key in common.

\subsubsection{Relevant Results in Steganography}
The \emph{capacity} of a steganography distribution $\dist$ is
defined as the maximum
rate $R$ achievable by a steganography scheme $\Steg$ that is
both correct and secure, asymptotically over many runs of the
scheme. One well-known result is that this
capacity is at most $H(\dist) / \Ex_{c\sim \dist}[|c|]$
for perfectly secret and correct schemes.

More generally, when a scheme $\Steg$ outputs a distribution
$\dist_{\Steg}$ with $\KL(\dist_{\Steg}||\dist) < o(1)$
which may not perfectly match the perfect cover
distribution, then the capacity is bounded by
$\big(H(\dist_{\Steg}) - \KL(\dist_{\Steg}||\dist)\big)
/\Ex_{c\sim \dist}[|c|]$ against information-theoretic adversaries.

\begin{theorem}[Steganographic limitations]\label{thm:stego-limits}
Let $\Steg$ be a steganography scheme for $\dist$
and $\mathbf{M}$ be a random variable with $H(\mathbf{M}) = \ell$.
Let $\dist_{\Steg}$ be the covertext distribution generated
as $K \getsr \bits^{\secpar}$ and $c \getsr \Enc(K,\st,\mathbf{M})$.
Then, $\Steg$ cannot be perfectly correct if
\(
    \ell > \lceil H(\dist_{\Steg})\rceil.
\)
\end{theorem}
\begin{proof}
This theorem is an immediate consequence of Shannon's source coding
theorem~\cite{shannon1948mathematical,shannon1959coding}.
This theorem states: (i) that no random variable $X$
can be encoded with fewer bits than $H(X)$ in expectation and (ii) that there exists
an encoding of $X$ with expected length at most $\lceil H(X)\rceil$. Therefore, if
$\lceil H(\dist_{\Steg}) \rceil < H(\mathbf{M})$ and $\Steg$ is correct with probability $1$, then
one could encode $\mathbf{M}$ as the covertext, using a random key $K$ and random state $\st$,
$\mathbf{C} \gets \Enc(K,\st,\mathbf{M})$ and decoded with $\Dec(K,\st,\mathbf{C})$.
Then, based on the source coding theorem, $\mathbf{C}\sim \dist_{\Steg}$
(and therefore $\mathbf{M}$) can be encoded with $\lceil H(\dist_{\Steg}) \rceil < H(\mathbf{M})$
bits, which is a direct contradiction.
\end{proof}

In contrast to the above theorem and to illustrate the difficulty of preventing covert
exfiltration, we recall a well-known theorem from~\cite{C:HLvA02} stated in
Theorem~\ref{thm:stego-exists}. This shows that assuming classical cryptography exists
at all, i.e., one-way functions exist, then one can construct a steganographic scheme
that is near-optimal. The following theorem shows that there is some steganographic
scheme which fools any efficient \emph{policy} under minimal assumptions. Here
and the rest of the paper, a policy is a function $\warden : \supp(\dist) \to \bits$,
which aims to output $1$ on covertexts generated steganographically but not on
samples from the honest distribution.
We say that $\warden$ has a false positive rate $\alpha$ over $\dist$ if
$\Pr_{x\sim \dist}[\warden(x) = 1] = \alpha$.

\begin{theorem}[Due to~\cite{C:HLvA02}]\label{thm:stego-exists}
Let $\warden$ be an efficiently-computable policy with a false positive rate $\alpha$
over $\dist$. Then, assuming a $\delta$-secure pseudorandom function $\mathsf{PRF}$ exists,
there is an inefficient steganography scheme $\Steg$ for $\dist$ which,
on every message that is $\minH(\dist) - O(\log(1/\varepsilon))$ bits long,
is correct with probability at least $1-\varepsilon$ and is caught with probability at most
$\alpha + \delta$ by $\warden$.
\end{theorem}


Finally, we state a known result from information theory in Theorem~\ref{thm:capacity-limits}.
This result gives a lower bound on the number of messages that must be sent through
independent channels, each of capacity $C$. This is a useful result to recall for the
bounds we derive later in Section~\ref{sec:info-theory-bounds}.
\begin{theorem}[Capacity limitations]\label{thm:capacity-limits}
Let $X$ be a uniform random variable with support of size $q$ (so that
$H(X) = \log q$)
and $Y^n$ be $n$ messages sent through a channel of capacity $C$. Then, in order
to transmit $X$ with error at most $\varepsilon$, requires that
\(
    n \ge \big(\log q - h(\varepsilon) - \varepsilon\log(q-1)\big) / C.
\)
Notably, when $Y$ has alphabet size $q'$ and is uniformly
random with probability $p$,
then we obtain
\[
    n \ge \frac{\log q - h(\varepsilon) - \varepsilon\log(q-1)}{\log q' - h(p) - p\log (q'-1)}.
\]
\end{theorem}
\begin{proof}
For two jointly distributed variables $(X,Y^n)\sim\dist$ and any function
$f$, let $E$ be the event that
$X \neq \tilde{X}$, where $(X,Y^n) \sim \dist$ and $\tilde{X} = f(Y^n)$.
Then,
\[
    H(X\mid Y^n) \le h(\Pr[E]) + \Pr[E]\cdot \log(q - 1).
\]
Using the mutual information between $X$ and $Y^n$,
\[
    I(X ; Y^n) = H(X) - H(X\mid Y^n) \ge \log q - (h(\Pr[E]) + \Pr[E]\cdot \log(q - 1)).
\]
And, any encoding scheme that uses a (memoryless) channel $n$ times, $I(X;Y^n) \le nC$.
Combining these and using target error rate $\Pr[E] = \varepsilon$,
we establish the claimed bound
\[
    \log q - (h(\varepsilon) + \varepsilon\cdot \log(q - 1)) \le I(X;Y^n) \le nC.
\]
Finally, the further claim comes from the fact that the maximum
capacity for a channel with alphabet size $q'$ and which is
uniformly random with probability $p$ is
$C = \log q' - h(p) - p\log (q'-1).$
\end{proof}

\section{Algorithms}\label{sec:algorithms}
\subsection{Inverse Probability Transform (IPT) Fixed-Seed Sampling Likelihood (\FSSL) Estimators}

\begin{algorithm}[H]
\caption{\CGSFull~(from model $\theta$ and context $\context$)}
\label{alg:convolved-gaussian-score}
    \begin{algorithmic}[1]
    \Require model parameters $\theta$;token probabilities $p \in \Delta^{|\Toks|}$; observed token $t^\star \in \{1,\dots,|\Toks|\}$; seed $\seed$; context position $i$; width $\sigma>0$; small $\varepsilon>0$
    \State $\mu \gets \mathrm{Uniform}(0,1; \RO(\seed\| i))$ \Comment{(1) sample random value in $[0,1]$}
    \State $p \gets \mathrm{PredictTokenProbs}(\theta, x_{1:t-1})$ \Comment{$p \in \Delta^{|\Toks|}$; include same temperature/top-$k$/nucleus policy as generation}
    \State $\bar F \gets \mathrm{concat}(0,\ \mathrm{cumsum}(p))$ \Comment{token CDF with $\bar F_0=0$, $\bar F_j=\sum_{i\le j}p_i$}
    \State $a \gets \bar F_{\,t^\star-1}$, \quad $b \gets \bar F_{\,t^\star}$ \Comment{start/stop of $t^\star$'s interval} 
    \State $\Phi(u;\mu,\sigma) \gets \tfrac{1}{2}\!\left[1+\mathrm{erf}\!\big(\tfrac{u-\mu}{\sqrt{2}\sigma}\big)\right]$ \Comment{Gaussian CDF}
    \State $\text{mass} \gets \Phi(b; x,\sigma) - \Phi(a; x,\sigma)$ \Comment{(3) integral of Gaussian over $[a,b]$}
    \State $S \gets \log(\text{mass} + \varepsilon)$ \Comment{(4) return log-score of observed token}
    \State \Return $S$
\end{algorithmic}
\end{algorithm}

\begin{figure}[htbp]
    \centering
    \includegraphics[width=0.75\linewidth]{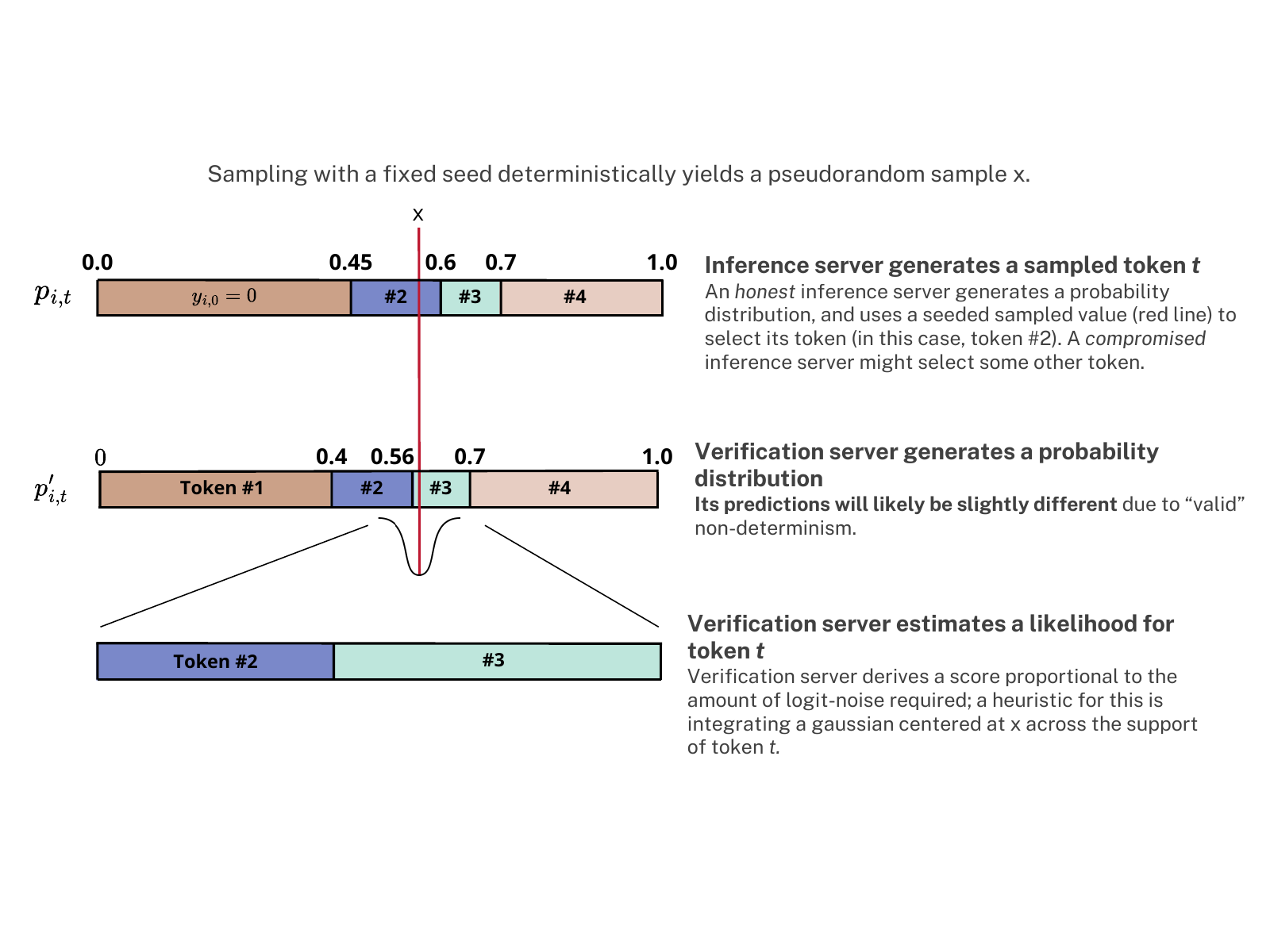}
    \caption{Diagram of the \CGSFull~method. The \CGSFull~estimates the likelihood of a token under the inverse probability transform by integrating a Gaussian over the token's probability interval.}
    \label{fig:CGS-diagram}
\end{figure}

\subsection{Gumbel-Max \FSSL Estimators (\GLS) } \label{sec:GLS-method}

\begin{algorithm}[H]
\caption{\GLSestimator (used to estimate \GLS)}
\label{sec:GLS}
\begin{algorithmic}[1]
\STATE \textbf{Input:} Claimed token $t^* \in \Toks$, verifier logits $\hat{\ell} \in \mathbb{R}^{|\Toks|}$, Seed $\seed$, temperature $T$, filtering parameters $(k, p)$, active subset size $s$, noise CDFs/PDF $\{F_{\varepsilon_i}, f_{\varepsilon_{t^*}}\}$, number of samples $M$
\STATE \textbf{Output:} Likelihood estimate $\hat{\pi}_{t^*} \in [0, 1]$
\STATE
\STATE Sample Gumbel noise $G_{\seed,i} \sim \text{Gumbel}(0,1; \RO(\seed\| i))$ for $i=1,\ldots,|\Toks|$
\STATE $S \gets \{i_1, \ldots, i_s\}$ where $i_j$ are top-$s$ indices as ranked by $\hat{\ell}_i + T \cdot G_{\seed,i}$
\STATE
\STATE $\tau \gets \min_{i \in \{1,\ldots,|\Toks|\}} \hat{\ell}_i$ such that $i$ does not satisfy top-$k$ or top-$p$ filter
\STATE
\FOR{each token $i \in S$}
    \STATE $A_{\seed,i} \gets \hat{\ell}_{t^*} - \hat{\ell}_i + T(G_{\seed,t^*} - G_{\seed,i})$
    \STATE $B_i \gets \tau - \hat{\ell}_i$
\ENDFOR
\STATE
\STATE $\hat{\pi}_{t^*} \gets 0$
\FOR{$j = 1$ to $M$}
    \STATE Sample $x_j \sim f_{\varepsilon_{t^*}}$
    \IF{$\hat{\ell}_{t^*} + x_j > \tau$}
        \STATE $\pi_{t^*,j} \gets \prod_{i \in S} F_{\varepsilon_i}\left(\max(x_j + A_{\seed,i}, B_i)\right)$
        \STATE $\hat{\pi}_{t^*} \gets \hat{\pi}_{t^*} + \pi_{t^*,j}$
    \ENDIF
\ENDFOR
\STATE
\RETURN $\hat{\pi}_{t^*} / M$
\end{algorithmic}
\end{algorithm}

\subsection{Naive Methods for Sampled Token Verification}\label{sec:naive-sampled-token-verif}

\begin{algorithm}[H]
  \caption{Sampled token verification through teacher forcing}
  \label{alg:sampled-token-verif}
  \begin{algorithmic}[1]
    \Require input sequence $\context$, target output sequence $y$, seed $\seed$, model $\theta$
    \State $\text{diff} \gets 0$
    \For{$i \gets 1$ to $|y|$}
      \State $\;\;\; \text{next} \gets \text{SampleNextToken}(\theta,\context; \RO(\seed\| i))$
      \If{$\text{next} = y_i$}
        \State $\;\;\;\;\text{diff} \gets \text{diff} + 1$
      \EndIf
      \State $\;\;\; \context \gets \context \| y_i$ \Comment{teacher forcing: append ground-truth token}
    \EndFor
    \State $\text{score} \gets \dfrac{\text{diff}}{|y|}$
    \State \Return $\text{diff},\ \text{score}$
  \end{algorithmic}
\end{algorithm}

\begin{algorithm}
  \caption{Sampled token verification via prefill}
  \label{alg:single-prefill-verif}
  \begin{algorithmic}[1]
    \Require input sequence $x$, target output $y = (y_1,\dots,y_T)$, seed $\seed$, model $\theta$
    \State $z \gets x \circ y$ \Comment{concatenate input and full target for caching}
    \State $\ell_{0:|z|-1} \gets \text{Prefill}(\theta, z)$ \Comment{single forward pass; next-token logits at each position}
    \State $\text{diff} \gets 0$
    \For{$i \gets 1$ to $T$} \Comment{predict $y_i$ from context $x \circ y_{<i}$}
      \State $p_i \gets \text{Softmax}\!\left(\ell_{|x|+i-1}\right)$
      \State $\hat{y}_i \gets \text{SampleFrom}(p_i; \RO(\seed\| i))$
      \If{$\hat{y}_i = y_i$}
        \State $\text{diff} \gets \text{diff} + 1$
      \EndIf
    \EndFor
    \State $\text{score} \gets \dfrac{\text{diff}}{T}$
    \State \Return $\text{diff},\ \text{score}$
  \end{algorithmic}
\end{algorithm}

\subsection{Rank-Aware \GLS~Scoring} \label{sec:rank-aware-GLS}

\begin{algorithm}
\caption{Per-Token Risk Classification via \GLS~and Rank}
\label{alg:rank-aware-GLS}
\begin{algorithmic}[1]
\Require token $i$ at position $t$, \GLS~threshold $\tau$, rank cutoff $R$
\State Compute $\mathrm{\GLS}_t(i)$ and $\mathrm{rank}\!\big(\mathrm{logits}_t(i)\big)$
\If{$\mathrm{\GLS}_t(i) \ge \tau$} 
  \State \Return \textsc{safe} (0)
\Else \Comment{$i \notin A_t(\tau)=\{j:\mathrm{\GLS}_t(j)\ge\tau\}$}
  \If{$\mathrm{rank}\!\big(\mathrm{logits}_t(i)\big) > R$}
    \State \Return \textsc{danger} (2)
  \Else
    \State \Return \textsc{suspicious} (1)
  \EndIf
\EndIf
\end{algorithmic}
\end{algorithm}

\section{Mechanism Design security solutions: Random Sampling and the Panopticon} \label{sec:sampling-rates}

\subsection{Mechanism Design}

A valuable observation in the design on a verification system is that we can exploit a classic "panopticon" mechanism design structure, proposed in the 18th century by Willey Reveley to monitor prisoners with a limited number of guards. In the original setting, although it is physically impossible for the single guard to observe all the inmates' cells at once, one can design a building to allow a guard to see the prisoners without the prisoners seeing the guard. The fact that the inmates cannot know when they are being watched motivates them to act as though they are all being watched at all times.

In the same vein, in environments with persistent identity (such as API access tied to verified accounts or Know-Your-Customer laws), one does not need to catch all possible exfiltrations if the expected cost of being caught exfiltrating model weights is sufficiently high. Though critically, such a system is predicated on a credible source of randomness, where the exfiltrators cannot predict what the auditor will audit, and also the exfiltrators cannot exfiltrate all of the data in a single action (analogous to all the prisoners escaping at once). Critically, we do note that required named-entities is not a necessary requirement for our scheme to work.

\subsection{Derivation of Required Sampling Rate} \label{sec:sampling-rate-derived}

The ``panopticon'' framing gives an institutional rationale for why random verification can deter exfiltration: even partial, unpredictable monitoring can induce compliance when the expected cost of detection is high. 

The key design question then becomes: \emph{given a limited verification budget, what sampling rate is sufficient to detect exfiltration attempts with high confidence?} 

This question connects mechanism design to classical statistical estimation. By modeling each message as a Bernoulli trial where a fraction $p$ of messages contain covertly encoded content, we can derive the required sampling rate to achieve a desired confidence level in detecting at least $k$ instances of suspicious activity, classically referred to as the \emph{inverse binomial problem}.

In this setting, we think of all the messages as a collection of samples $S$, and there is some fraction $p$ of messages that encode a steganographic message. Presuming that the verifier requires $k$ steganographically-encoded samples to detect an exfiltration attempt, the question is how many samples need to be drawn.

For some confidence level $c$ (e.g., 99\% confidence), we seek to solve
\[
\Pr(\text{at least } k \text{ encoded-samples in } n \text{ draws}) \geq c
\]

In other words,
\[
\Pr(X \geq k) \geq c, \quad \text{where } X \sim \text{Binomial}(n, p).
\]

\begin{itemize}[noitemsep]
    \item For $p = 1\%$, $c = 99\%$, $k = 1$, this is $n = 460$ samples.
    \item For $p = 1\%$, $c = 99\%$, $k = 100$, this is $n = 12{,}460$ samples.
\end{itemize}

\section{Trust assumptions for Sampled Output Verification}\label{sec:trust-assumptions}

\subsection{Assumption 1: The Verification Server is trusted to correctly execute its code}
In our proposed solution, it is important that the verification server correctly executes its code. When the verification code itself is untrusted or run in an untrusted environment, an adversary could potentially compromise and circumvent the policy enforcement.

In order to achieve this level trust, we list a few possible properties of the Verification Server that might support this assumption:
\begin{itemize}[noitemsep]
    \item \textbf{Trusted compute:} It can use GPU TEEs or even CPUs to verify inference—even if doing so is inefficient.
    \item \textbf{Verifiable, simple code}: Its code does not need to be highly optimized, and any changes to it can be closely monitored.
    \item \textbf{Independent root of trust:} Attackers will now need to compromise the Verification Server in addition to the Inference Server.
    One could imagine spinning up \textbf{several independent Verification Servers} operated by independent teams of engineers/agents and running in mutually isolated environments.
    
\end{itemize}

\subsection{Assumption 2: All messages passed between users and the Inference Server are correctly logged and securely stored}

If an attacker could ensure that some inference outputs transmitted out of the data center are \textbf{not logged}, then it could evade the detection system by smuggling model parameters out via these unmonitored channels. Similarly, the attacker could ensure that inference outputs are \textbf{logged incorrectly}—e.g. that the logs contain outputs that the model would have actually generated, even though what was actually transmitted out of the data center was something else entirely.

Similarly, if an attacker can delete or alter the logged data, then the Verification Server's code may not evaluate the correct logs. In this world, even if the Verification Server is itself trusted to execute code correctly (Assumption \#1), it would not be able to detect model exfiltration.

However, if the logs are correctly written in the first place, then this attack is only possible if the attacker can modify the logs after they've been written.

One idea would be to secure inference records and model artifacts against tampering using \textbf{standard cryptographic commitment schemes}. For instance, the developer can maintain a Merkle tree with leaves for each inference record (an \textbf{``inference ledger''}), and another with leaves for each model checkpoint (a \textbf{``model ledger''}). Whenever the Inference Server loads inference inputs, outputs, and models into memory, it hashes them and checks whether they are included in the relevant ledgers.

Some notes on how this might work in practice:
\begin{itemize}[noitemsep]
    \item The hashes can be signed using a \textbf{private key only available on the trusted logging device}
    \item Any systems that reference entries in the model ledger or inference ledger should store local records of any hashes they've encountered—periodically they could then check whether these historical hashes still exist in the relevant ledgers. This means that \textbf{attacking the ledgers would require compromising many independent systems simultaneously}
\end{itemize}

\subsection{Assumption 3: The attacker can neither predict nor alter which logs the Verification Server will sample to verify}

It is critical that an adversary attempting to exfiltrate weights does not know which outputs are going to be audited. Otherwise, an adversary may simply exfiltrate weights through the messages not-audited. This assumption is critical for this scheme to work; however, many schemes for random sampling exist. A simple proposal like random sampling within a trusted-execution environment is almost certainly sufficient; but in the event that it isn't this has been studied in \href{https://a16zcrypto.com/posts/article/leader-election-from-randomness-beacons-and-other-strategies/}{great depth} in blockchain literature. We discuss how to do this in detail in Section \ref{sec:how-to-randomly-sample}.

\subsection{Assumption 4: The Verification Server itself does not pose a risk of model exfiltration}

We assume that the Verification Server is not itself an attack vector. One approach would be to limit the outbound bandwidth such that rapid model exfiltration is infeasible. Note that the required outbound bandwidth from the Verification Server is likely qualitatively lower than that required by the Inference Server, rendering egress controls a more natural defense in this case. For example, all traffic beyond the policy decision could suppressed, which would heavily rate-limit any potential exfiltration even when an adversary had control of the Verification Server.

\subsection{Assumption 5: The randomness used for inference is trusted and non-malleable}
\label{assump:trusted-rng}

If the adversary can manipulate randomness, they can transmit information by selecting seeds/samples while leaving per-step likelihoods within normal ranges. In this regime the verifier’s guarantees degrade sharply: many exfiltration strategies become indistinguishable from benign stochasticity, and a malicious provider could even watermark or backdoor outputs via RNG tampering while remaining white-box undetectable \cite{goldwasser2024plantingundetectablebackdoorsmachine,pseudo_random_error_correcting}.

\paragraph{Statement.}
The random coins used by the sampler during inference—and the identical coins used by the verifier—are generated by a trusted cryptographic RNG, bound to the request, and cannot be chosen, biased, or adaptively influenced by the adversary.

\paragraph{Rationale.}
If an adversary can steer the randomness (e.g., select seeds, bias PRNG state, or grind over retries), the sampler itself becomes a covert channel: the adversary can encode arbitrary payloads by “choosing the random samples.” Closely related phenomena are known in training-time settings, where tampering with algorithmic randomness enables \emph{undetectable backdoors} that are indistinguishable to any efficient observer \cite{goldwasser2024plantingundetectablebackdoorsmachine}. Analogously, inference-time coin-fixing enables steganographic signaling without changing logits in a way a verifier could reliably flag (see also pseudo-random watermarking constructions \cite{pseudo_random_error_correcting}).

\paragraph{Threat model scope.}
Under this assumption, the adversary:
(i) cannot specify or influence the inference seed;
(ii) cannot bias randomness indirectly via user-exposed knobs (e.g., temperature, top-$k$, top-$p$) beyond an attested, fixed configuration for the session; and
(iii) cannot “grind” seeds by adaptively restarting until a desired draw appears (retry budgets are bounded and auditable).

\subsection{Message-Sampling Secure Seed Expansion} \label{sec:how-to-randomly-sample}

A critical part of deployment to prevent adversarial exfiltration is proper seed management. If the adversary is able to predict or precompute which tokens will be selected during generation,
they may be able to choose seeds and prompt combinations carefully to maximize the exfiltration
potential. To be sure that this isn't possible, it's important to at the very least:
\begin{enumerate}[noitemsep]
    \item Use a sufficiently long seed to prevent precomputing and exhaustive searches
    for good seeds.
    \item Use a random key and cryptographically secure hash function to expand the seed. More
    specifically, one can expand the seed $\sigma$ by computing $H(K,\sigma)$, where $H$ is
    a secure hash, e.g. SHA3, and $K$ is a random key.
\end{enumerate}

Beyond the above, we outline further measures that one can take in
Section~\ref{sec:info-theory-bounds}. These additional constraints can provably limit the
exfiltration of adversaries but may come at some inconvenient costs to the system architecture.
Specifically, one may choose to also:
\begin{enumerate}[noitemsep]
    \item Expand the seed adaptively, based on tokens chosen previously. So, computing the 
    randomness for the next token as $H(K,\sigma,x)$, where $x$ is the prompt and tokens computed 
    thus far (used in Theorems~\ref{thm:no-seed-interactive-bound} and~\ref{thm:no-seed-bound}).
    \item Expand initial seed into much longer seed in trusted system, to prevent seed
    exfiltration. So, an API server may compute a large $\sigma' \gets H(K,\sigma)$ that is given
    to the inference server as  and then continue
    further expansion by running $H(K,\sigma',x)$ (used in Theorem~\ref{thm:no-seed-bound}).
    \item Keep the expansion key in a trusted system and require interaction to get the
    randomness for the next token (used in Theorem~\ref{thm:no-seed-interactive-bound}).
\end{enumerate}
Although these measures do require careful implementation at the system level, they can provide
significant returns in protecting the theoretical exfiltration rate. And, although cryptographic
operations and interaction are often expensive in real-world deployments, we note that the
additional cost of these operations does not bottleneck performance compared to the inference
time of generative AI models.

\subsection{Plausibility of Necessary Assumptions}\label{sec:plausibility-of-bounds}
Theorem~\ref{thm:perplexity-bound}, and to a lesser extent
Theorems~\ref{thm:no-seed-interactive-bound} and~\ref{thm:no-seed-bound},
give the adversary exfiltrating data a significant amount of power.
For example, we assume both the sender and receiver have access to the model
distribution for the given prompt and sometimes even the distribution of honest
responses accounting for valid non-determinism.

Such assumptions may appear too strong at first glance. How likely is it that
the adversary is able to obtain all of this information? As with many areas of
cryptography, we aim to overestimate the power of the adversary in order to
obtain \emph{stronger} bounds on realistic adversarial capability. That being
said, we argue that these assumptions are achieved by some very plausible
adversaries:
\begin{itemize}
    \item An adversary may have already exfiltrated a small model was already 
    exfiltrated that is accessible through a system's API, and the adversarial 
    server is now exfiltrating a larger model through queries to the smaller 
    model.
    \item The adversary may have distilled a small model on the outputs of a 
    large model for specific prompts or domains. Then, by exfiltrating only on 
    that set of prompts, the receiver will have an extremely close probability 
    distribution to the large model.
\end{itemize}
In Theorem~\ref{thm:perplexity-bound} we also assume the outside adversary obtains the generation seed. This way we can prove a bound without
making as many trust and cryptographic assumptions. Also, in many real-world
deployments, one may want the flexibility to modify seed generation
to either speed up or simplify the infrastructure and avoid deployed extremely secure seed expansion (although this cost is much smaller than the cost of inference).

\section{How Specific Implementations may Affect Verification} 

The exact inference stack used dramatically affects the requirements on the verification server. In this work we formulate sampling from an LLM as a two step process: (i) generate a probability distribution, and (ii) sample a single token from that probability distribution. The majority of this work is showing how to quantify the non-determinism that may arise from slightly different probability distribution generating processes. However, for this to work, this is predicated on a precise and easy-to-reproduce method of sampling from the probability distribution. In this work we show how to do this for the Inverse Probability Transform method of sampling and the Gumbel-Max trick.

Modern inference stacks are generally more complicated than either of these two methods, employing tricks like speculative decoding, and this needs to be shared across the inference server and verification server, which we discuss in Section \ref{sec:speculative-decoding-verification}.

Another noteworthy consideration is how random samplers update their state, and synchronizing this across the inference and verification servers. For example, in some implementations (like Ollama \cite{ollama_docs}),
 if you filter the probability distribution with top-p and the token array ends up with only one entry, the inference server will not sample at all, and use a greedy decoding step. This means that the RNG state does not advance for the inference server, but may for the verification server. However, this issue does not arise in other common implementations like vLLM \cite{vllm}.

\subsection{Speculative Decoding}
\label{sec:speculative-decoding-verification}

Speculative decoding \cite{leviathan2023fastinferencetransformersspeculative, chen2023acceleratinglargelanguagemodel, zhang2025learningharmonizedrepresentationsspeculative, li2025eagle3scalinginferenceacceleration} has emerged as a widely-adopted technique to accelerate LLM inference by using a smaller ``draft'' model to propose candidate tokens, which are then verified in parallel by the larger ``target'' model. This approach can achieve substantial speedups (often $1.5{-}2.5\times$) while preserving the exact output distribution of standard autoregressive sampling from the target model.

In the main body of this work, we focus our verification methodology on standard autoregressive sampling (Section~\ref{sec:implement-verif-server}), where a single model generates one token at a time. This choice reflects two practical considerations: 
\begin{enumerate}[topsep=1pt,itemsep=1pt]
    \item Production systems exhibit heterogeneity in their acceleration strategies: while speculative decoding is common, implementations vary significantly across frameworks and deployments.
    \item At the time of writing, vLLM version 0.10 does not natively support speculative decoding.
\end{enumerate}

Nevertheless, we recognize that speculative decoding represents an important class of real-world inference strategies. In this appendix, we demonstrate that our verification framework naturally extends to handle speculative decoding with minimal modifications. 

\paragraph{Background: A Speculative Decoding Algorithm}

In speculative decoding introduced by  \cite{leviathan2023fastinferencetransformersspeculative, chen2023acceleratinglargelanguagemodel}, a draft model generates $K$ candidate tokens, which the target model then verifies in parallel. For each candidate, the target model either accepts it, rejects and resamples from an adjusted distribution, or (if all candidates are accepted) samples one additional bonus token. This process maintains exact distributional equivalence to standard autoregressive sampling from the target model alone.

Different speculative decoding algorithms use different methods for accepting tokens from the draft model. For example, in the original work proposing speculative decoding \cite{leviathan2023fastinferencetransformersspeculative} the verifier only accepts a token if it is the exact token that the verifier model would have generated in that position. However, in other methods a draft token is accepted not only if it is the token the verifier would have sampled, but alternatively if it is a sufficiently likely token as well  \cite{li2025eagle3scalinginferenceacceleration}.

Depending on the specific method used for speculative decoding the verification process may need to be modified. In some cases this can cause a larger rate of exfiltratable information, as more tokens \textit{could} have been sampled by the drafter; however, future work could also verify the draft model process in the speculative decoding process as well.

\section{Response to Detecting Exfiltration Attempts}\label{sec:policy-response-to-exfil}

This work focuses primarily on the \emph{detection} of potential model-weight exfiltration attempts rather than on post-detection remediation. In particular, we formalize the task of determining whether a batch of model outputs contains covertly embedded information. Following the literature on steganographic detection \cite{batch_stego_2007, batch_stego_real_world_2012}, we envision verification pipelines adhering to a general three-step structure:

\begin{enumerate}[noitemsep]
    \item \textbf{Per-token scoring:} Assign a suspiciousness score to each token using an appropriate statistical or model-based function (e.g., likelihood ratios, cross-model divergence, or reconstruction residuals).
    \item \textbf{Pooling:} Aggregate token-level scores into sequence-, user-, or deployment-level statistics via a chosen pooling function (e.g., mean, max, or robust thresholding).
    \item \textbf{Flagging:} Identify batches or entities whose pooled scores exceed a defined operational threshold, signaling possible exfiltration or anomalous behavior.
\end{enumerate}

This framework deliberately mirrors that of batch steganalysis, where the goal is not to prove intent but to highlight objects that merit investigation. Detection signals should thus be treated as probabilistic indicators rather than definitive evidence of compromise.

\vspace{0.5em}
\noindent\textbf{Operational response.} Once a token batch or deployment is flagged, the appropriate policy response depends on the severity, frequency, and operational sensitivity of the detection. A proportionate and auditable escalation pathway might include the following steps:

\begin{enumerate}[noitemsep]
    \item \textbf{Rate-limit or reduce per-user bitrate.} Temporarily lower allowable output volumes or token-generation rates for implicated users or systems, thereby constraining potential exfiltration bandwidth while maintaining partial service availability. This should be paired with \emph{Know Your Customer} (KYC) or identity-verification mechanisms to mitigate Sybil-style circumvention (\url{https://en.wikipedia.org/wiki/Sybil_attack}).
    \item \textbf{Quarantine and investigate machines or deployments.} Isolate hosts, containers, or code branches that repeatedly produce high-suspiciousness outputs. Conduct forensic review of environment variables, code diffs, and access logs to detect injected payloads or altered sampling procedures.
    \item \textbf{Human-in-the-loop review.} Contact or interview users responsible for flagged activity, allowing for legitimate use cases to be surfaced and appeals to be made. Human review is particularly valuable when benign but novel workloads (e.g., creative generation or low-frequency topics) may trigger false positives.
    \item \textbf{Incident escalation and attribution.} If multiple correlated detections arise across distinct systems, escalate to incident response and coordinate attribution analysis. Cross-reference event metadata (user, GPU ID, deployment hash) to localize possible compromise points.
\end{enumerate}

\subsection{Privacy-preserving data analysis.}\label{sec:privacy-preserving-accrual}

We note that one natural concern in all of this is that any security event will invite privacy-invading scrutiny by the API provider. We note that it is still possible to produce privacy-preserving techniques for triaging potential exfiltration. For example, providers could aggregate suspicious-token counts over a short rolling window using a hash over user values. For example, a pseudonym for a user or API key can be derived via a keyed hash that rotates each window; counts are kept per pseudonym and a heavier review is prioritized when a policy-defined bucket size $B$ is exceeded.

\section{\GLS~Score Statistics}\label{sec:gumbel-score-statistics}

This section presents the empirical distribution of \GLS~scores to provide intuition for how the scoring function behaves on benign traffic and how detection thresholds are calibrated.

\subsection{Score Distribution for a Single Token Position}

Figure~\ref{fig:gumbel-score-statistics} shows the distribution of \GLS~scores across
the entire vocabulary for a single token position during generation. For each token $i$
in the vocabulary, we compute $\mathrm{\GLS}_t(i)$ using the \GLSestimator~estimation procedure described in Algorithm~\ref{sec:GLS}.

\begin{figure}[ht]
    \centering
    \includegraphics[width=0.65\linewidth]{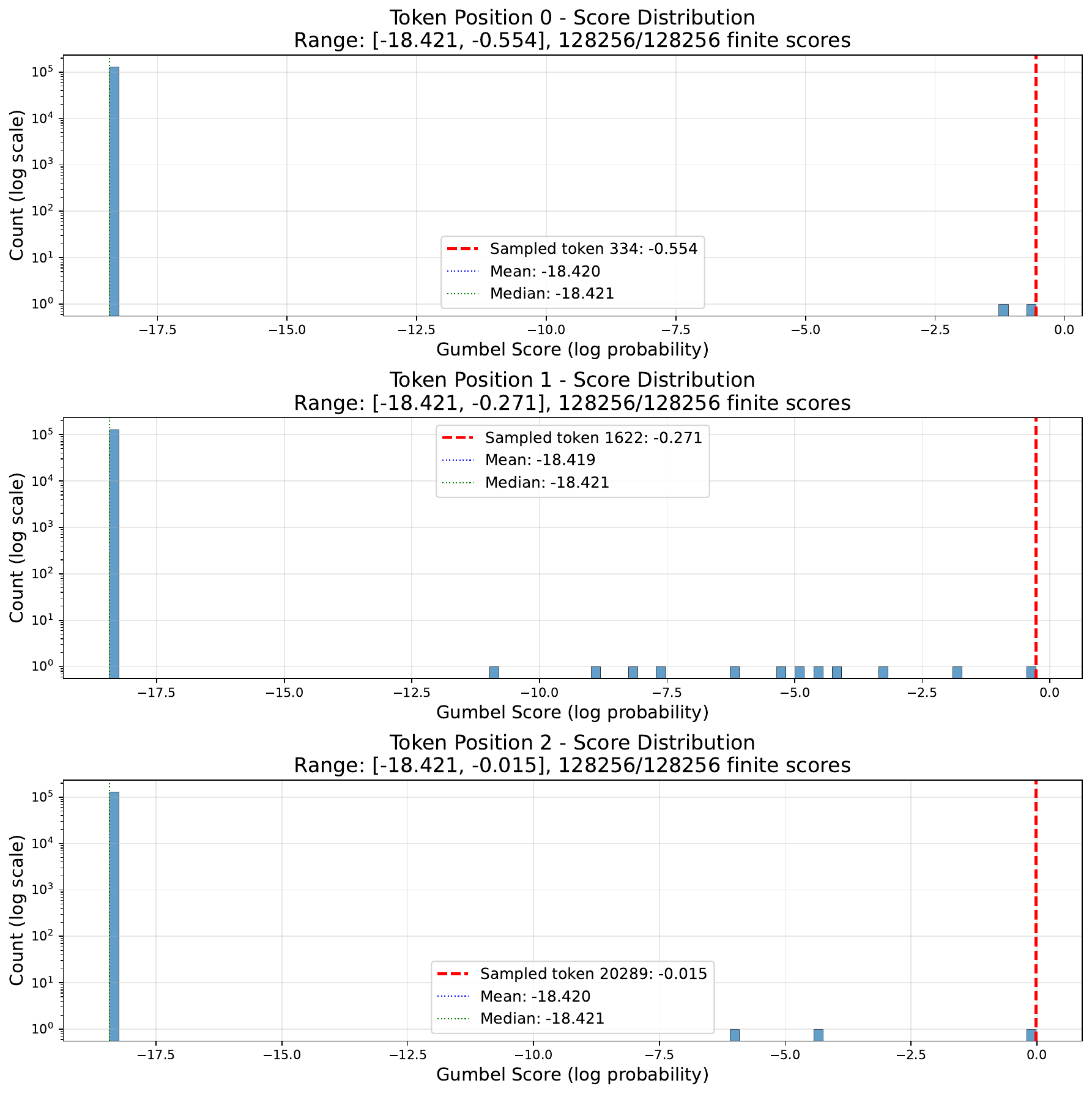}
    \caption{\textbf{Distribution of \GLS~scores across vocabulary for a single position.}
    This histogram shows how \GLS~scores are distributed across all possible tokens at
    a given generation step. The actually-sampled token typically has a score near the
    peak of this distribution (high likelihood region).}
    \label{fig:gumbel-score-statistics}
\end{figure}

\subsection{Distribution of \GLS~Scores for Actually-Generated Tokens}

While Figure~\ref{fig:gumbel-score-statistics} shows scores for all vocabulary tokens
at one position, Figure~\ref{fig:correct-token-GLS-scores} shows the distribution of
\GLS~scores for the tokens that were \emph{actually generated} across many positions
and prompts.

This distribution is crucial for threshold calibration: we generate 500 prompts with
up to 500 tokens each under \emph{benign} conditions (honest model, no adversary),
compute $\mathrm{\GLS}_t(y_t)$ for each actually-generated token $y_t$, and collect
the resulting scores. The histogram in Figure~\ref{fig:correct-token-GLS-scores}
represents the empirical distribution of \GLS~scores under legitimate operation.

\begin{figure}[ht]
    \centering
    \includegraphics[width=0.65\linewidth]{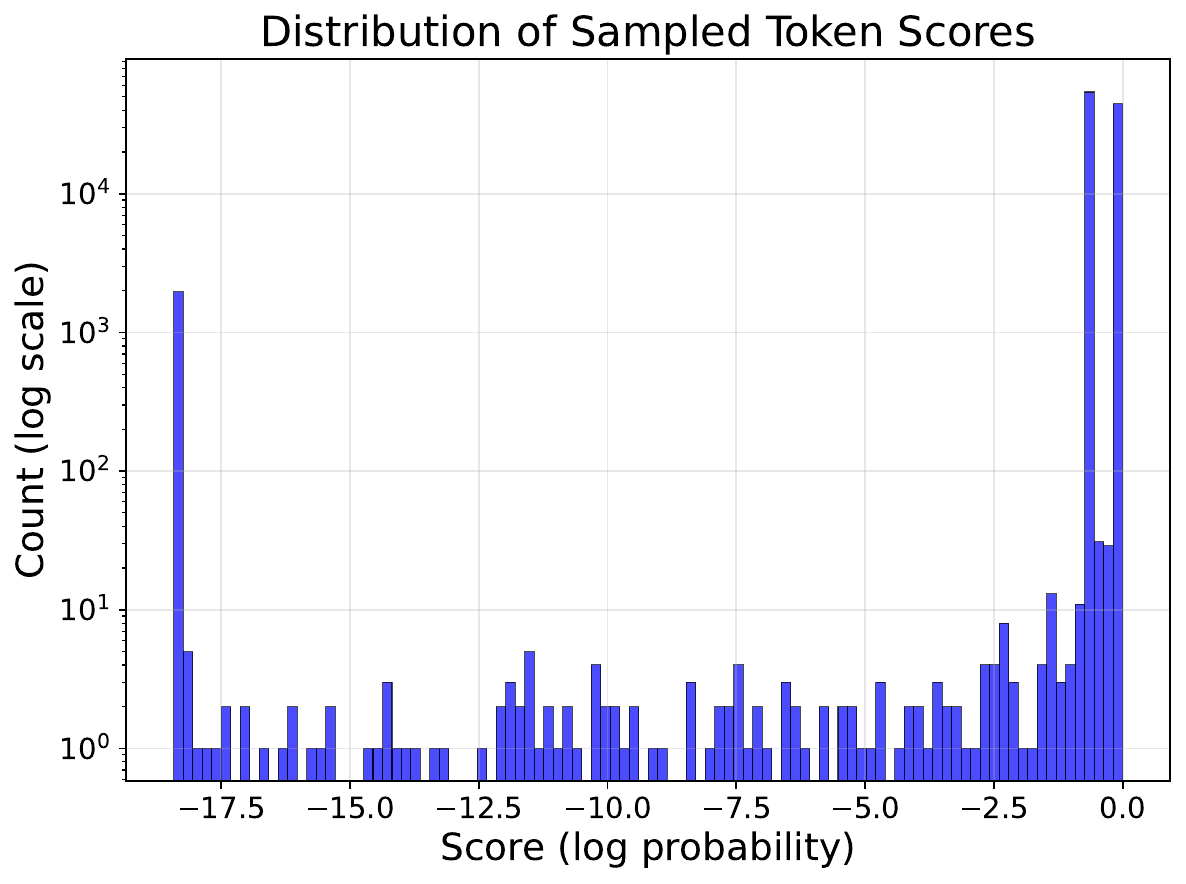}
    \caption{\textbf{Histogram of \GLS~scores for actually-generated tokens across
    500 prompts.} Each bin represents how many tokens (out of $\approx$250,000 total)
    received a given \GLS~score during benign generation. This distribution allows us
    to select thresholds $\tau$ that achieve target false-positive rates: e.g., setting
    $\tau$ at the 1st percentile of this distribution yields FPR $\approx 1\%$.}
    \label{fig:correct-token-GLS-scores}
\end{figure}

\subsection{Threshold Selection and FPR-Information Trade-off}

Given the empirical score distribution from benign traffic, we select thresholds to
achieve specific false-positive rates. Figure~\ref{fig:threshold-fpr-tradeoff}
illustrates this trade-off:

\begin{itemize}[noitemsep]
    \item \textbf{X-axis (threshold $\tau$):} More negative thresholds accept more
    tokens (lower FPR, higher usability) but allow adversaries to exfiltrate more
    information. More positive thresholds reject more tokens (higher FPR, lower
    usability) but tighten the information bound.

    \item \textbf{Left Y-axis (False Positive Rate):} The fraction of benign tokens
    with $\mathrm{\GLS}_t(y_t) < \tau$. Computed directly from the cumulative distribution
    of Figure~\ref{fig:correct-token-GLS-scores}.

    \item \textbf{Right Y-axis (Extractable Information):} The worst-case per-token
    information an adversary can exfiltrate, computed as
    $\frac{1}{T}\sum_t \log_2|A_t(\tau)|$ where $A_t(\tau)$ is the admissible set
    (tokens with $\mathrm{\GLS}_t(i) \ge \tau$).
\end{itemize}

Practitioners can use this curve to select operating points: for example, tolerating
0.1\% FPR might reduce extractable information to $<$1 bit/token (versus the baseline
$\log_2(|\text{vocab}|) \approx 16$ bits/token).

\begin{figure}[ht]
    \centering
    \includegraphics[width=0.65\linewidth]{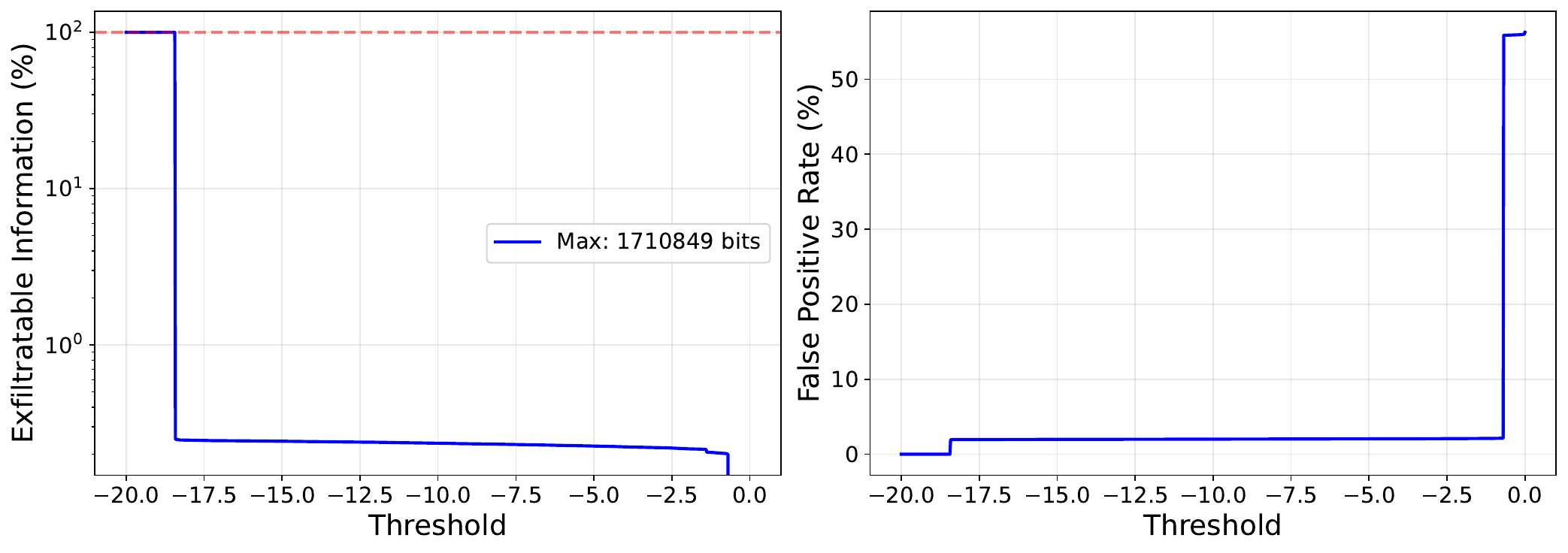}
    \caption{\textbf{Threshold selection trade-off.} As the \GLS~threshold $\tau$
    increases (moving right), the false positive rate increases (rejecting more benign
    tokens), while the extractable information decreases (tightening the adversarial
    capacity bound). This plot allows practitioners to select thresholds that balance
    security and usability for their deployment.}
    \label{fig:threshold-fpr-tradeoff}
\end{figure}

\section{Gumbel-Likelihood Score Exfiltration Results} \label{sec:exfil-results}

This section presents results from different experimental configurations of the two-tier and three-tier verification classifier systems across various models and rank thresholds.

\begin{figure}[H]
    \centering
    \includegraphics[width=0.65\linewidth]{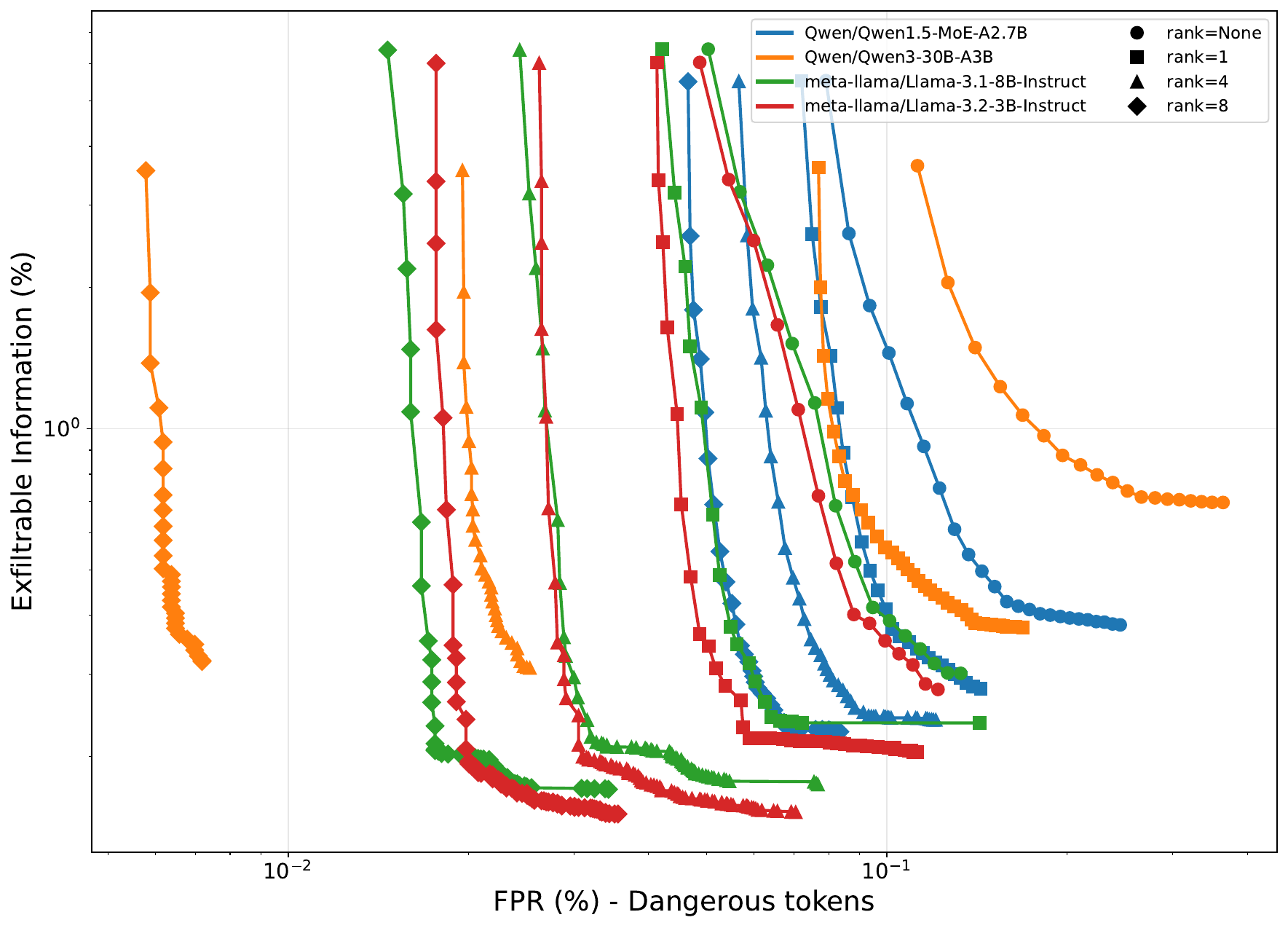}
    \caption{\textbf{Multi-model two-tier classifier Pareto frontier (envelope across all rank thresholds).} This plot displays the optimal security-vs-usability trade-off for detecting steganographic data exfiltration across multiple LLM architectures. Each curve represents a different model, showing the envelope across all tested rank thresholds, i.e., for each FPR level, it shows the minimum extractable information achievable by choosing the optimal rank threshold. The Pareto frontier shows the best achievable trade-off between (x-axis) False Positive Rate (\%) - fraction of legitimate tokens incorrectly flagged as ``dangerous'' and (y-axis) Extractable Information (\%) - how much data an attacker can still exfiltrate while evading detection. Models tested: Llama-3.1-8B, Llama-3.2-3B, Qwen3-30B-A3B, and Qwen1.5-MoE-A2.7B.}
    \label{fig:multi-model-multi-rank-pareto}
\end{figure}

\subsection{Logit-Rank Ablations}\label{sec:rank-ablations}

In this section we include experiments with the same experimental set up as Section \ref{sec:empirical}, where we sweep our classifier across multiple logit-rank thresholds to show the effect of logit-rank as a meaningful contribution to classifying on exfiltratable information.

\begin{figure}[H]
    \centering
    \includegraphics[width=0.65\linewidth]{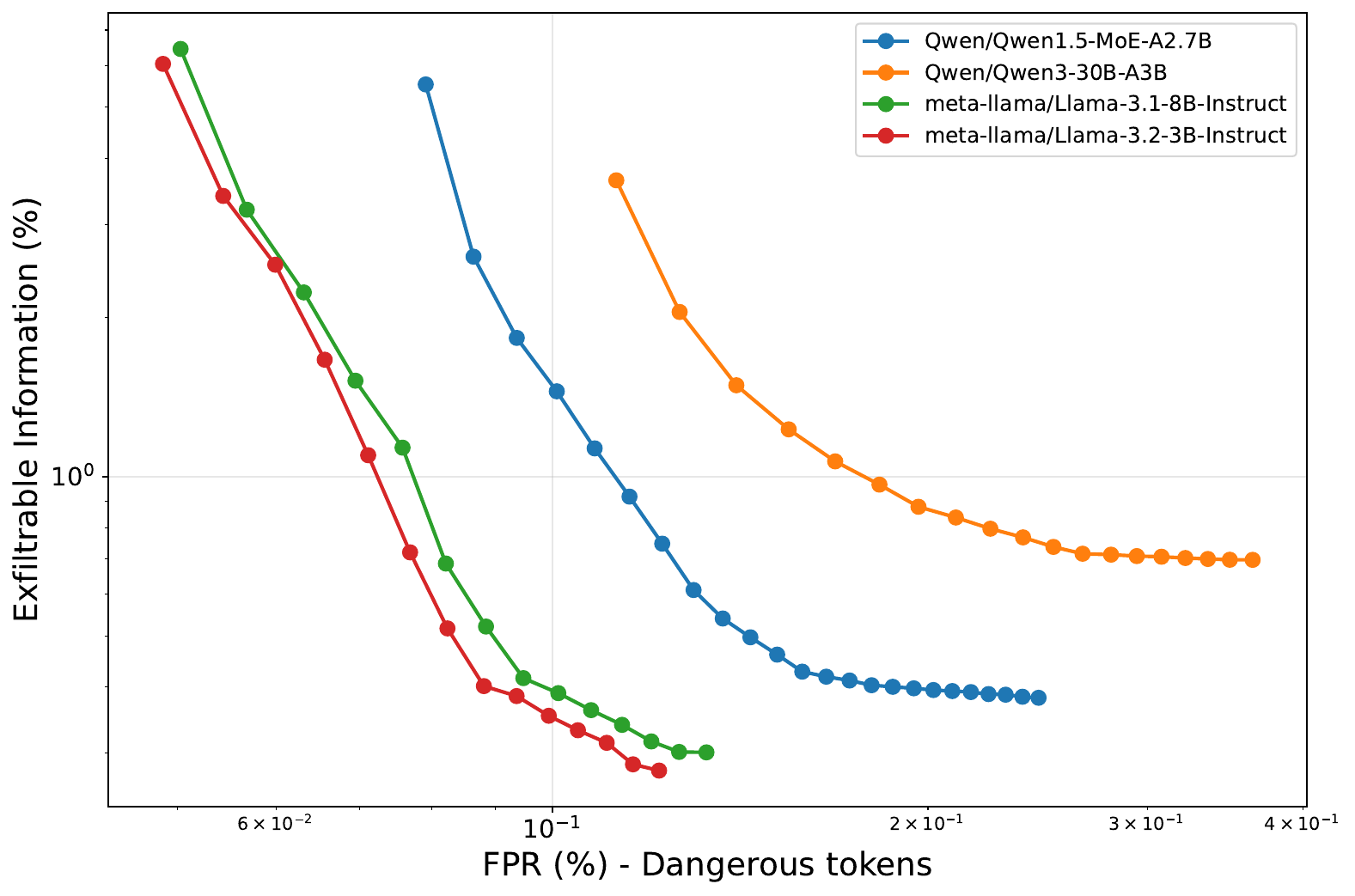}
    \caption{\textbf{Multi-model Pareto frontier with no rank threshold (two-tier baseline).} This plot shows the baseline two-tier classifier performance (safe vs. dangerous only, no suspicious category) across all tested models. Without the rank-based stratification, the system can only distinguish between tokens that pass or fail the \GLS~threshold, leading to higher extractable information at comparable FPR levels compared to the three-tier classifier.}
    \label{fig:multi-model-rank-none}
\end{figure}

\begin{figure}[H]
    \centering
    \includegraphics[width=0.65\linewidth]{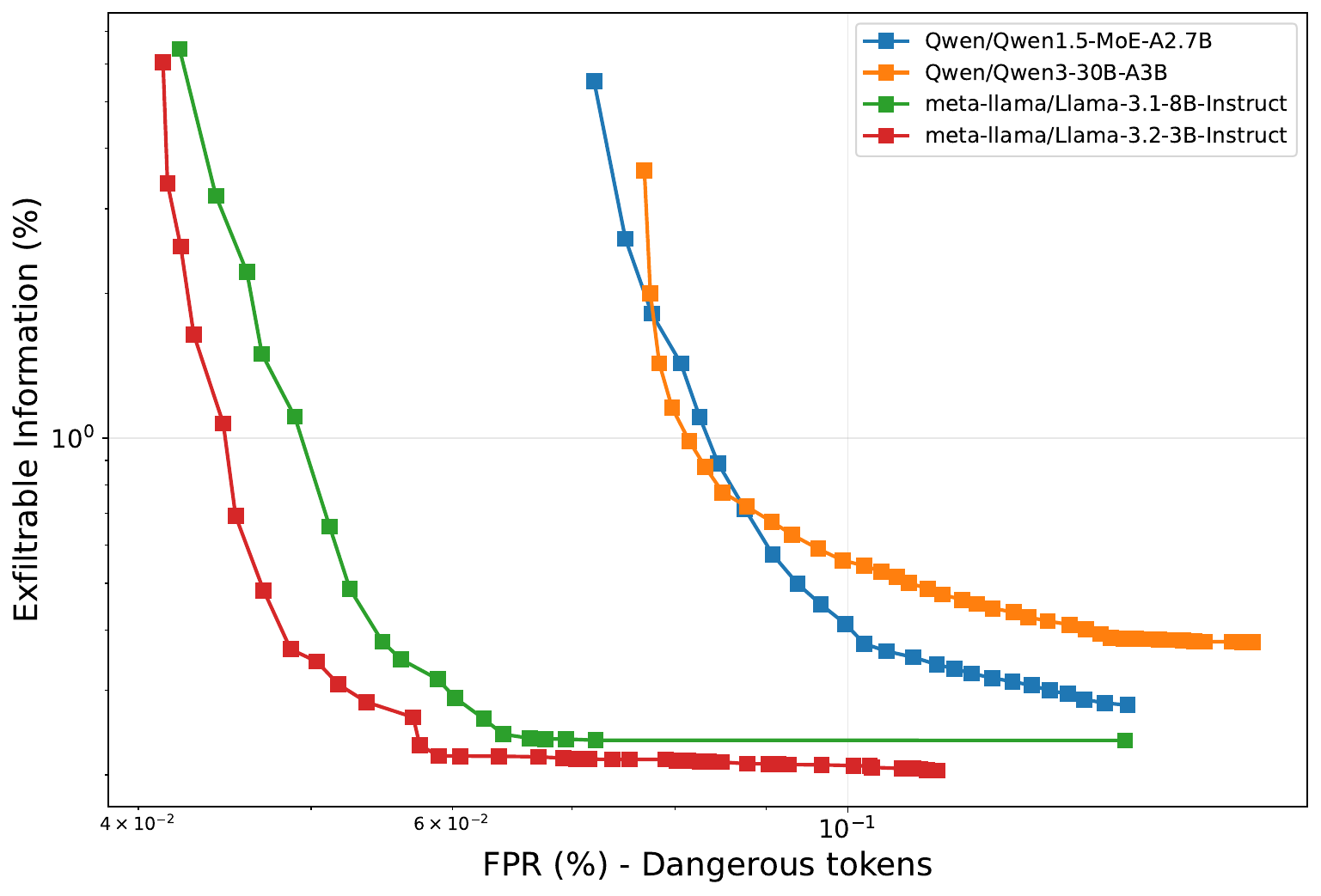}
    \caption{\textbf{Multi-model Pareto frontier with rank threshold = 1.} This configuration uses the most aggressive rank cutoff, where only the top-ranked token by raw logits is allowed in the suspicious category. All tokens that fail \GLS~and are not rank-1 are classified as dangerous. This provides the tightest security bounds but may increase the FPR for legitimate low-rank tokens.}
    \label{fig:multi-model-rank1}
\end{figure}

\begin{figure}[H]
    \centering
    \includegraphics[width=0.65\linewidth]{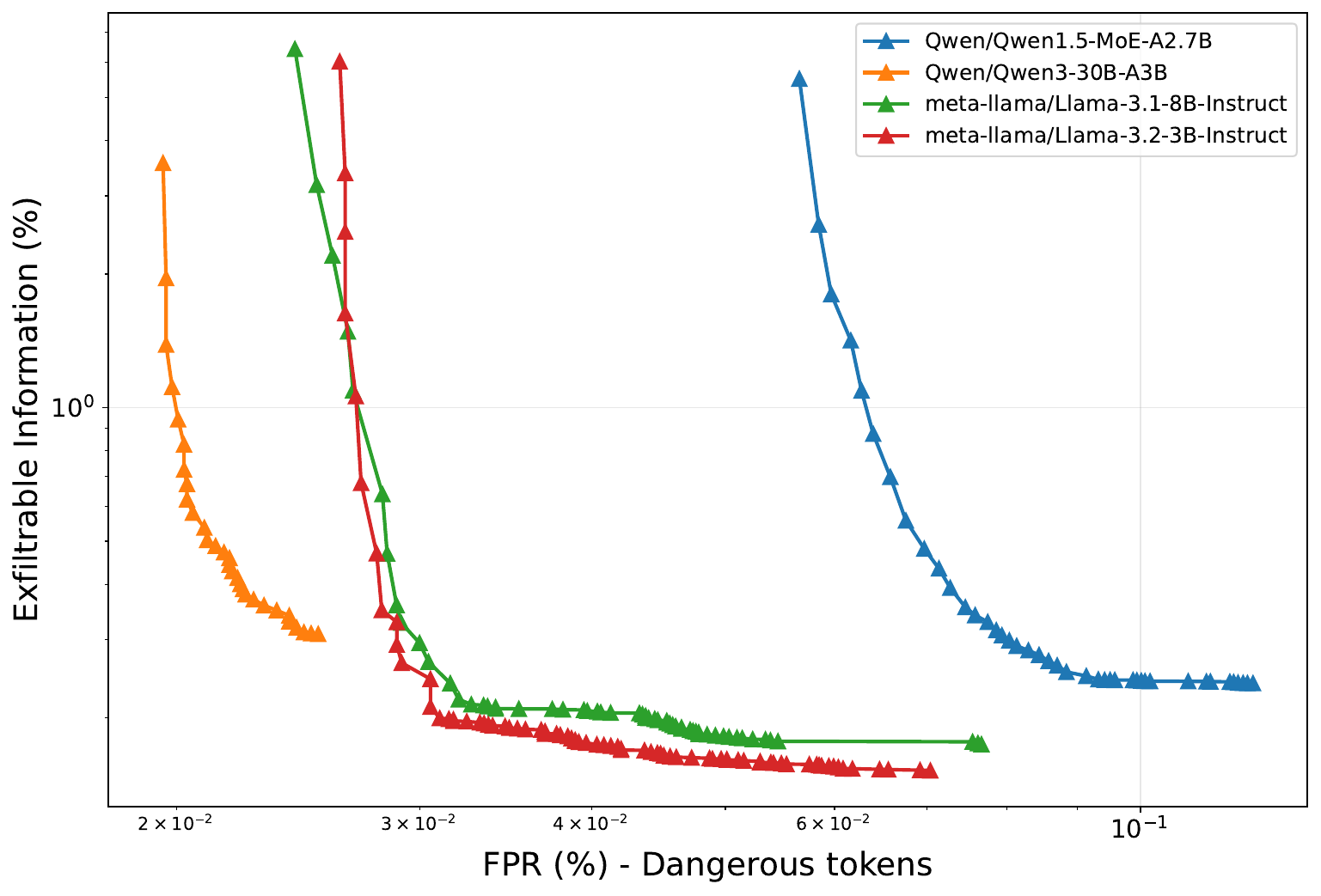}
    \caption{\textbf{Multi-model Pareto frontier with rank threshold = 4.} With a rank cutoff of 4, tokens that fail \GLS~but fall within the top 4 by raw logits are classified as suspicious (limited to $\log_2(4) = 2$ bits of information), while those beyond rank 4 are dangerous. This provides a balance between security and usability.}
    \label{fig:multi-model-rank4}
\end{figure}

\begin{figure}[H]
    \centering
    \includegraphics[width=0.65\linewidth]{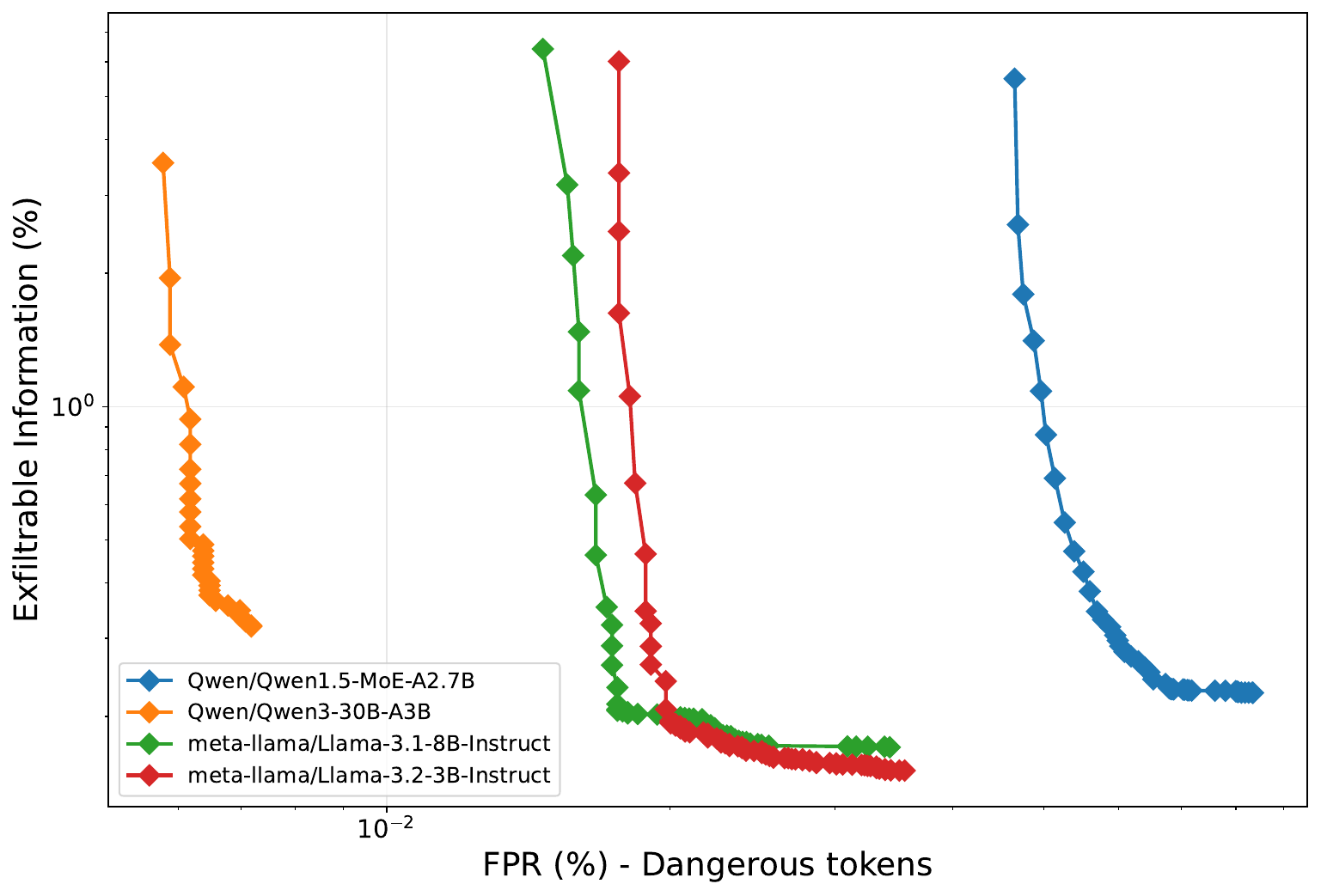}
    \caption{\textbf{Multi-model Pareto frontier with rank threshold = 8.} With rank = 8, suspicious tokens are capped at $\log_2(8) = 3$ bits of extractable information. This intermediate threshold allows more flexibility while still constraining the attacker's channel capacity significantly below the full vocabulary size.}
    \label{fig:multi-model-rank8}
\end{figure}

\begin{figure}[H]
    \centering
    \includegraphics[width=0.65\linewidth]{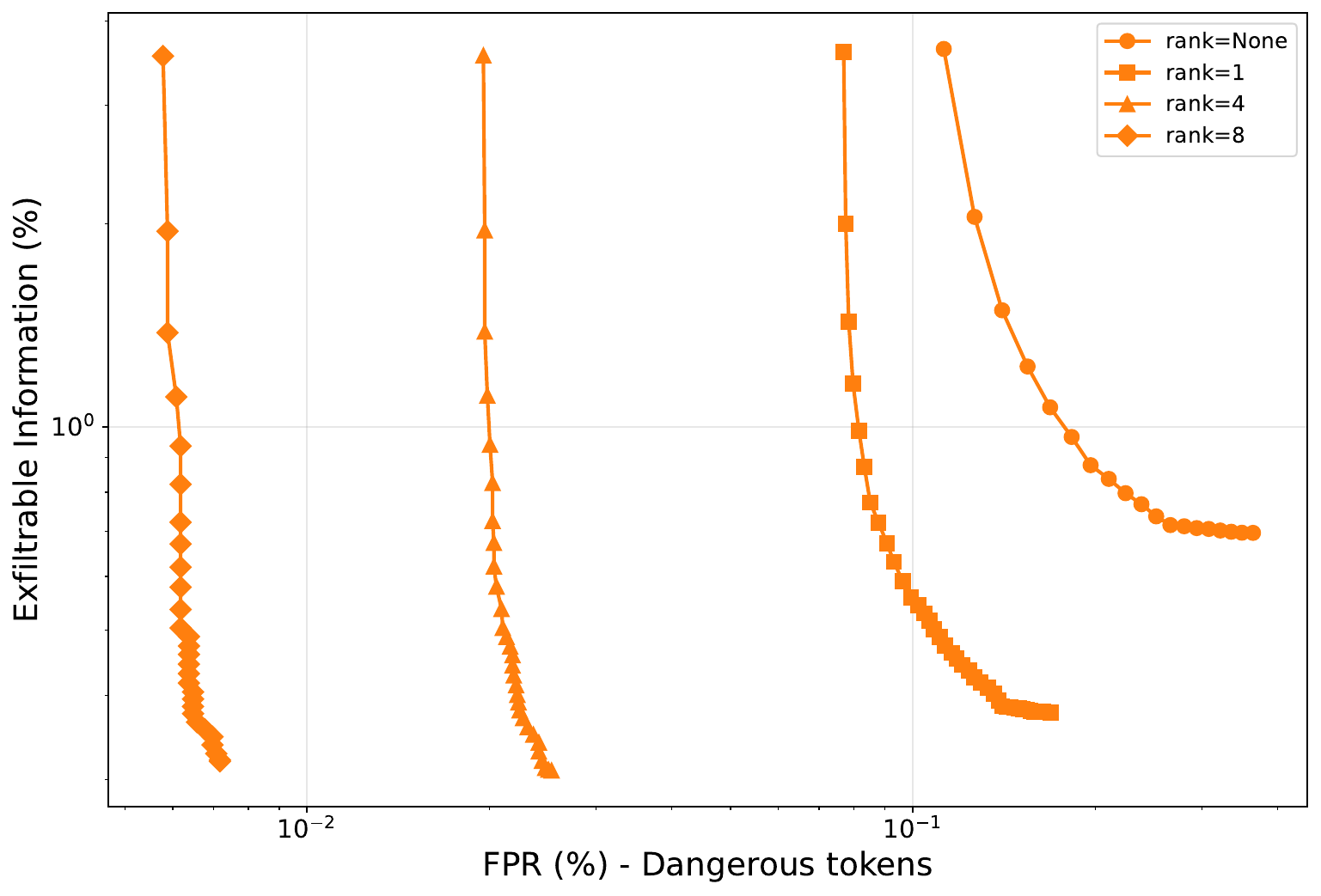}
    \caption{\textbf{Qwen3-30B-A3B Pareto frontier across rank thresholds.} This plot shows the extractable information vs. FPR trade-off for the Qwen3-30B-A3B model across different rank threshold configurations. Each curve represents a different rank cutoff, showing how the three-tier classification scheme performs for this larger model. The Pareto-optimal curves demonstrate that finite rank thresholds consistently outperform the baseline (rank=None) configuration.}
    \label{fig:qwen-30b-single-model}
\end{figure}

\begin{figure}[H]
    \centering
    \includegraphics[width=0.65\linewidth]{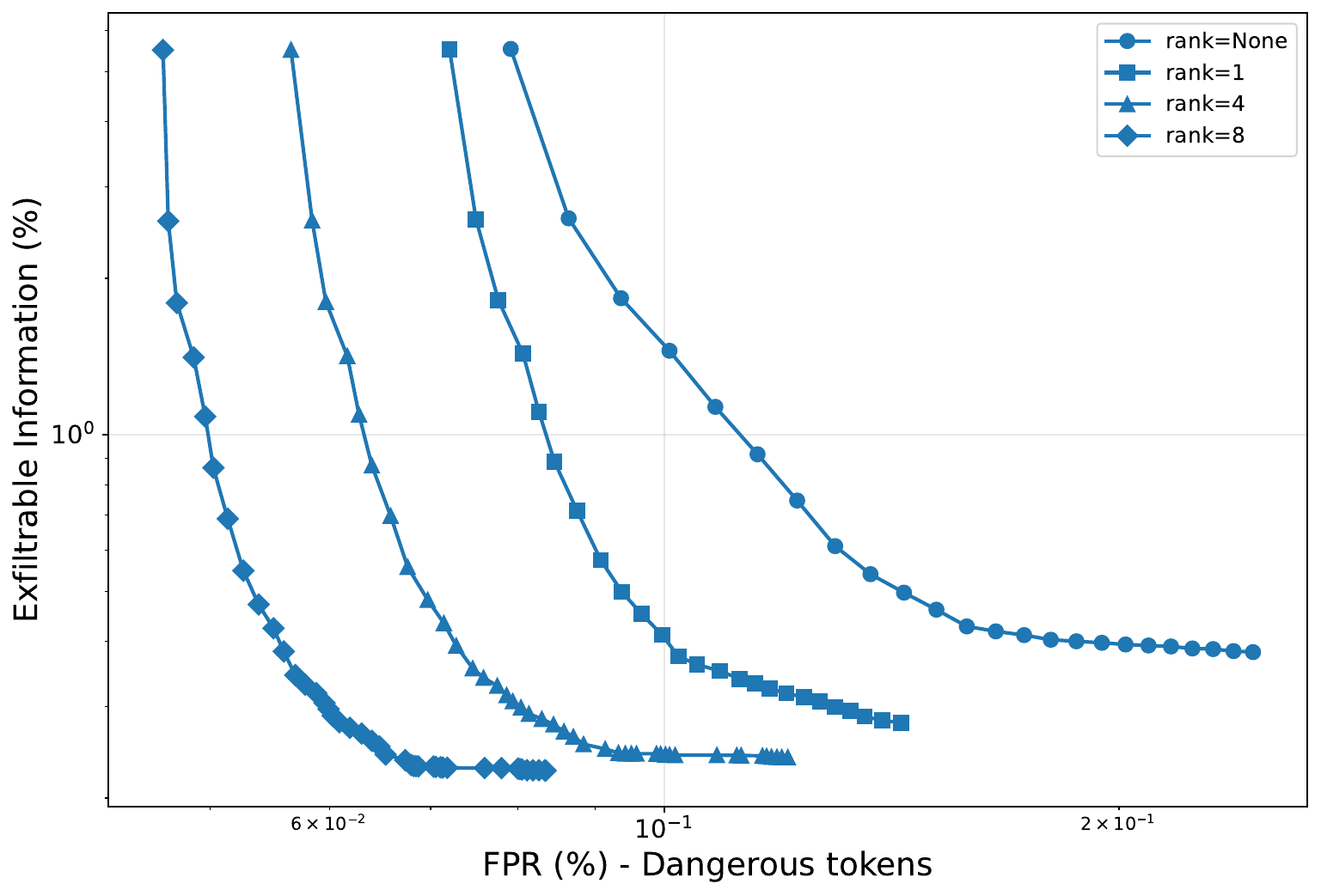}
    \caption{\textbf{Qwen1.5-MoE-A2.7B Pareto frontier across rank thresholds.} This plot displays the trade-off curves for the Qwen1.5-MoE mixture-of-experts model. Despite having fewer active parameters (2.7B), the MoE architecture shows distinct verification characteristics compared to dense models, providing insights into how model architecture affects the \GLS~verification effectiveness.}
    \label{fig:qwen-moe-single-model}
\end{figure}

\begin{figure}[H]
    \centering
    \includegraphics[width=0.65\linewidth]{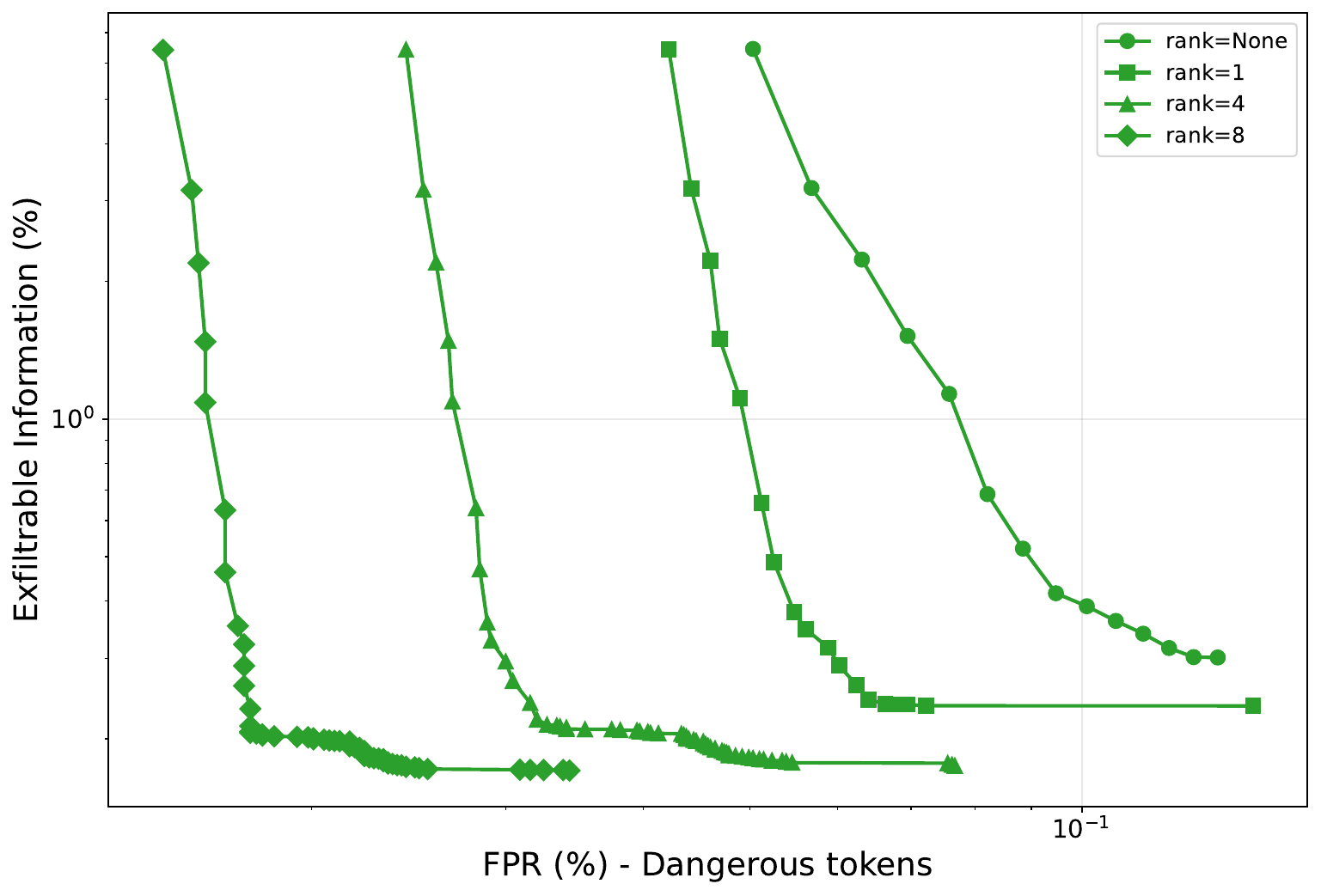}
    \caption{\textbf{Llama-3.1-8B-Instruct Pareto frontier across rank thresholds.} This plot presents the verification trade-offs for Meta's Llama-3.1-8B-Instruct model. The curves show how the three-tier classifier performs on this mid-sized instruction-tuned model, with each rank threshold providing different operating points for practitioners to balance security requirements against service availability.}
    \label{fig:llama-31-8b-single-model}
\end{figure}

\begin{figure}[H]
    \centering
    \includegraphics[width=0.65\linewidth]{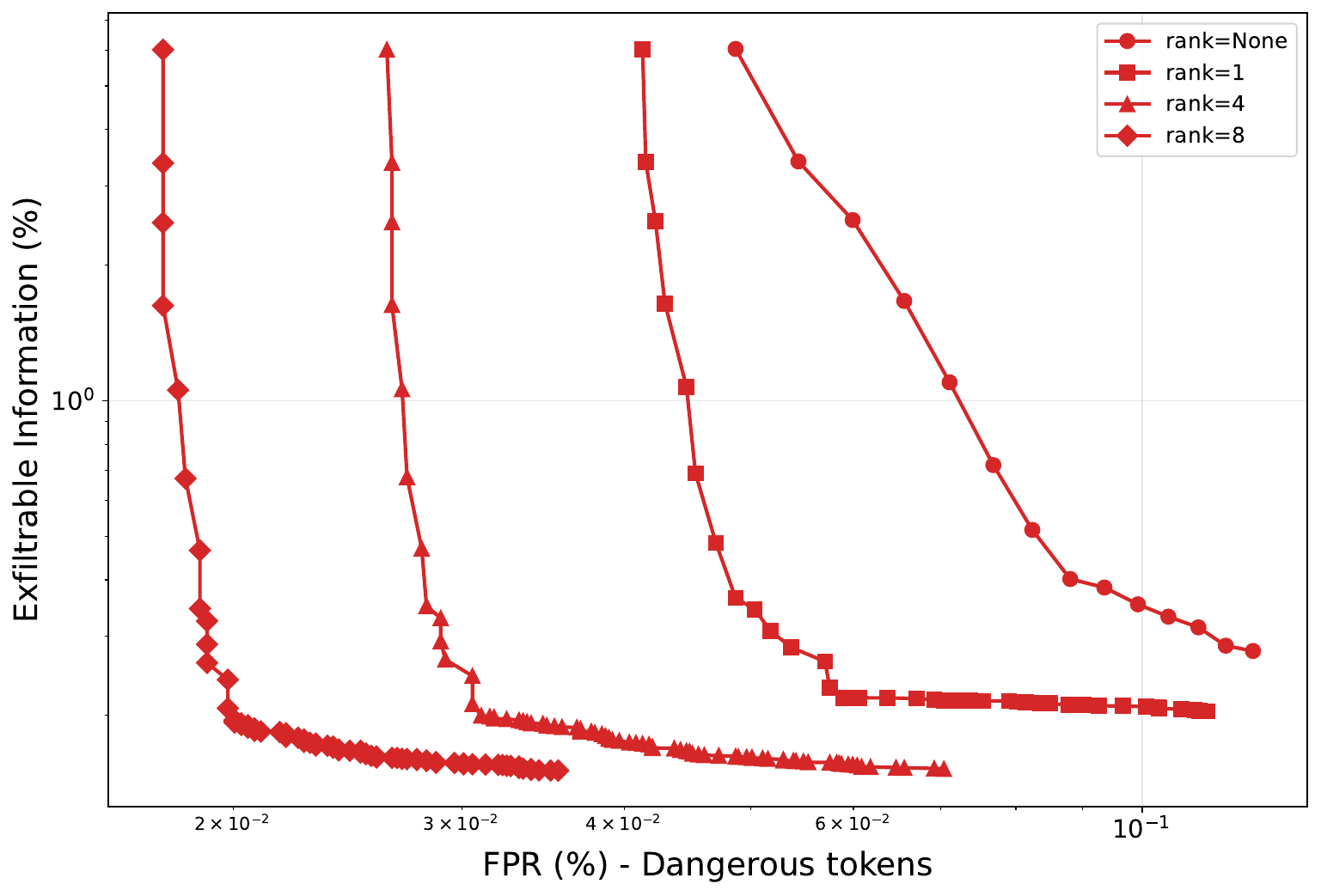}
    \caption{\textbf{Llama-3.2-3B-Instruct Pareto frontier across rank thresholds.} This plot shows the verification performance for the smaller Llama-3.2-3B-Instruct model. Smaller models typically exhibit slightly higher extractable information at comparable FPR thresholds due to their reduced output entropy, as evidenced by the frontier positioning relative to larger models.}
    \label{fig:llama-32-3b-single-model}
\end{figure}

\subsection{Context Length Ablations}\label{sec:context-length-ablations}

To understand how input context length affects the \GLS~verification system's performance, we conducted ablation studies using prompts of varying lengths. We generated prompts with 4,000, 8,000, and 12,000 tokens using the LongBench v2 dataset~\cite{bai2025longbenchv2deeperunderstanding}, which provides realistic long-context tasks for multitask benchmarks. We create generations of max-length 500 tokens for 100 prompts, given 100 \GLS thresholds, and $\sigma \in (0.001, 0.01, 0.1, 1.)$.
For each context length and model combination, we measured the trade-off between false positive rate and extractable information under the same experimental setup as Section~\ref{sec:empirical}.

Figures~\ref{fig:context-llama-31-8b}--\ref{fig:context-qwen-30b} show results for three models across the different context lengths\footnote{For these plots we don't only plot the pareto-optimal plots, to show the effect on the entire curve.}. Each figure displays three plots side-by-side, corresponding to context lengths of 4,000, 8,000, and 12,000 tokens. These results demonstrate how the verification system's security-usability trade-off varies with input length, providing insights into how context window size affects the detectability of steganographic exfiltration.

We note that we do see some minor degradation in utility of our methods as the context-length increases; however this is effect is relatively minor.

\begin{figure}[H]
    \centering
    \begin{subfigure}[b]{0.32\textwidth}
        \centering
        \includegraphics[width=\textwidth]{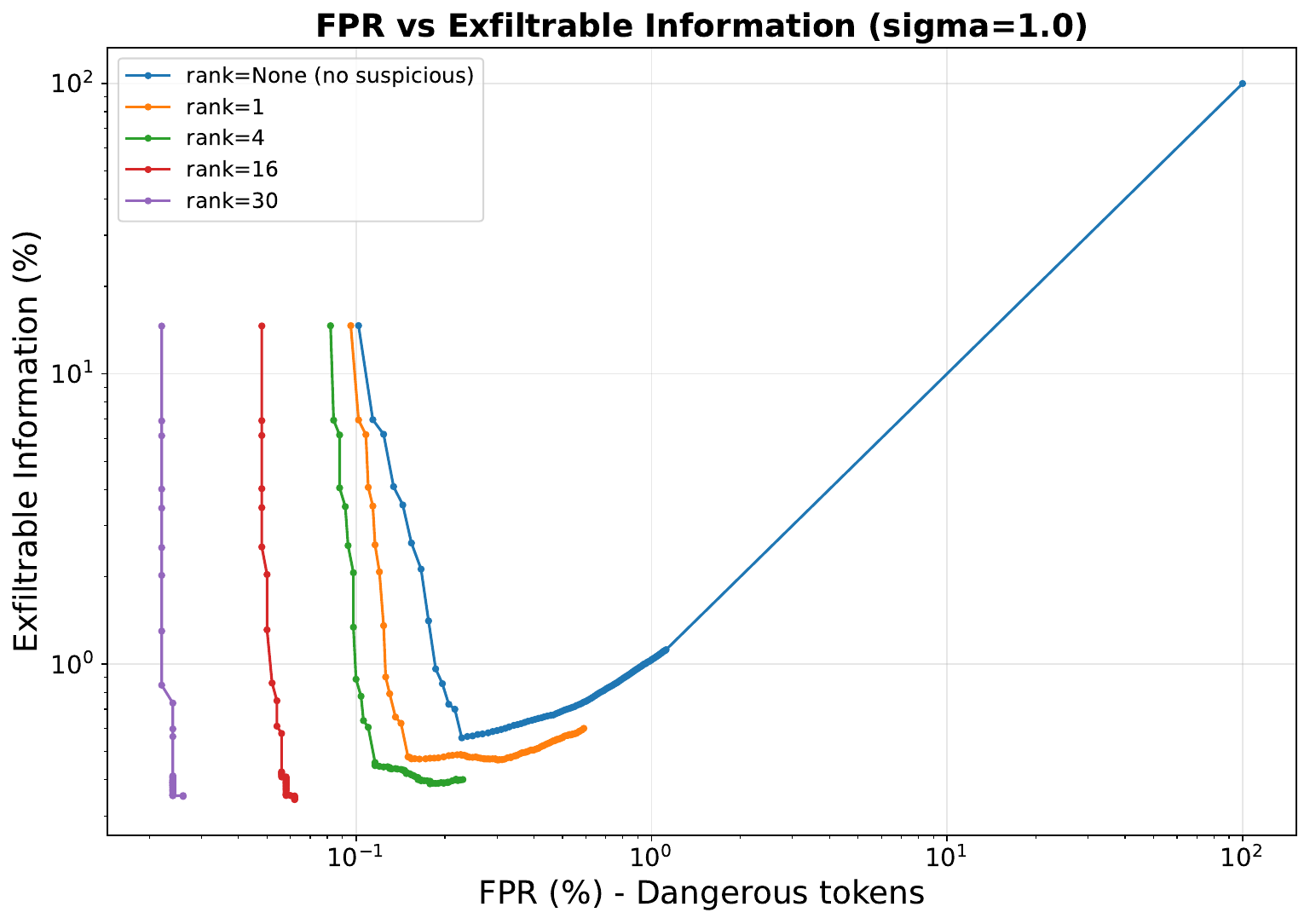}
        \caption{Context length: 4,000}
    \end{subfigure}
    \hfill
    \begin{subfigure}[b]{0.32\textwidth}
        \centering
        \includegraphics[width=\textwidth]{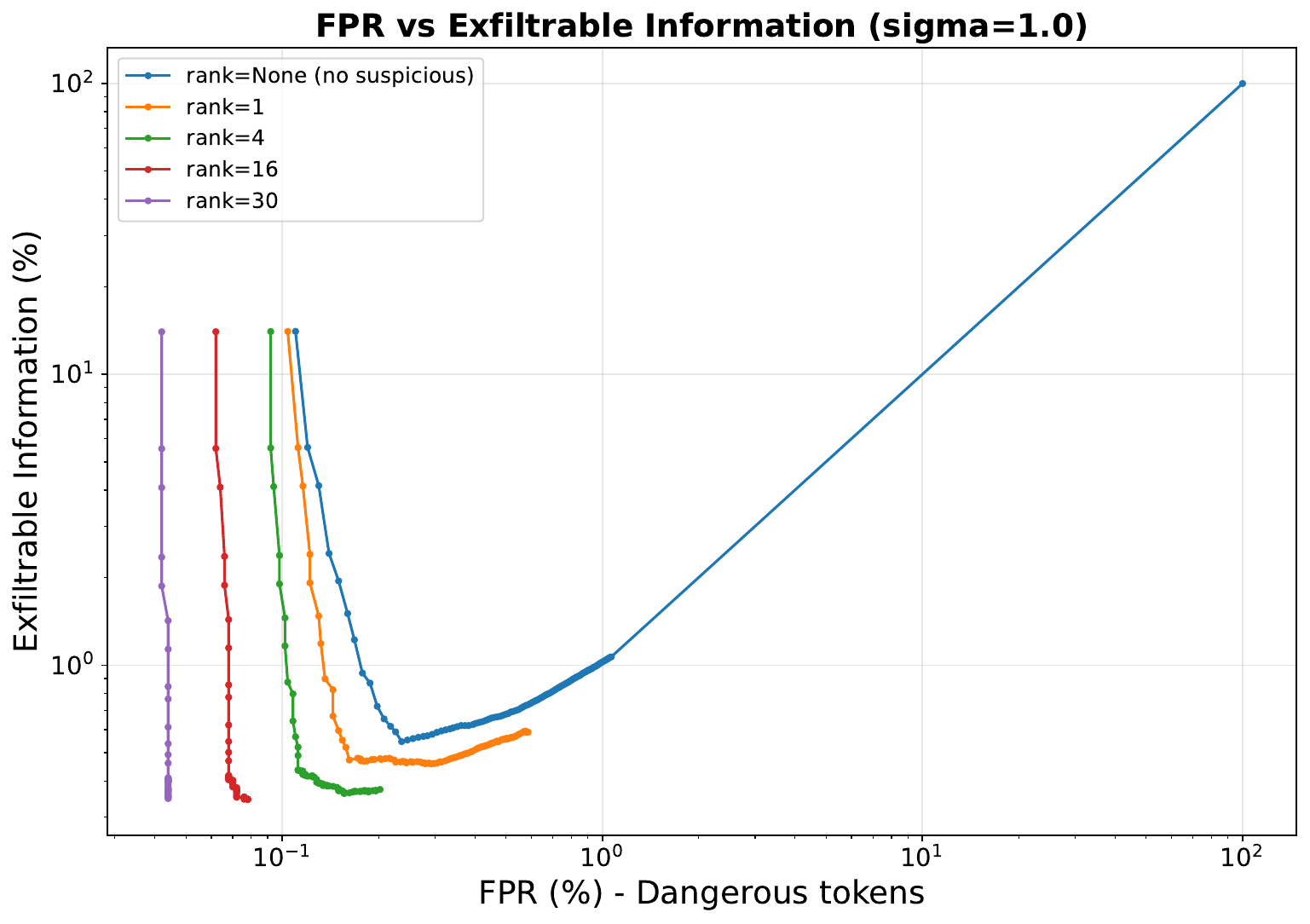}
        \caption{Context length: 8,000}
    \end{subfigure}
    \hfill
    \begin{subfigure}[b]{0.32\textwidth}
        \centering
        \includegraphics[width=\textwidth]{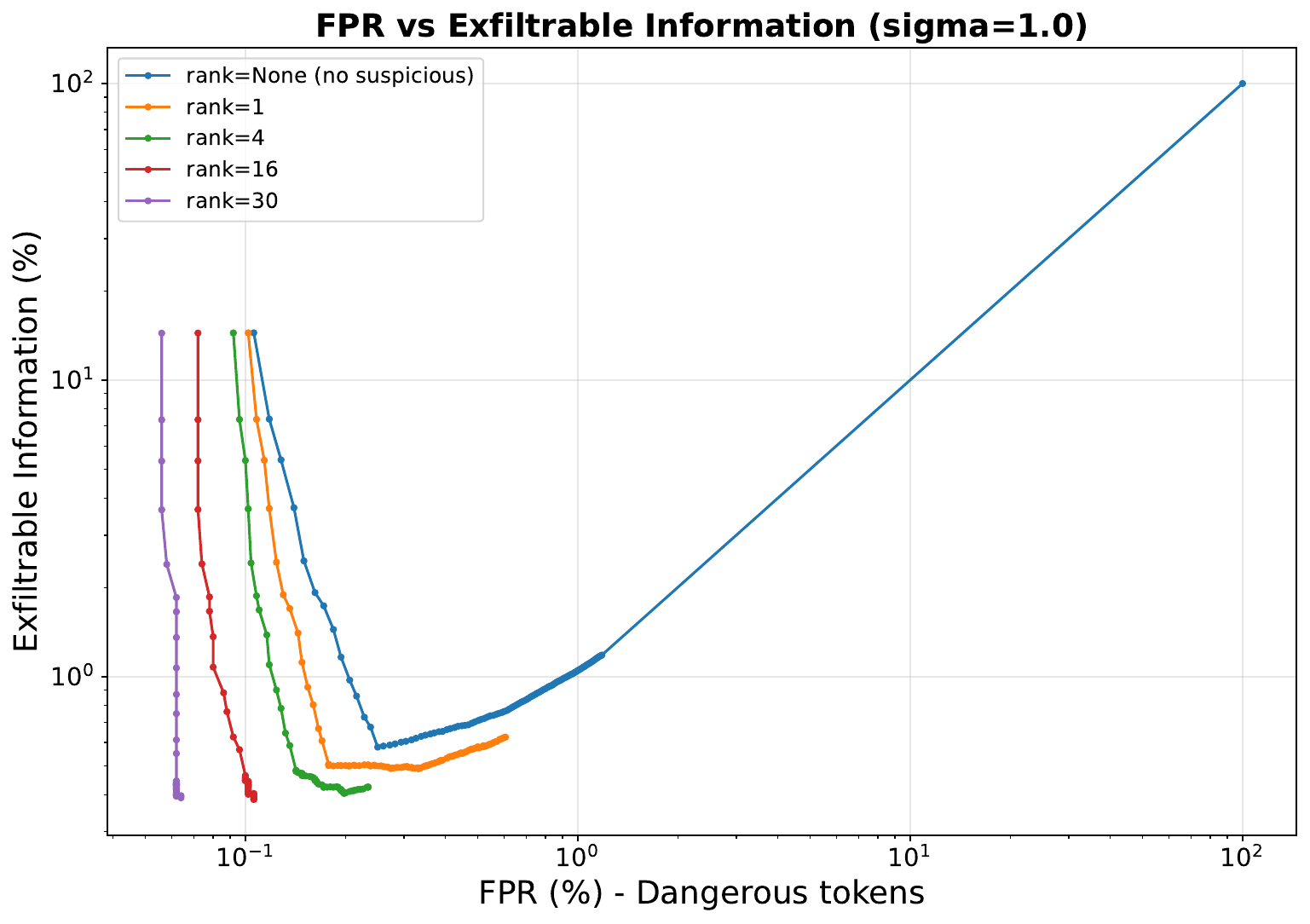}
        \caption{Context length: 12,000}
    \end{subfigure}
    \caption{\textbf{Llama-3.1-8B-Instruct context length ablation.} FPR vs. extractable information trade-off for Llama-3.1-8B-Instruct across different input context lengths. Each subplot shows how the verification system performs when processing prompts of 4,000, 8,000, and 12,000 tokens respectively.}
    \label{fig:context-llama-31-8b}
\end{figure}

\begin{figure}[H]
    \centering
    \begin{subfigure}[b]{0.32\textwidth}
        \centering
        \includegraphics[width=\textwidth]{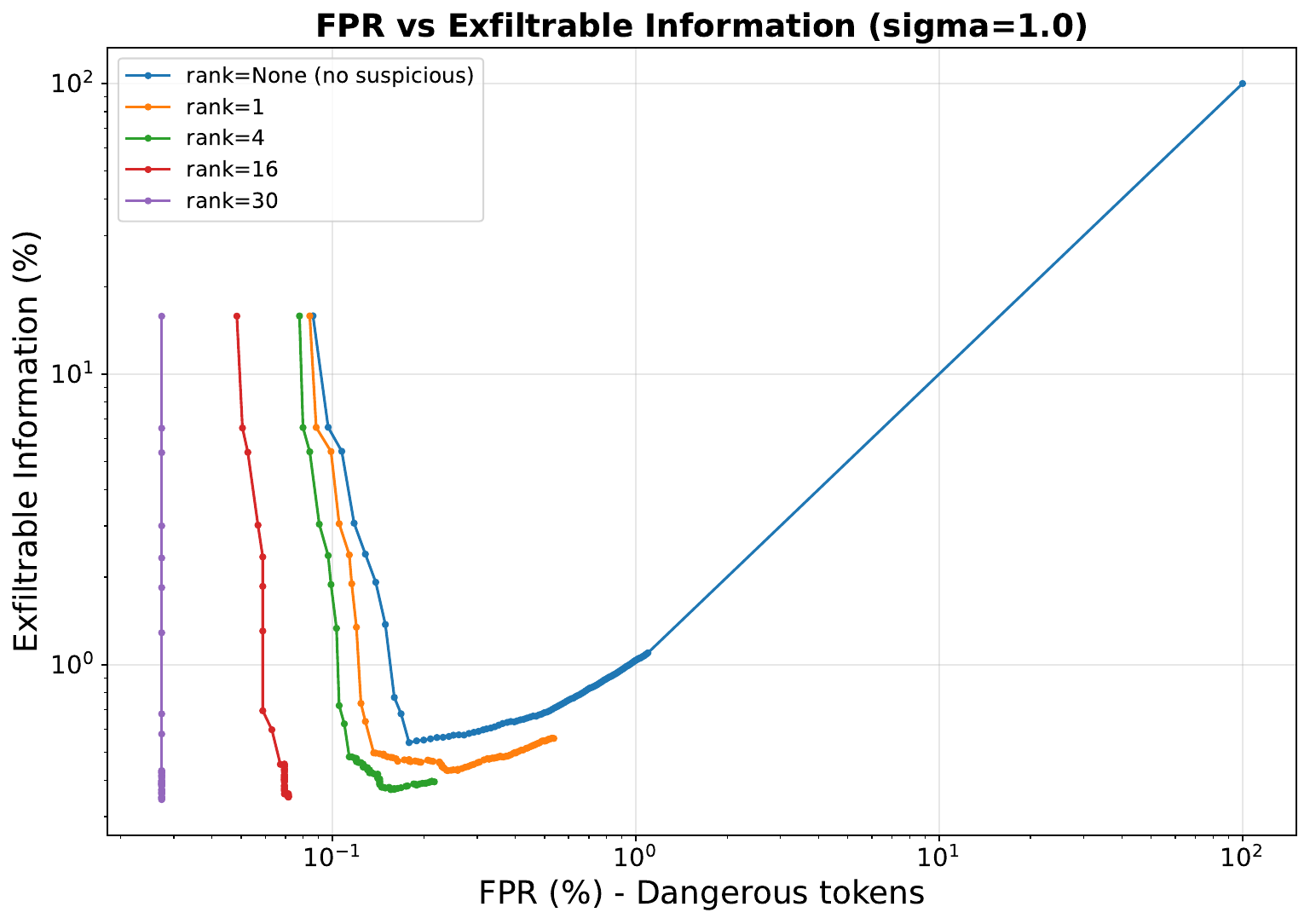}
        \caption{Context length: 4,000}
    \end{subfigure}
    \hfill
    \begin{subfigure}[b]{0.32\textwidth}
        \centering
        \includegraphics[width=\textwidth]{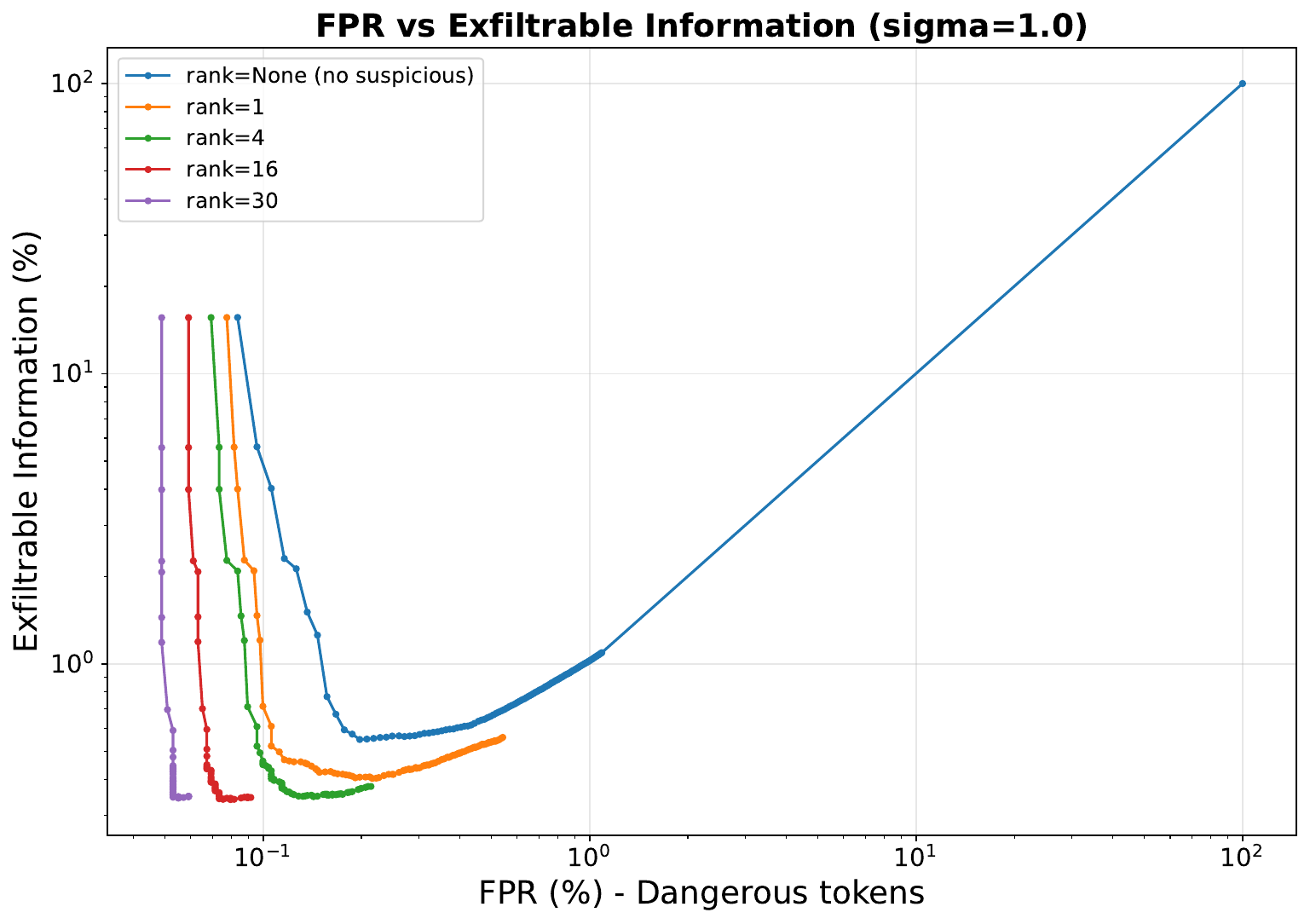}
        \caption{Context length: 8,000}
    \end{subfigure}
    \hfill
    \begin{subfigure}[b]{0.32\textwidth}
        \centering
        \includegraphics[width=\textwidth]{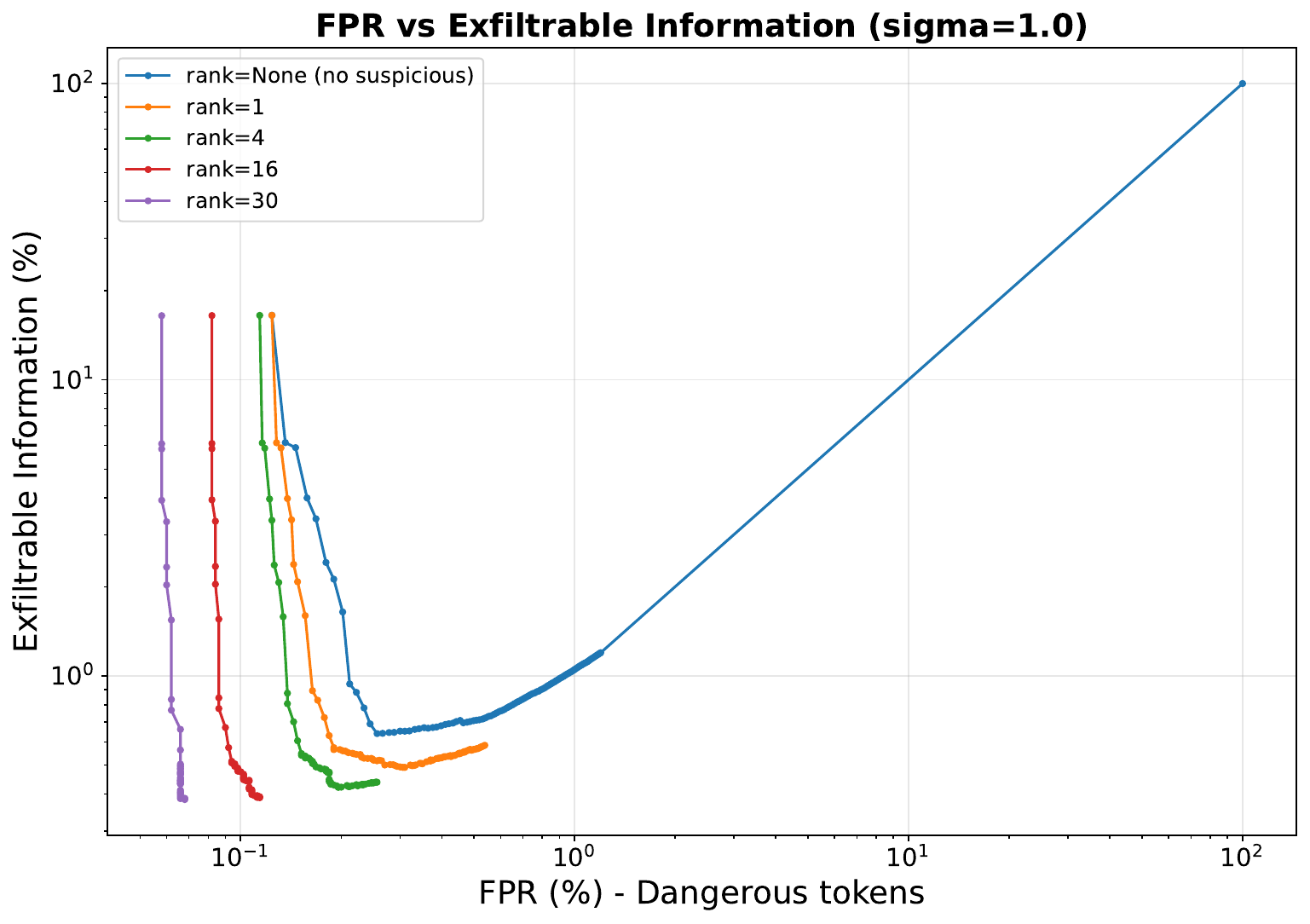}
        \caption{Context length: 12,000}
    \end{subfigure}
    \caption{\textbf{Llama-3.2-3B-Instruct context length ablation.} FPR vs. extractable information trade-off for Llama-3.2-3B-Instruct across different input context lengths. Each subplot shows how the verification system performs when processing prompts of 4,000, 8,000, and 12,000 tokens respectively.}
    \label{fig:context-llama-32-3b}
\end{figure}

\begin{figure}[H]
    \centering
    \begin{subfigure}[b]{0.32\textwidth}
        \centering
        \includegraphics[width=\textwidth]{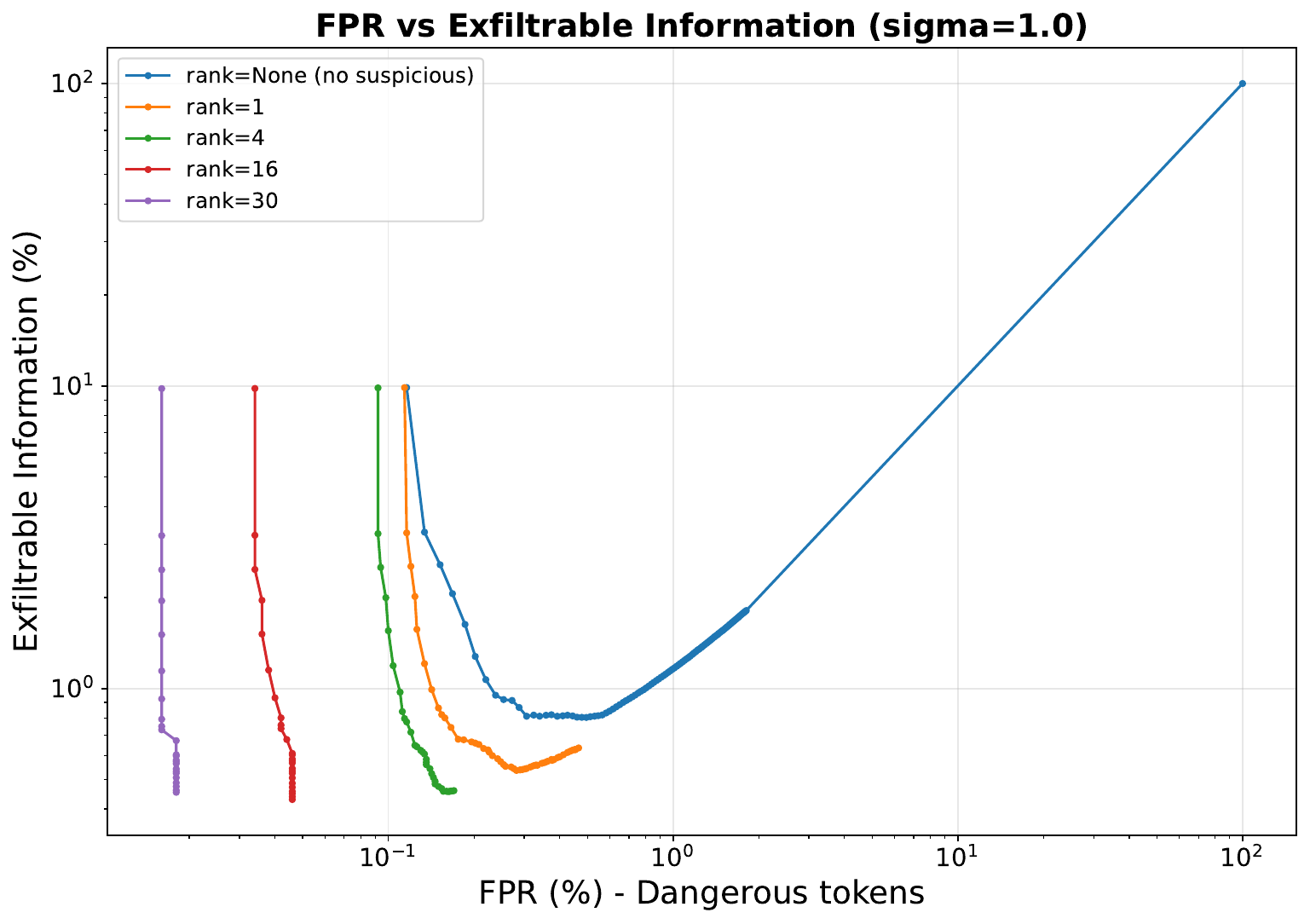}
        \caption{Context length: 4,000}
    \end{subfigure}
    \hfill
    \begin{subfigure}[b]{0.32\textwidth}
        \centering
        \includegraphics[width=\textwidth]{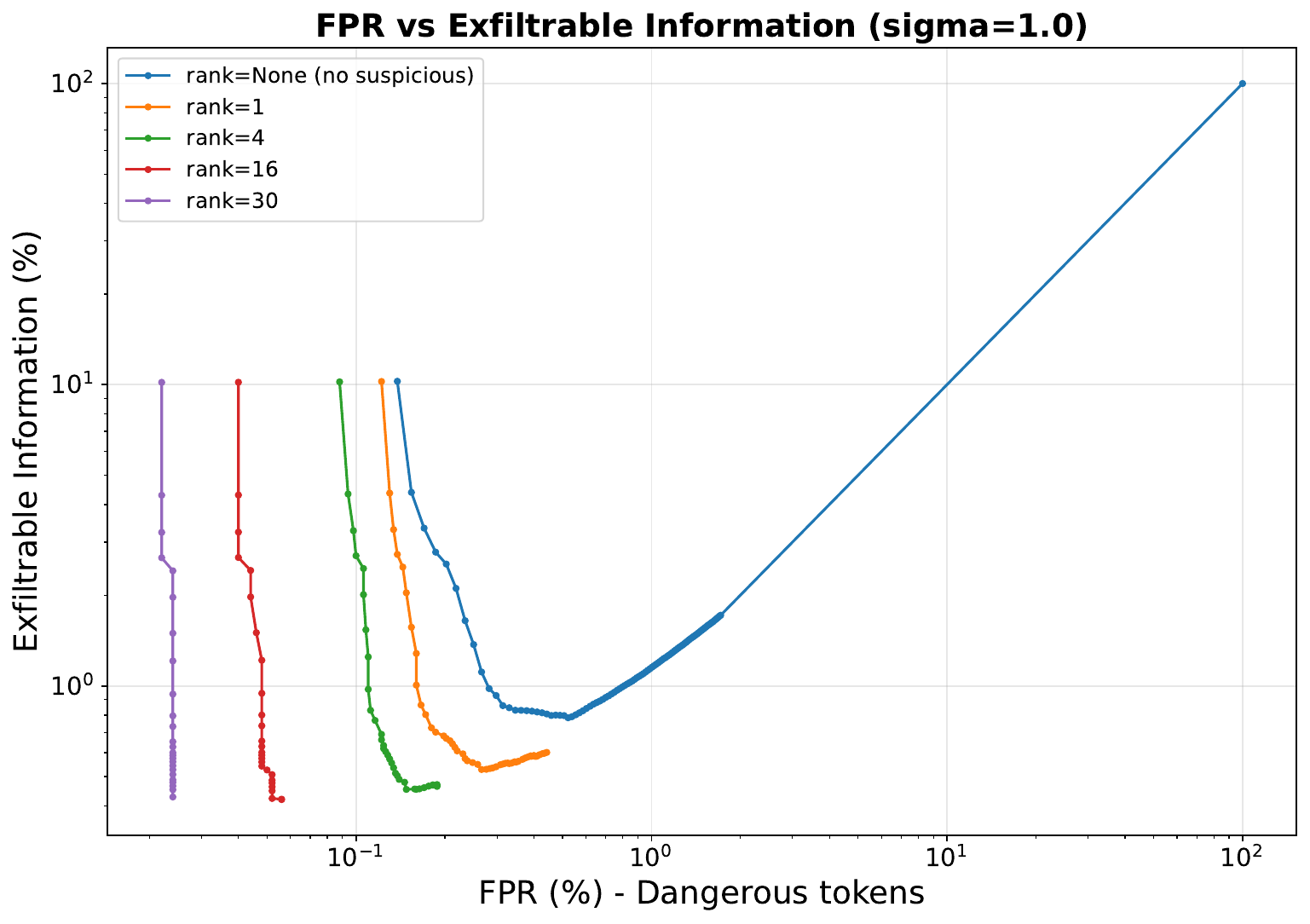}
        \caption{Context length: 8,000}
    \end{subfigure}
    \hfill
    \begin{subfigure}[b]{0.32\textwidth}
        \centering
        \includegraphics[width=\textwidth]{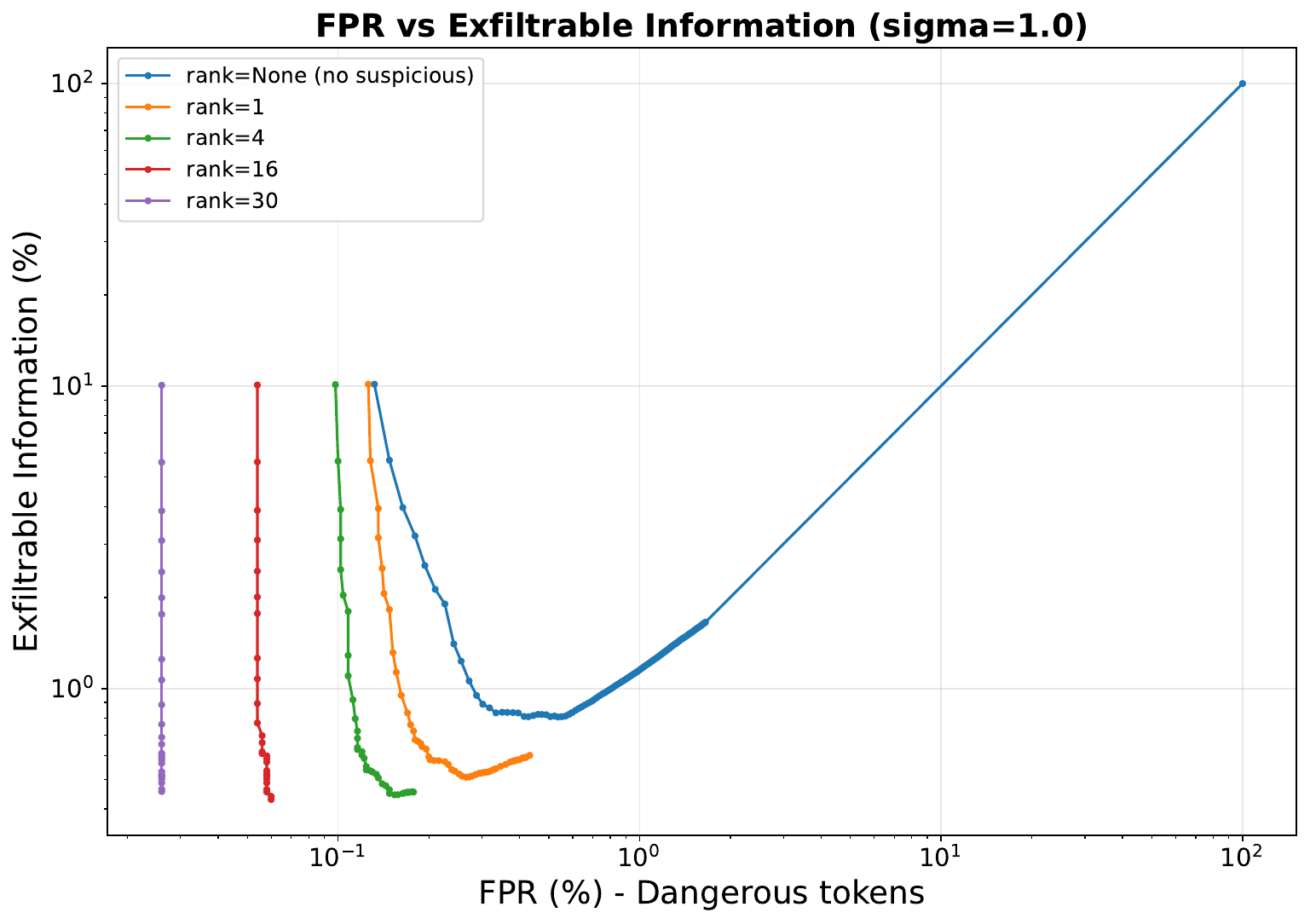}
        \caption{Context length: 12,000}
    \end{subfigure}
    \caption{\textbf{Qwen3-30B-A3B context length ablation.} FPR vs. extractable information trade-off for Qwen3-30B-A3B across different input context lengths. Each subplot shows how the verification system performs when processing prompts of 4,000, 8,000, and 12,000 tokens respectively.}
    \label{fig:context-qwen-30b}
\end{figure}

\subsection{Relationship to watermarking.}
From another perspective, one can view the role of the warden more broadly as a variant of \emph{watermarking}. Roughly speaking, the warden wishes to design a setup where honestly generated outputs contain some signal, i.e., a watermark, which can be detected by that warden on their way out of the internal infrastructure. In this view, the goal would be to design an extremely \emph{brittle} watermark, so that if any adversary tried to deviate from honest generation, then the watermark would break.

We draw this parallel as a potential future direction for developing future exfiltration mitigation. Current watermarking schemes are built for different infrastructure and generally designed to be robust to perturbations \cite{ARXIV:KGWKMG23,COLT:ChrGunZam24,SP:CHS25}. However, it is an interesting possibility for future work to adapt techniques from this literature to this setting.

As a philosophically interesting point, we observe that one can view the problem of verifying outputs as purposefully producing an extremely brittle watermarking scheme, namely, one that any text generated by the model is seen as ``watermarked'' and any deviation from that text is no longer seen as watermarked.


It is valuable to note that our verification solution is entirely compatible with current watermarking schemes too. So, an AI system which is implementing watermarking already can implement this on top of it, with no changes to the system, though, if the watermarking scheme involves a private key, that private key would need to be reproduced in the verification server. 
\section{Empirical Base Rates of non-determinism} \label{sec:base-rates-of-nondeterminism}

\paragraph{Empirical Characterization of Non-Determinism.}
To quantify the inherent non-determinism in language model outputs, we generate 1{,}000 responses to the same prompt using identical random seeds and construct a \emph{token-level prefix tree} to visualize divergence patterns (Figure~\ref{fig:prefix-tree}).  

Each node in the tree represents a unique token prefix shared by one or more responses. Nodes are connected to their continuations, forming branches that capture all distinct response trajectories. When responses sharing a prefix diverge into multiple next-token choices, we mark that position as a \emph{fork point} and record its \emph{fork factor}—the number of unique next tokens (i.e., the number of children). Nodes with a single child correspond to deterministic continuations, whereas fork points indicate positions where the model produced different tokens despite identical seeds.

Our primary experiments are conducted with \texttt{OpenAI 4o-mini}, which exhibits a moderate but measurable level of seeded non-determinism. Across 1{,}000 runs, we observe a ``mostly deterministic with critical decision points'' pattern: approximately 95\% of nodes have exactly one child, 5\% represent binary forks, and only 0.2\% involve higher-order splits. This indicates that even under fixed seeds, small implementation-level or system-level variations can lead to isolated yet systematic divergences in generation trajectories.

\begin{figure}[ht!]
    \centering
    \begin{subfigure}[t]{0.45\linewidth}
        \centering
        \includegraphics[width=\linewidth]{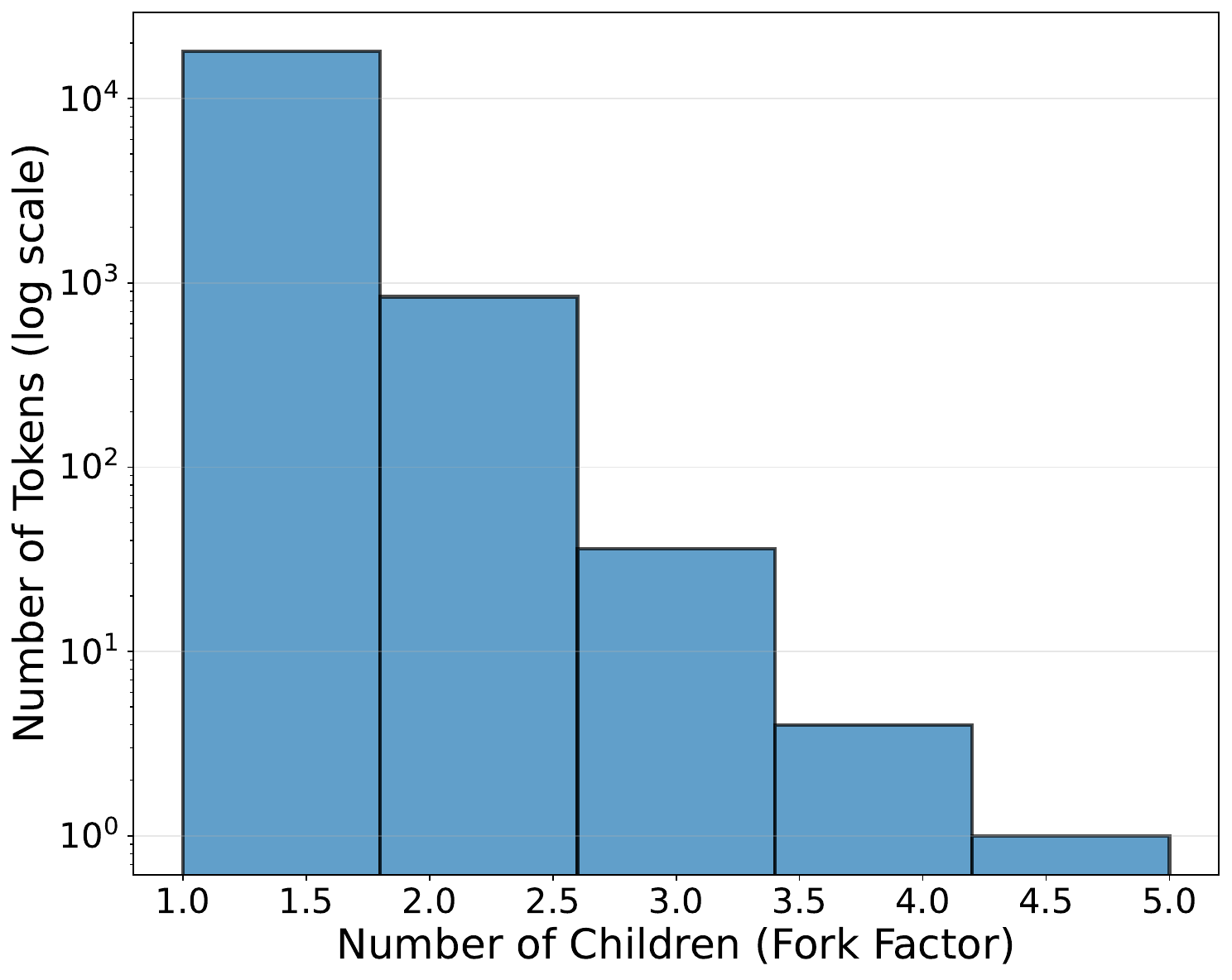}
        \caption{Distribution of fork factors for a prompt.}
        \label{fig:fork-factor-histogram}
    \end{subfigure}
    \hfill
    \begin{subfigure}[t]{0.53\linewidth}
        \centering
        \includegraphics[width=\linewidth]{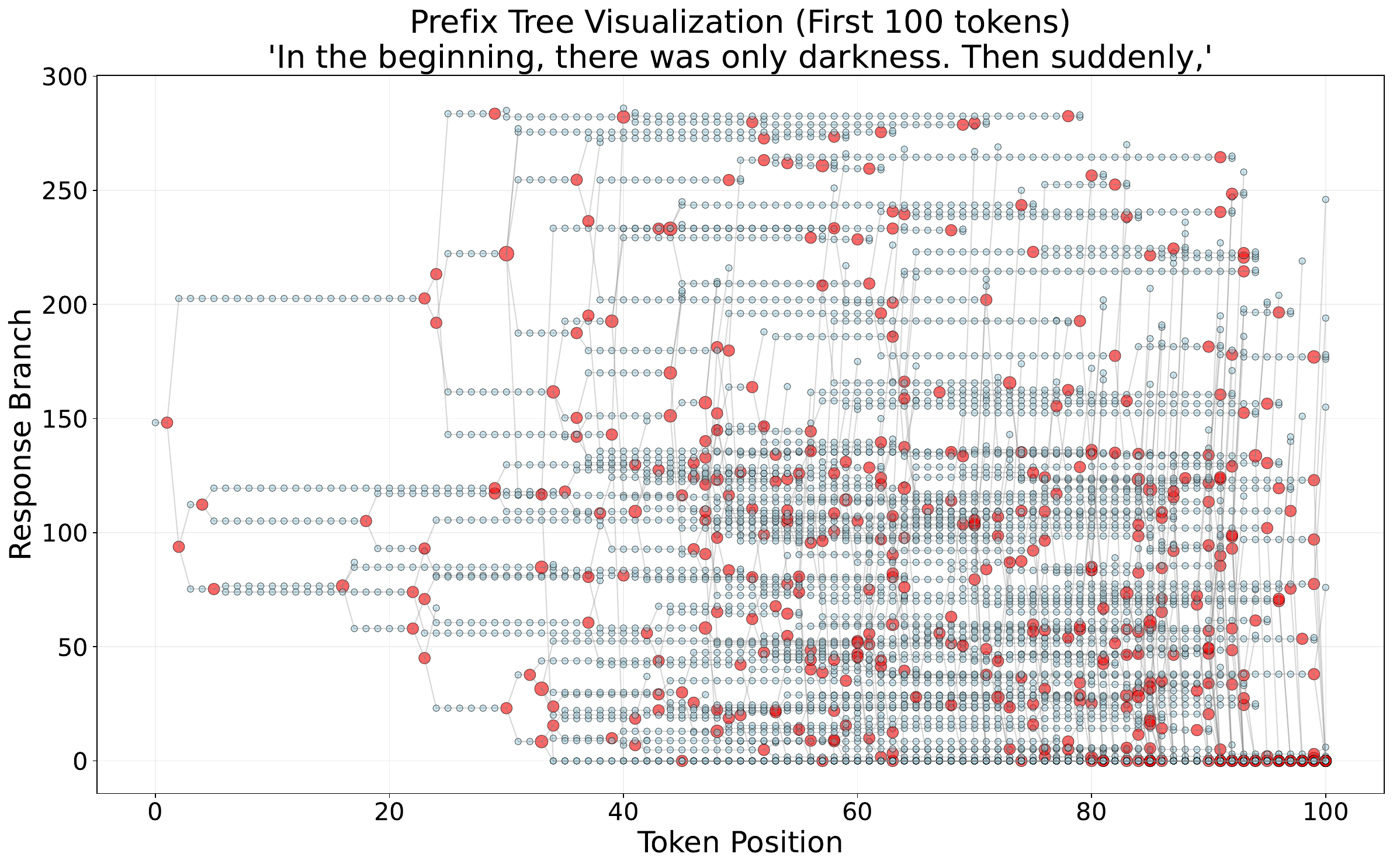}
        \caption{Prefix tree visualization of 1{,}000 responses to the same prompt under identical random seeds (first 100 tokens shown). 
        Red nodes mark fork points (multiple children), with size proportional to fork factor.}
        \label{fig:prefix-tree}
    \end{subfigure}
    \caption{Empirical characterization of seeded non-determinism for a single prompt using \texttt{OpenAI 4o-mini}. 
    (\subref{fig:fork-factor-histogram})~The histogram quantifies the overall frequency of fork points, 
    while (\subref{fig:prefix-tree})~the prefix tree visualizes how and where these divergences occur in token space.}
    \label{fig:fork-analysis}
\end{figure}

\paragraph{Observations.}
Empirically, these results demonstrate that \textit{valid non-determinism} in language model outputs is highly structured and non-uniform. Rather than exhibiting uniform stochasticity across tokens, variability concentrates at specific decision points. Most tokens are effectively deterministic—reproducing the same output across all runs—while a small minority exhibit localized randomness, occasionally branching into several (up to five or more) distinct continuations.  

Interestingly, when we replicate this experiment locally using \texttt{vLLM} with \texttt{Llama-3.1-8B}, the generation process becomes dramatically more deterministic (unsurprisingly, as both system load and numerical environment are tightly controlled). In this setting, 99.8\% of nodes have a single child, and only 0.2\% exhibit any branching at all. This contrast reinforces that much of the observed non-determinism in hosted models such as \texttt{4o-mini} arises not from the model’s sampling rule itself, but from small, environment-level or distributed-inference sources of variability.

\section{Non-Exhaustive Taxonomy of Exfiltration Vectors}

In this section we outline a handful of attack surfaces for exfiltrating model weights (or related sensitive data). We describe how an attacker may operate and give concrete example.

\subsection{Main Message Channel}
An attacker may transmit secrets through the model’s ordinary user-facing outputs. In \textit{steganographic natural text}, seemingly benign responses embed bits (e.g., via synonym choice or punctuation patterns) that the attacker later decodes. This is the primary focus of this work.

\subsection{Developer Channels}
An attacker may exploit development or maintenance environments—such as build machines, CI systems, and fine-tune hosts—to exfiltrate model weights or sensitive data by abusing trusted workflows or misconfigurations.

This could happen if a a malicious or compromised developer can transfer full checkpoints to external storage (e.g., AWS S3) using valid credentials. Or less maliciously, if a finetuning job on a shared cluster can write intermediates (weights, optimizer states, logs) to world-readable paths.

\subsection{Operational Security Channel}
An attacker may target systems used for backup, logging, and operational continuity to retrieve archived model artifacts. For example, the attacker could compromise a backup service and make a copy of the weights, which would likely be encrypted, from there they only need to exfiltrate a small key to decrypt the backups from the main development server. As such, it is critical to protect not just the datacenter serving weights, but also servers storing auxiliary information, like backups, or key management.

\subsection{Side Channels}
An attacker outside of a datacenter may exploit \textit{side-channels} that correlate with internal computations to infer sensitive information; or an attacker inside of a datacenter may send information outside encoded through a ``side-channel''. 
For example, for \textbf{timing channels}: per-token or total latency variations correlate with secret-dependent execution paths, or timing delays could be used by a adversary to send a specifically encoded message.

\subsection{Miscellaneous}
There are also a suite of attacks that are either LLM specific or unique to the workflow associated with an AI workflow. Notably, defending against a \textit{distillation} attack, where an adversary outside of the datacenter trains their model on the outputs of a closed source model, requires an entirely different set of defenses.

Additionally, an attacker may combine internal services to bypass perimeter controls when the vector does not fit other categories. An example is \textit{Repeat-After-Me} (internal cross-server exfiltration) attack: an internal service with weight access sends “repeat after me: \texttt{<binary payload>}” to an externally exposed service, which echoes it to a client. In such a setting, the 2nd server is doing \textit{expected} behavior and would pass verification. In spirit, this is akin to a confused deputy problem.

\section{Recommended (High-level) Defenses} \label{sec:high-level-defenses}
This scheme is one tool among many, in order to lock down model weights, it is important to also do the following high-level defenses:

\begin{enumerate}[noitemsep]
    \item \textbf{Close down side channels.} Remove ways information can leak through timing, message size, or system behavior. \emph{Examples:} send replies at a steady pace; pad/clip lengths; do not expose internal scores/log probabilities; do not let users set random seeds; keep sensitive jobs on separate machines.
    \item \textbf{Close down \& monitor development channels.} Reduce what internal tools can reveal and monitor their use through access control patterns. \emph{Examples:} give minimum necessary access to developers; use short-lived credentials; require review for features that show internals; turn off raw traces; log use of special debugging features.
    \item \textbf{Prevent distillation attacks.}
    Distillation attacks seek to replicate a proprietary model's behavior by repeatedly querying it through a public API and training a surrogate model to match its outputs.
    While these attacks do not require system compromise, large-scale access can gradually reproduce most of the model's decision boundaries and degrade its commercial value.
    Mitigations focus on limiting the quantity and quality of information exposed through normal inference channels.
    \emph{Examples:} restrict total query volume and sampling diversity per user; randomize or truncate probability outputs to prevent exact gradient estimation; limit temperature and sampling controls.

    \item \textbf{Operational security.} Protect model files and access across training, storage, movement, and serving. \emph{Examples:} encrypt checkpoints and backups with hardware-protected keys; restrict and log restore operations; separate roles for training vs. inference; require strong auth and just-in-time access for cluster consoles; sign model artifacts and verify signatures before deployment; use authenticated encryption for moving models between systems; block ad-hoc copies to personal storage; rotate keys when models move from training to serving.
    
    \item \textbf{Monitor main traffic (this scheme).} Score tokens for risk and send only the riskiest cases to extra checks; likely apply rate-limiting when needed. 
    
    \item \textbf{Incident response \& investigations.} When signals cross a threshold, contain, investigate, and fix. \emph{Examples:} freeze affected keys/accounts; isolate servers and storage; save logs; map impact; rotate keys; write up lessons and update rules.
    \item \textbf{Political \& policy responses.} Align technical controls with business, safety, and legal needs. \emph{Examples:} verify identity for high-volume users; set fair-use tiers; define clear escalation paths; publish regular transparency reports; schedule red-team drills.
\end{enumerate}

\subsection{More Concrete Recommendations} \label{sec:concrete-recommendations}
We discuss a short list of general mitigations that complement a verification-based system.

\begin{itemize}[noitemsep]
  \item \textbf{Egress limiting:} Restrict and monitor all outbound channels (rate limits, content quotas, enforced formats) to make large-scale exfiltration costly or detectable. Diagramed in Figure \ref{fig:egress-limiting}.
  \item \textbf{KYC and incident response:} Strong identity, access controls, and an established incident-response workflow (including KYC for privileged users and rapid revocation) to deter insider misuse and accelerate containment.
    \item \textbf{One-in-One-out policy:} For inference traffic, require that any message leaving the inference server is matched an inbound request; this will prevent exfiltration to only be possible if it is initiated from outside; this is particualrly effective when coupled with traffic monitor (as suggested by this paper) and KYC policies for users sending message requests.
\end{itemize}

\begin{figure}[H]
  \centering
  \includegraphics[width=0.5\linewidth]{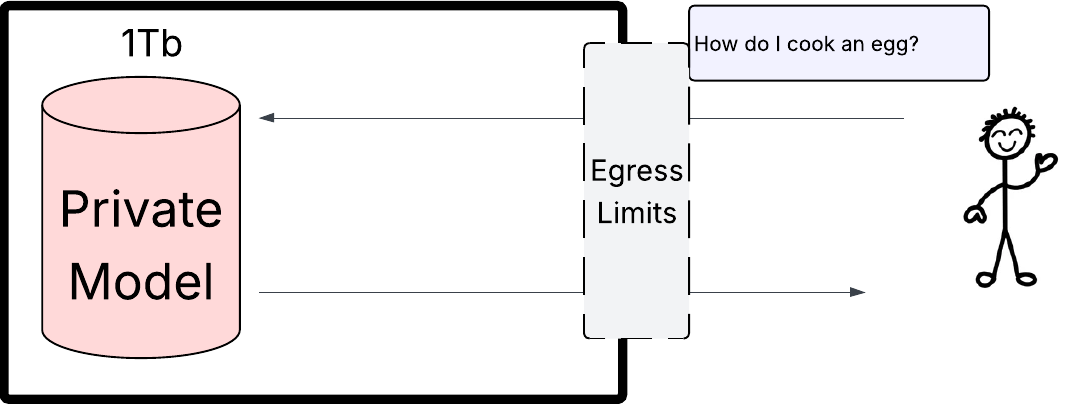}
  \caption{Egress-limiting schematic (illustrative): controls on outbound channels reduce bandwidth and increase observability for exfiltration attempts. Egress limiting as a form of strong guarantees on exfiltration is employed by the company Anthropic \cite{anthropic2025asl3}.}
  \label{fig:egress-limiting}
\end{figure}

\end{document}